\documentclass{article}

\usepackage[a4paper, margin=1in]{geometry}

\usepackage[linesnumbered,ruled,noend]{algorithm2e}
\usepackage{tikz}

\let\theoremstyle\relax

\usepackage{amsthm}
\usepackage{enumitem}
\usepackage{multirow}
\usepackage{longtable}
\usepackage{xltabular}
\usepackage{mathtools}
\usepackage{amsmath,lipsum,amssymb}
\usepackage{bm}
\usepackage{authblk}
\usepackage{anyfontsize}
\usepackage{natbib}
 \bibpunct[, ]{(}{)}{,}{a}{}{,}%

\newtheorem{theorem}{Theorem}     

\theoremstyle{definition}

\theoremstyle{remark}

\newcommand{\keywords}[1]{\par\noindent{\textbf{Keywords:} #1}\par}

\newtheorem{lemma}[theorem]{\bf R}

\begin{document}




\title{A Reduction-Driven Local Search for the Generalized Independent Set Problem}

\author[1]{Yiping Liu}
\author[1]{Yi Zhou\thanks{Corresponding author: \texttt{zhou.yi@uestc.edu.cn}}}
\author[1]{Zhenxiang Xu}
\author[1]{Mingyu Xiao}
\author[2]{Jin-Kao Hao}

\affil[1]{School of Computer Science and Engineering, University of Electronic Science and Technology of China, Chengdu, China \\
\texttt{yliu823@aucklanduni.ac.nz}, \texttt{zhenxiangxu@std.uestc.edu.cn}, \texttt{myxiao@uestc.edu.cn}}

\affil[2]{LERIA, Université d’Angers, 2 Boulevard Lavoisier, 49045 Angers, France \\
\texttt{jin-kao.hao@univ-angers.fr}}

\date{}






\maketitle

\begin{abstract}
The Generalized Independent Set (GIS) problem extends the classical maximum independent set problem by incorporating profits for vertices and penalties for edges. 
This generalized problem has been identified in diverse applications in fields such as { forest harvesting}, competitive facility location, social network analysis, and even machine learning.
However, solving the GIS problem in large-scale, real-world networks remains computationally challenging. 
In this paper, we explore data reduction techniques to address this challenge.
We first propose 14 reduction rules that can reduce the input graph with rigorous optimality guarantees.  
We then present a reduction-driven local search (RLS) algorithm that integrates these reduction rules into the pre-processing, the initial solution generation, and the local search components in a computationally efficient way.
The RLS is empirically evaluated on 278 graphs drawn from different application scenarios. 
The results indicate that the RLS is highly competitive -- 
For most graphs, it achieves significantly superior solutions compared to other known solvers, and it effectively provides solutions for graphs exceeding 260 million edges, a task at which every other known method fails.
Analysis also reveals that data reduction plays a key role in achieving  such a competitive performance.
\end{abstract}

\section{Introduction}
\label{introduction}
Graphs are essential tools for modeling complex systems including social networks, geometric networks, and forest management systems. The analysis of graph data has led to the development of numerous combinatorial problems that help address real-world challenges.
In this paper, we study the {\textit{Generalized Independent Set}} (GIS) problem.
Let $G=(V,E=E_p\cup E_r)$ be an undirected graph where $V$ is the vertex set and $E$ is the edge set that is partitioned into a \textit{permanent edge} set $E_p$ and a \textit{removable edge} set $E_r$.
Each vertex $v\in V$ is assigned a profit $w(v)\in \mathbb{R}$.
Each removable edge $e\in E_r$ is associated with a penalty $p(e)\in \mathbb{R}$.
A \textit{generalized independent set} of the graph $G$ is a vertex set $I\subseteq V$ such that no permanent edge exists between any two vertices of $I$ in $G$.
Then, the GIS problem asks for a generalized independent set $I$ of maximum net benefit, which is the total vertex profits of $I$ minus the total penalties of removable edges whose end vertices are in $I$.
An example is shown in Figure~\ref{fig: facility}.

\begin{figure}[h]
    \centering
    \caption{ A toy example of the GIS problem. The set of red vertices is a generalized independent set of  net benefit of $(2+6+5)-1=12$. 
    }
    \vspace{0.3cm}
    \includegraphics[width=0.5\linewidth]{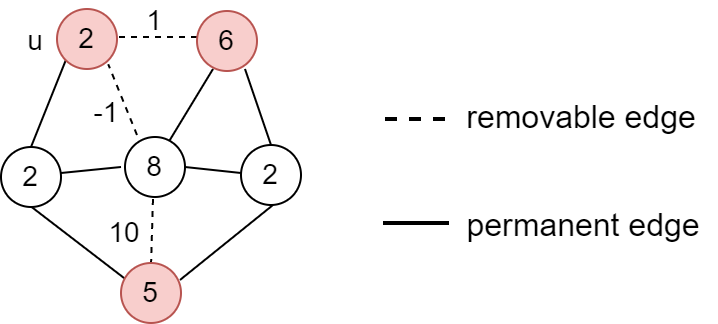}
    \label{fig: facility}
\end{figure}

{ When all the edges are permanent and each vertex has a profit of 1, the GIS problem becomes the well-known \textit{maximum independent set problem} (MIS), which asks for a maximum{-cardinality} set of non-adjacent vertices {in} a given graph.} 
As the GIS problem allows a more flexible definition of vertex profit and edge penalties, it captures a broader range of real-world scenarios than the classical independent set problem.
We {give} detailed examples of { applications} in diverse domains in Section~\ref{sec:applications}.
However, since MIS is NP-hard, W[1]-hard and difficult to approximate \citep{cygan2015parameterized}, the equivalence above indicates that the GIS problem is also computationally hard.

Existing works have proposed various approaches to solve the GIS problem practically. Exact solution methods include mixed integer linear programming \citep{colombi2017generalized} and unconstrained binary quadratic programming \citep{hosseinian2019algorithms}.
They can solve instances with up to 400 vertices and 71,820 edges in three hours so far \citep{zheng2024exact}.
On the other hand, heuristic search methods have demonstrated their effectiveness in providing high-quality solutions within reasonable computational time (e.g., \cite{colombi2017generalized,nogueira2021iterated}). 
For example, the { most recent} algorithm, named {\it  Adaptive Local Search} (ALS), proposed by \cite{zheng2024exact} achieves near-optimal solutions on all standard benchmark instances and scales to problems with up to 18,000 vertices.
Nevertheless, we observed that these heuristic algorithms struggle to maintain solution quality in even larger instances.

Data reduction, which reduces input size while preserving problem optimality, has recently emerged as a promising technique for improving algorithm scalability on large instances. The technique has been successfully applied to problems such as independent set, vertex cover, and maximum clique \citep{grossmann2023finding,abu2022recent,verma2015solving,walteros2020maximum,lamm2019exactly}.
For the GIS problem, research on data reduction is still at its early stages.  
There are three main challenges in developing data reduction techniques for the GIS problem:
(1) The complex interplay among vertex profits, edge penalties, and the underlying graph structure,
(2) Translating reduction rules into efficient, practical algorithms, which often requires {incrementally applying the reduction rules} \citep{akiba2016branch,chang2017computing},
(3) Integrating reductions into local search algorithms while balancing search time and efficiency.

In this paper, we address these challenges by proposing a novel reduction-driven local search algorithm for solving the GIS problem.
Our contributions are summarized as follows.
\begin{enumerate}[leftmargin=*]
    \item To the best of our knowledge, this is the first work to apply data reduction techniques to the GIS problem. We introduce 14 reduction rules that are classified into four groups: edge transformation, neighborhood-based reduction, low-degree reduction, and vertex-pair reduction. 
    We provide rigorous proofs that these rules preserve the problem optimality. 
    To efficiently implement these reduction rules, we introduce an ordering strategy and construct a dependency graph to manage their applications. 

    \item We propose an efficient reduction-driven local search algorithm, named RLS, to solve the GIS problem. We incorporate our proposed reduction rules into three main components of the RLS: the pre-processing, which incrementally reduces the size of the original input instance, 
    the initial solution generation, which interplays the vertex selection and data reduction, and the neighborhood search, which uses a novel reduction-based move operator to guide the search trajectory. 

    \item We evaluate the proposed the RLS algorithm on a total of 278 instances.
   On small instances { with up to 400 vertices and 80,000 edges}, the RLS consistently achieves the known optimal solutions with  significantly shorter running time than exact algorithms. 
   On large instances, the RLS exhibits superior scalability over existing heuristic algorithms in solution quality. 
   Notably, the RLS successfully handles instances with { over 18 million vertices and} over 260 million edges, where existing heuristic methods fail.
    Our experimental results also demonstrate that data reduction techniques significantly improve the  computational efficiency of the RLS.
\end{enumerate}

\noindent
The rest of the paper is organized as follows: Section~\ref{sec: background} reviews applications and related works, followed by Section~\ref{sec: preliminary} which introduces the notations and definitions.
In Section~\ref{sec: reduction}, we present our reduction rules for the GIS problem. 
Then, in Section~\ref{sec: local search}, we describe the RLS algorithm. 
Section~\ref{sec: experiments} presents computational results and comparisons with existing methods. 
In the last section, we conclude the paper and outline future directions.

\section{Background} \label{sec: background}
In this section, we provide some application examples of the GIS model and review existing GIS algorithms.
\subsection{Applications Related to GIS}\label{sec:applications}

In the following, we highlight several important applications of the GIS problem:

\noindent {\bf Forest harvesting.} 
The planning of forest {harvesting} involves selecting forest areas (cells) for timber, aiming to maximize benefits like timber value while minimizing ecological damage from excessive cutting. As a GIS problem, forest cells are vertices weighted by harvesting benefits, and edges represent penalties for harvesting adjacent cells (\cite{hochbaum1997forest}).

\noindent {\bf Competitive facility location.} 
In determining optimal locations for service facilities, each potential site has a utility value. However, the utility diminishes when multiple facilities compete for customers in the same area. 
The facility location problem, modeled with vertices as facilities and edges as geographic relationships, can be framed as a GIS problem to maximize utilities (\cite{hochbaum200450th}).
Similar spatial conflicts exist in airplane seating \citep{pavlik2021airplane} and label placement \citep{mauri2010new}.

\noindent{\bf Social network analysis.} { Potentially}, the GIS problem has broad applications in social network analysis. A maximum independent set corresponds to the complement of a minimum vertex cover and is equivalent to a maximum clique in the graph’s complement (\cite{sanchis1996some}). These properties make GIS useful for tasks such as detecting cliques or communities (\cite{gschwind2021branch}) and optimizing coverage and reach within social networks (\cite{puthal2015efficient}).


\noindent {\bf Other applications.} GIS has been applied in machine learning to aggregate low-quality results into better results \citep{brendel2010segmentation}, to track objects in multi-object tracking tasks \citep{brendel2011multiobject}, and to serve as a pooling layer in graph neural networks \citep{nouranizadeh2021maximum}.

\subsection{Existing Algorithms} \label{sec: related work}

The GIS problem was first introduced in \cite{hochbaum1997forest}.
Existing methods for solving the GIS problem can be mainly divided into exact and heuristic algorithms. 

\noindent {\bf Exact Algorithms.}
In pursuit of optimal solutions, \cite{hochbaum200450th} introduced an integer linear programming model that assigns variables to each node and edge, and associates each edge with a specific constraint. { As a result, the integer linear program for a graph of 2000 vertices (e.g. C2000.5 in the DIMACS competition set) has approximately 500,000 variables and 500,000 constraints, which is quite large for modern ILP solvers like CPLEX. (It fails to even report a solution within 3 hours for this instance.)}
Then, \cite{colombi2017generalized} reformulated the GIS problem as an unconstrained binary quadratic program (UBQP), and developed a branch-and-cut algorithm, along with additional valid inequalities.
Note that they also generated a total of 216 random instances which often serve as benchmark datasets in the later work.
\cite{hosseinian2019algorithms} employed another { quadratic} program to bound the optimal solution, which was later incorporated into a branch-and-bound framework.
Recently, \cite{zheng2024exact} used a Lagrangian relaxation in their branch-and-bound framework to estimate tight upper bounds.
Their experiments showed that their algorithm is able to solve the GIS instances with up to 400 vertices within $3$ hours. 

\noindent {\bf Heuristic Algorithms.}
\cite{kochenberger2007effective} also developed a UBQP formulation of the GIS problem and solved it using tabu search. 
\cite{colombi2017generalized} enhanced performance by integrating advanced tabu search with UBQP.
\cite{nogueira2021iterated} directly solved the original problem by using an iterated local search heuristic. They explored move operations of adding one or two vertices to the incumbent solution.
More recently, \cite{zheng2024exact} proposed a more refined local search heuristic``AL'' to explores three move operations -- adding one vertex to the solution, dropping one vertex from the solution, and replacing a vertex in the solution with another one not in the solution. 
{ The ALS algorithm has been reported to achieve state-of-the-art performance on the standard benchmark set of 216 instances proposed by \cite{colombi2017generalized}.
In particular, it is able to find high-quality solutions for instances with up to $17,903$ vertices within $5$ minutes. 
However, its performance on graphs with more than 20,000 vertices has not been systematically evaluated. Our experiments show that ALS quickly becomes impractical on such large-scale instances.}


\noindent {\bf Data Reduction Techniques.}
{ Data reduction techniques for integer programming models can be traced back several decades, referred to as variable fixing and coefficient reduction \citep{crowder1983solving}. Then, the parameterized complexity community has also adopted this technique, known as \textit{kernelization}, to achieve strong theoretical bounds on the time complexity of many hard combinatorial graph problems \citep{cygan2015parameterized}.
}
Recently, data reduction techniques are intensively used for practically tackling MIS \citep{abu2022recent}. 
For example, \cite{akiba2016branch} shows that branch-and-reduce outperforms pure branch-and-bound, while \cite{dahlum2016accelerating} and \cite{chang2017computing} { integrate reductions with heuristics} to achieve near-optimal performance. 
These techniques enable handling graphs with millions of vertices and edges \citep{lamm2019exactly,xiao2021efficient}.
Besides, the effectiveness of reductions has also been acknowledged in other fundamental problems like the clique problem \citep{walteros2020maximum}, and the matching problem \citep{koana2021data}. 
However, as for the GIS problem, no existing work focuses on reduction rules. 
{ Extending existing reductions for the independent set problem to the GIS problem, where both vertices and edges are weighted, remains challenging.}




\section{Preliminary} \label{sec: preliminary}


 According to the definition, an input instance of the GIS problem is an (undirected weighted) graph $G=(V, E=E_p\cup E_r)$ where each vertex $v\in V$ is associated with a profit $w(v)\in \mathbb{R}$ and each removable edge $e\in E_r$ is associated with a penalty { $p(e)\in \mathbb{R}$}.
For a vertex set $S\subseteq V$, we define $nb(S)$ as the {\em net benefit of $S$}, i.e., 
\[
nb(S)= \sum_{v\in S} w(v)-\sum_{u\in S,v\in S, e=\{u,v\}\in E_r} p(e).
\]
{Then, a \textit{generalized independent set} of $G$ is a vertex subset $S\subseteq V$ such that $G[S]$ contains no permanent edge. The GIS problem can be formulated as finding a generalized independent set such that its net benefit is maximized.} We denote the optimal objective value of the GIS problem in $G$ as
{ $\alpha(G)$, i.e.,
\[
\alpha(G)=\max_{S\subseteq V}nb(S) \quad s.t. \quad  E_p\cap \{\{u,v\}\mid u,v\in S\}=\varnothing
\]
}
{ We also denote $MGIS(G)$ as the family of generalized independent sets $S\subseteq V$ that achieve the optimal objective value $\alpha(G)$. Each $S\in MGIS(G)$ corresponds to an optimal solution.}
\begin{table}[t]
\centering
\caption{Table of Notations}
\label{tab:notations}
\resizebox{0.8\textwidth}{!}{%
\begin{tabular}{|l|l|}
\hline
$G=(V,E=E_p\cup E_r)$ & 
\begin{tabular}[l]{@{}l@{}} 
a graph $G$ with vertex set $V$ and edge set $E$, \\
$E_p, E_r$ represent the set of permanent edges and the set of removable edges, respectively
\end{tabular} 
\\ \hline
$w(v), p(\{u,v\})$ & the profit of including a vertex $v$, and the penalty of including both vertices $u$ and $v$ \\ \hline
$N(v), N_p(v), N_r(v)$ & \begin{tabular}[l]{@{}l@{}} { subsets of $V$, representing} the sets of vertices connected to $v$, vertices connected to $v$ by \\ a permanent edge and vertices connected to $v$ by a removable edge, respectively \end{tabular} \\ \hline
$|N(v)|,|N_p(v)|, |N_r(v)|$ & the degree, permanent degree, and removable degree of $v$, respectively \\ \hline
$N[v], N_p[v], N_r[v]$ & \begin{tabular}[l]{@{}l@{}}{ subsets of $V$, representing} the sets $N(v)\cup \{v\}$, $N_p(v)\cup \{v\}$, $N_r(v)\cup \{v\}$, respectively \end{tabular} \\ \hline
$\Delta$ & the maximum degree of vertices in $G$ \\ \hline
$\Tilde{w}(v)$ & \begin{tabular}[c]{@{}l@{}}$\Tilde{w}(v)= w(v)+\sum_{x\in N_r(v)} \max(0,-p(\{v,x\}))$
\end{tabular} \\ \hline
$S, {G[S], E(G[S])}$ & a vertex subset of $G$, the subgraph induced by $S$, and the edge set of the induced subgraph
\\ \hline
$w(S), w^+(S)$ & $w(S)=\sum_{v\in S}w(v)$, $w^+(S) = \sum_{v\in S} \max(0,\Tilde{w}(v))$
\\ \hline
$nb(S)$ &  the net benefit of a vertex set $S$, $nb(S)= \sum_{v\in S} w(v)-\sum_{ u\in S,v\in S,e=\{u,v\}\in E_r} p(e)$ \\ \hline
$\alpha(G),MGIS(G)$ & \begin{tabular}[c]{@{}l@{}} the largest net benefit value achievable without including any permanent edge, \\ and the family of all subsets of $V$ attaining $\alpha(G)$ value under the same constraint \end{tabular} \\ \hline
\end{tabular}%
}
\end{table}

For convenience, we report more notations in Table~\ref{tab:notations}.
Note that we use $\Tilde{w}(v)$ to represent the profit of $v$ plus the { absolute sum of  non-positive} penalties of edges that are incident to $v$.
Given that $S\subseteq V$ is a generalized independent set, $w^+(S)$ is an upper bound on the net benefit of { $S$.
For example, in Figure~\ref{fig: facility}, $\Tilde{w}(u)=2+2.5=4.5$.
When $S$ denotes the set of red vertices in Figure~\ref{fig: facility}, $w^+(S)=4.5+6+5=15.5$.
}

In the rest of the paper, we use the notation $\coloneqq$ to represent assignment and $=$ to indicate equivalence.

\noindent {\bf Integer Programming.}
The GIS problem can be conveniently formulated as follows \citep{colombi2017generalized}:
\begingroup
\addtolength{\jot}{-0.5em}
\begin{align}
\alpha(G)=\max \quad & \sum_{i\in V}w(i) x_i-\sum_{\{i,j\}\in E_r} p(\{i,j\}) y_{ij} \label{eq:1} \\
\text{s.t. \quad} &x_i+x_j\leq 1 && \{i,j\}\in E_p  \label{eq:2}  \\
& x_i+x_j-y_{ij}\leq 1 && \{i,j\}\in E_r \label{eq:3}  \\
& { x_i\geq y_{ij}, x_j\geq y_{ij}} && \{i,j\}\in E_r\label{eq:5} \\
& x_i\in \{0,1\}  && i\in V \label{eq:4}  \\
& y_{ij}\in \{0,1\} && \{i,j\}\in E_r \label{eq:6} 
\end{align}
\endgroup

\noindent
where binary variables $x_i$ indicate whether the vertex $i$ is selected, and binary variables $y_{ij}$ indicate whether the removable edge $\{i,j\}$ is included.

\section{Reduction Rules} \label{sec: reduction}

In this section, we develop reduction rules that identify reducible graph structures without compromising optimality.
We rely on four operations to reduce the current instance to a simplified instance.
(1) \textit{Edge transformation}, which identifies  { removable edges} that can be deleted or transformed to { permanent edges};
(2) \textit{Vertex inclusion}, which identifies vertices that are in { at least one} set $I\in MGIS(G)$;
(3) \textit{Vertex exclusion}, which identifies vertices that { do not belong to a} set $I\in MGIS(G)$.
(4) \textit{Folding}, which identifies vertices that can be merged into a single vertex.
Based on the four operations, we introduce a total of 14 reduction rules as follows.
{ For simplicity, we denote the graph obtained after applying a reduction rule by $G'=(V',E'=E_p'\cup E_r')$. Unless otherwise specified, we assume that $w'(v)$, the profit of a vertex $v\in V'$,keeps its original profit $w(v)$ if $v$ also exists in  $V$, and $p'(e)$, the penalty of a removable edge $e\in E_r'$, keeps its original penalty $p(e)$ if $e$ also exists in $E_r$.}

\subsection{Edge Transformation}

We first introduce two edge transformation rules.



\begin{lemma}[Edge Removal Reduction] \label{lem: reduction 1}
If there exists an edge { $e=\{u,v\}\in E_r$} such that $p(e)=0$, then $\alpha(G)=\alpha(G')$ where $G'$ is obtained by only removing the edge $e$ from $G$.
\end{lemma}
The proofs of the correctness of this reduction rule, and the subsequent reduction rules, are deferred to Appendix~\ref{app: proofs}.
R\ref{lem: reduction 1} implies removing an edge $e$ with $p(e)=0$ does not change optimality.





\begin{lemma}[Edge-penalty Reduction] \label{lem: reduction 2}
    If there is an edge { $e=\{u,v\}\in E_r$} such that $p(e)> min(\Tilde{w}(u),\Tilde{w}(v))$, then $\alpha(G)=\alpha(G')$ where $G'$ is obtained by {only moving $e$ from $E_r$ to $E_p$ in $G$}.
\end{lemma}
Compared to a permanent edge, a removable edge is more challenging to reduce because of the difficulties of identifying whether both endpoints of a removable edge are in a set $S\in MGIS(G)$. R\ref{lem: reduction 2} allows us to transform a removable edge into a permanent edge without sacrificing optimality, underpinning many further reduction rules.




\subsection{Neighborhood-based Reductions}
We  further consider the local neighborhood of a vertex.

\begin{lemma}[Neighborhood Weight Reduction] \label{lem: reduction 3}
If there is a vertex $u\in V$ { such} that $w(u)\geq w^+(N(u))$, 
then { $u$ is in at least one set $I\in MGIS(G)$} and 
$\alpha(G)=\alpha(G')+w(u)$ where $G'\coloneqq G[V\setminus N_p[u]]$, and $w'(v)\coloneqq w(v)-p(\{u,v\})$ for each $v\in N_r(u)$.
\end{lemma}

We extend R\ref{lem: reduction 3} by considering more neighboring structures to obtain R\ref{lem: reduction 4}-R\ref{lem: reduction 6}.


\begin{lemma}[Neighborhood Penalty Reduction] \label{lem: reduction 4}
If there exists a vertex $u\in V$ such that $w(u)\geq  w^+(N_p(u))+\sum_{x\in N_r(u)}\max(0, p(\{u,x\}))$, then 
{$u$ is in at least one set $I\in MGIS(G)$} and
$\alpha(G)=\alpha(G')+w(u)$ where $G'\coloneqq G[V\setminus N_p[u]]$, and $w'(v)\coloneqq w(v)-p(\{u,v\})$ for each $v\in N_r(u)$. 
\end{lemma}



\begin{lemma}[Negative Profit Reduction] \label{lem: negative profit} \label{lem: reduction 5}
If there exists a 
vertex $u\in V$ such that $w(u)<0$ and $\sum_{v\in N_r(u)} \min(0,p(\{u,v\}))>w(u)$, then 
{ $u$ is not in any $I\in MGIS(G)$} and 
$\alpha(G)=\alpha(G')$ where $G'\coloneqq G[V\setminus \{u\}]$.
\end{lemma}

{Next, we show the \textit{clique reduction} rule. A \textit{clique} of $G$ is a subset of vertices in which every two vertices are connected by a permanent edge.}
\begin{lemma}[Clique Reduction] \label{lem: reduction 6}
If there exists a vertex $u\in V$ such that $N_p(u)$ is a clique  of $G$ and $w(u)\geq \sum_{v\in N_r(u)} \max(0,p(\{u,v\}))+ \max(0,\max_{v\in N_p(u)}\Tilde{w}(v))$, then {$u$ is in at least one set $I\in MGIS(G)$}, $\alpha(G)=\alpha(G')+w(u)$ where $G'\coloneqq G[V\setminus N_p[u]]$, and $w'(v)\coloneqq w(v)-p(\{u,v\})$ for each $v\in N_r(u)$.
\end{lemma}





\subsection{Low-degree Reductions}
Low-degree vertices such as degree-one and degree-two vertices are common in large real-world graphs. Next, we introduce five reduction rules to reduce these vertices.

\subsubsection{Degree-one Vertices Reduction}

\begin{lemma}[Degree-one Vertices Reduction] \label{lem: reduction 7}
    Let $u$ be a vertex with a single neighbor $v$.

    Case 1: $\{u,v\}\in E_r$ and $w(u)\geq p(\{u,v\})$. (1) If $w(u)\geq 0$, then $u$ is in at least one set $I\in MGIS(G)$ and $\alpha(G)=\alpha(G')+w(u)$, where $G'\coloneqq G[V\setminus \{u\}]$ and $w'(v)\coloneqq w(v)-p(\{u,v\})$.  (2) If $w(u)<0$, then $\alpha(G)=\alpha(G')$, where $G'\coloneqq G[V\setminus \{u\}]$ and $w'(v)\coloneqq w(v)-p(\{u,v\})+w(u)$. {$u$ is in at least one set $I\in MGIS(G)$ iff $v$ is in at least one set $I\in MGIS(G')$}.

    Case 2: $\{u,v\}\in E_p$ or $\{u,v\}\in E_r$ but $w(u)<p(\{u,v\})$. (1) If $w(u)\geq 0$ and $w(u)\geq \Tilde{w}(v)$, then {$u$ is in at least one set $I\in MGIS(G)$, $v\notin I$}, and $\alpha(G)=\alpha(G')+w(u)$ where $G'\coloneqq G[V\setminus \{u,v\}]$. (2) If $w(u)\geq 0$ and $w(u)<\Tilde{w}(v)$, then {$u$ and $v$ cannot be both included in the same $I\in MGIS(G)$}, we have $\alpha(G)=\alpha(G')+w(u)$, where $G'\coloneqq G[V\setminus \{u\}]$ and $w'(v)\coloneqq w(v)-w(u)$. 
    (3) If $w(u)<0$, then
    {$u$ is not in any $I\in MGIS(G)$} and
    $\alpha(G)=\alpha(G')$ where $G'\coloneqq G[V\setminus \{u\}]$. 
\end{lemma}
Note that R\ref{lem: reduction 7} covers all degree-one vertices. In other words, any vertex of degree one can be removed using this rule.






\subsubsection{Degree-two Vertices Reductions}
{ For ease of describing $R8$ and $R9$, we assume $p(\{u,v\})=+\infty$ for each $\{u,v\}\in E_p$. This assumption does not affect the optimal objective value in the graph. }


\begin{lemma}[Permanently Connected Neighbors Reduction] \label{lem: reduction 8}
    Let $u$ be a vertex with exactly two neighbors $x,y$ and $\{x,y\}\in E_p$. The following table shows the rules for reducing $u$ under different conditions.

\begingroup
\vspace{0.3cm}
\centering
\resizebox{\textwidth}{!}{%
\begin{tabular}{|c|c|l|}
\hline
\multicolumn{2}{|c|}{conditions}  & \multicolumn{1}{c|}{reduction} \\ \hline
\multirow{4}{*}{\begin{tabular}[c]{@{}c@{}}Case 1: \\  $w(u)\geq \max(p(\{u,x\}),p(\{u,y\}))$\end{tabular}}
& $w(u)\geq 0$ & \begin{tabular}[l]{@{}l@{}}{\bf R8.1.1:} { $u$ is in at least one set $I\in MGIS(G)$},\\ $\alpha(G)=\alpha(G')+w(u)$, where $G'\coloneqq G[V\setminus \{u\}]$, \\ $w'(x)\coloneqq w(x)-p(\{u,x\})$, $w'(y)\coloneqq w(y)-p(\{u,y\})$.\end{tabular} \\ \cline{2-3} 
 & $w(u)<0$ & \begin{tabular}[l]{@{}l@{}}{\bf R8.1.2:} {$u$ is in at least one set $I\in MGIS(G)$ iff $x\in I$ or $y\in I$}, \\ $\alpha(G)=\alpha(G')$, where $G'\coloneqq G[V\setminus \{u\}]$, \\ $w'(x)\coloneqq w(x)+w(u)-p(\{u,x\})$, $w'(y)\coloneqq w(y)+w(u)-p(\{u,y\})$.\end{tabular} \\ \hline
\multirow{6}{*}{\begin{tabular}[c]{@{}c@{}}Case 2: \\ $p(\{u,x\})> w(u)\geq p(\{u,y\})$\\ w.l.o.g, assume $p(\{u,x\})> p(\{u,y\})$\end{tabular}} & \begin{tabular}[l]{@{}l@{}}$w(u)\geq 0$ and\\  $w(u)\geq \Tilde{w}(x)$\end{tabular} & \begin{tabular}[l]{@{}l@{}}{\bf R8.2.1:}
{$u$ is in at least one set $I\in MGIS(G)$ and $x\notin I$}, \\
$\alpha(G)=\alpha(G')+w(u)$ where $G'\coloneqq G[V\setminus \{u,x\}]$, \\ and $w'(y)\coloneqq w(y)-p(\{u,y\})$.\end{tabular} \\ \cline{2-3} 
 & \begin{tabular}[l]{@{}l@{}}$w(u)\geq 0$ and\\  $w(u)<\Tilde{w}(x)$\end{tabular} & \begin{tabular}[l]{@{}l@{}}{\bf R8.2.2:} 
 $\alpha(G)=\alpha(G')+w(u)$ where $G'\coloneqq G[V\setminus \{u\}]$,\\  $w'(x)\coloneqq w(x)-w(u)$, $w'(y)\coloneqq w(y)-p(\{u,y\})$. \\
 {$u$ is in at least one set $I\in MGIS(G)$ iff $x$ is not in at least one set $I'\in MGIS(G')$}.
 \end{tabular} \\ \cline{2-3} 
 & $w(u)<0$ & \begin{tabular}[l]{@{}l@{}}{\bf R8.2.3:}
 $\alpha(G)=\alpha(G')$ where $G'\coloneqq G[V\setminus \{u\}]$,\\  $w'(y)\coloneqq w(y)+w(u)-p(\{u,y\})$.
 \\
 {$u$ is in at least one set $I\in MGIS(G)$ iff $y$ is in at least one set $I'\in MGIS(G')$}.
 \end{tabular} \\ \hline
\multirow{7}{*}{\begin{tabular}[c]{@{}c@{}}Case 3: \\ $\min(p(\{u,x\}),p(\{u,y\}))>w(u)$ \\ w.l.o.g, assume $\Tilde{w}(x)\geq \Tilde{w}(y)$\end{tabular}} & \begin{tabular}[l]{@{}l@{}}$w(u)\geq 0$ and\\  $w(u)\geq \Tilde{w}(x)$\end{tabular} &\begin{tabular}[l]{@{}l@{}}{\bf R8.3.1:}  {$u$ is in at least one set $I\in MGIS(G)$},\\ $\alpha(G)=\alpha(G')+w(u)$ where $G'\coloneqq G[V\setminus \{u,x,y\}]$ \end{tabular} \\ \cline{2-3} 
 & \begin{tabular}[c]{@{}c@{}}$w(u)\geq 0$ and \\  $\Tilde{w}(y)\leq w(u)<\Tilde{w}(x)$\end{tabular} & \begin{tabular}[l]{@{}l@{}}{\bf R8.3.2:} $\alpha(G)=\alpha(G')+w(u)$ where $G'\coloneqq G[V\setminus \{u,y\}]$, \\ $w'(x)\coloneqq w(x)-w(u)$.
 \\
 {$u$ is in at least one set $I\in MGIS(G)$ iff $x$ is not in at least one set $I'\in MGIS(G')$}.
 \end{tabular} \\ \cline{2-3} 
 & \begin{tabular}[c]{@{}c@{}}$w(u)\geq 0$ \\ and $w(u)<\Tilde{w}(y)$\end{tabular} & \begin{tabular}[l]{@{}l@{}}{\bf R8.3.3:} $\alpha(G)=\alpha(G')+w(u)$ where $G'\coloneqq G[V\setminus \{u\}]$,\\  $w'(x)\coloneqq w(x)-w(u)$, $w'(y)\coloneqq w(y)-w(u)$.
 \\
 {$u$ is in at least one set $I\in MGIS(G)$ iff both $x$ and $y$ are not in at least one set $I'\in MGIS(G')$}.
 \end{tabular} \\ \cline{2-3} 
 & $w(u)<0$ & \begin{tabular}[l]{@{}l@{}}{\bf R8.3.4:} {$u$ is not in any $I\in MGIS(G)$}, $\alpha(G)=\alpha(G')$ where $G'\coloneqq G[V\setminus \{u\}]$. \end{tabular} \\ \hline 
\end{tabular}%
}
\vspace{0.2cm}
\endgroup
\end{lemma}

\vspace{-0.1cm}
Examples of $R8$ are presented in Figure~\ref{fig:reduction 8}. For example, in the first plot, $u$ satisfies case 1 and $w(u)\ge 0$. Then by $R8.1.1$, we can remove $u$ from the original graph $G$ (left) to get the reduced graph $G'$(right) and $\alpha(G)=\alpha(G')+10$.

For ease of describing $R9$-$R11$, we further assume $p(\{u,v\})= 0$ for $\{u,v\}\notin E$. Note that this does not alter the objective value in the graph. { Then, for every pair of vertices $\{u,v\}$, there is a corresponding penalty $p(\{u,v\})$.}

\begin{lemma}[Non-permanently Connected Neighbors Reduction] \label{lem: reduction 9}
    Let $u$ be a vertex with two neighbors $x$ and $y$. Assume $p(\{u,x\})\geq p(\{u,y\})$ without loss of generality. If $\{x,y\}\in E_r$ or $\{x,y\}\notin E$, the following table summarizes the reduction rules under different conditions.

\begingroup
\vspace{0.1cm}
\centering
\resizebox{\textwidth}{!}{%
\begin{tabular}{|c|c|l|}
\hline
\multicolumn{2}{|c|}{conditions} & \multicolumn{1}{c|}{reduction} \\ \hline
\multirow{7}{*}{\begin{tabular}[c]{@{}c@{}}Case 1: \\ $w(u)\geq p(\{u,x\})$\end{tabular}} & \begin{tabular}[c]{@{}c@{}}$w(u)\geq 0$ and \\ $w(u)\geq p(\{u,x\})+p(\{u,y\})$\end{tabular} & \begin{tabular}[l]{@{}l@{}}{\bf R9.1.1:}
{$u$ is in at least one set $I\in MGIS(G)$},\\
$\alpha(G)=\alpha(G')+w(u)$ where $G'\coloneqq G[V\setminus \{u\}]$,\\  $w'(x)\coloneqq w(x)-p(\{u,x\})$, $w'(y)\coloneqq w(y)-p(\{u,y\})$\end{tabular} \\ \cline{2-3} 
 & \begin{tabular}[c]{@{}c@{}}$w(u)\geq 0$ and \\ $w(u)< p(\{u,x\})+p(\{u,y\})$\end{tabular} & \begin{tabular}[l]{@{}l@{}}{\bf R9.1.2:} $\alpha(G)=\alpha(G')+w(u)$ where $G'\coloneqq (V\setminus \{u\}, E(G[V\setminus \{u\}])\cup \{\{x,y\}\})$,\\ $w'(x)\coloneqq w(x)-p(\{u,x\})$, $w'(y)\coloneqq w(y)-p(\{u,y\})$, $p'(\{x,y\})\coloneqq p(\{x,y\})-p(\{u,x\})-p(\{u,y\})+w(u)$ \\
 {$u$ is in at least one set $I\in MGIS(G)$ iff either $x$ or $y$ is not in at least one set $I'\in MGIS(G')$}.
 \end{tabular} \\ \cline{2-3} 
 & $w(u)<0$ & \begin{tabular}[l]{@{}l@{}}{\bf R9.1.3:}  $\alpha(G)=\alpha(G')$ where $G'\coloneqq (V\setminus \{u\}, E(G[V\setminus \{u\}])\cup \{\{x,y\}\})$,\\  $w'(x)\coloneqq w(x)+w(u)-p(\{u,x\})$, $w'(y)\coloneqq w(y)+w(u)-p(\{u,y\})$,  $p'(\{x,y\})\coloneqq p(\{x,y\})+w(u)$.
 \\
 {$u$ is in at least one set $I\in MGIS(G)$ iff either $x$ or $y$ is in at least one set $I'\in MGIS(G')$}.
 \end{tabular} \\ \hline
\multirow{10}{*}{\begin{tabular}[c]{@{}c@{}}Case 2: \\ $p(\{u,x\})> w(u)\geq p(\{u,y\})$\end{tabular}} & \begin{tabular}[c]{@{}c@{}}$w(u)\geq 0$ and \\ $w(u)\geq p(\{u,x\})+p(\{u,y\})$\end{tabular} & \begin{tabular}[l]{@{}l@{}}{\bf R9.2.1:} 
 $\alpha(G)=\alpha(G')+w(u)$ where $G'\coloneqq (V\setminus \{u\}, E(G[V\setminus \{u\}])\cup \{\{x,y\}\})$,\\  $w'(y)\coloneqq w(y)-p(\{u,y\})$, $w'(x)\coloneqq w(x)-w(u)$, $p'(\{x,y\})\coloneqq p(\{x,y\})+p(\{u,x\})-w(u)$.
 \\
 {$u$ is in at least one set $I\in MGIS(G)$ iff either $x$ is not in or $y$ is in at least one set $I'\in MGIS(G')$}.
 \end{tabular} \\ \cline{2-3} 
 & \begin{tabular}[c]{@{}c@{}}$w(u)\geq 0$ and \\ $w(u)< p(\{u,x\})+p(\{u,y\})$\end{tabular} & \begin{tabular}[l]{@{}l@{}}{\bf R9.2.2:} 
 $\alpha(G)=\alpha(G')+w(u)$ where $G'\coloneqq (V\setminus \{u\}, E(G[V\setminus \{u\}])\cup \{\{x,y\}\})$,\\  $w'(y)\coloneqq w(y)-p(\{u,y\})$, $w'(x)\coloneqq w(x)-w(u)$, $p'(\{x,y\})\coloneqq p(\{x,y\})-p(\{u,y\})$.
 \\
 {$u$ is in at least one set $I\in MGIS(G)$ iff $x$ is not in at least one set $I'\in MGIS(G')$}.
 \end{tabular} \\ \cline{2-3} 
 & \begin{tabular}[c]{@{}c@{}}$w(u)<0$ and \\ $w(u)\geq p(\{u,x\})+p(\{u,y\})$\end{tabular} & \begin{tabular}[l]{@{}l@{}}{\bf R9.2.3:} $\alpha(G)=\alpha(G')$ where $G'\coloneqq (V\setminus \{u\}, E(G[V\setminus \{u\}])\cup \{\{x,y\}\})$,\\  $w'(y)\coloneqq w(y)+w(u)-p(\{u,y\})$, $p'(\{x,y\})\coloneqq p(\{x,y\})+p(\{u,x\})$.
 \\
 {$u$ is in at least one set $I\in MGIS(G)$ iff $y$ is in at least one set $I'\in MGIS(G')$}.
 \end{tabular} \\ \cline{2-3} 
 & \begin{tabular}[c]{@{}c@{}}$w(u)<0$ and \\ $w(u)< p(\{u,x\})+p(\{u,y\})$\end{tabular} & \begin{tabular}[l]{@{}l@{}}{\bf R9.2.4:} $\alpha(G)=\alpha(G')$ where $G'\coloneqq (V\setminus \{u\}, E(G[V\setminus \{u\}])\cup \{\{x,y\}\})$,\\ $w'(y)\coloneqq w(y)+w(u)-p(\{u,y\})$, $p'(\{x,y\})\coloneqq p(\{x,y\})+w(u)-p(\{u,y\})$.
 \\
 {$u$ is in at least one set $I\in MGIS(G)$ iff $x$ is not in and $y$ is in at least one set $I'\in MGIS(G')$}.
 \end{tabular} \\ \hline
\multirow{6}{*}{\begin{tabular}[c]{@{}c@{}}Case 3: \\ $p(\{u,y\})>w(u)$\end{tabular}} & $w(u)\geq 0$ & \begin{tabular}[l]{@{}l@{}}{\bf R9.3.1:} $\alpha(G)=\alpha(G')+w(u)$ where $G'\coloneqq (V\setminus \{u\}, E(G[V\setminus \{u\}])\cup \{\{x,y\}\})$, \\ $w'(x)\coloneqq w(x)-w(u)$, $w'(y)\coloneqq w(y)-w(u)$, $p'(\{x,y\})\coloneqq p(\{x,y\})-w(u)$.
\\
 {$u$ is in at least one set $I\in MGIS(G)$ iff both $x$ and $y$ are not in at least one set $I'\in MGIS(G')$}.
\end{tabular} \\ \cline{2-3} 
 & \begin{tabular}[c]{@{}c@{}}$w(u)<0$ and \\ $w(u)<p(\{u,x\})+p(\{u,y\})$\end{tabular} & {\bf R9.3.2:}
 {$u$ is not in any $I\in MGIS(G)$}, $\alpha(G)=\alpha(G')$ where $G'\coloneqq G[V\setminus \{u\}]$. \\ \cline{2-3} 
 & \begin{tabular}[c]{@{}c@{}}$w(u)<0$ and \\ $w(u)\geq p(\{u,x\})+p(\{u,y\})$\end{tabular} & \begin{tabular}[l]{@{}l@{}}{\bf R9.3.3:} $\alpha(G)=\alpha(G')$ where $G'\coloneqq (V\setminus \{u\}, E(G[V\setminus \{u\}])\cup \{\{x,y\}\})$, \\ $p'(\{x,y\})\coloneqq p(\{x,y\})-w(u)+p(\{u,x\})+p(\{u,y\})$.
 \\
{ $u$ is in at least one set $I\in MGIS(G)$ iff both $x$ and $y$ are in at least one set $I'\in MGIS(G')$}.
 \end{tabular} \\
 \hline
\end{tabular}%
}
\vspace{0.05cm}
\endgroup

\end{lemma}

\begin{figure}[h!]
    \centering
     \caption{Examples of $R8$. Suppose that vertex $u$ is adjacent to $x$ and $y$, and $\{x,y\}\in E_p$. 
     $G_R$ represents the residual graph $G[V\setminus \{u,x,y\}]$.
     $\alpha$ and $\alpha'$ indicate the optimal objective value in the original graph and the graph after reduction, respectively.
     }
     \vspace{0.05cm}
\includegraphics[width=\linewidth]{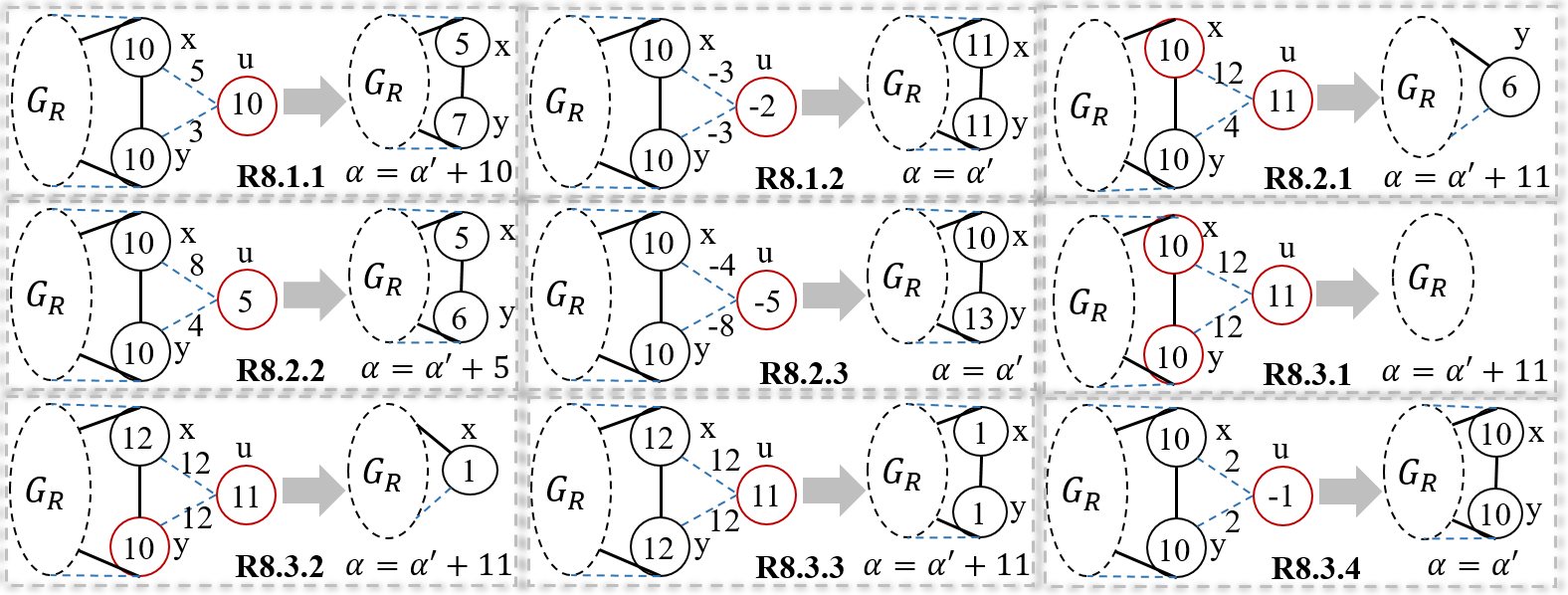}
    \label{fig:reduction 8}
\end{figure}

\begin{figure}[h!]
    \centering
     \caption{Examples of $R9$. Suppose that vertex $u$ is adjacent to vertices $x$ and $y$, and $\{x,y\}\in E_r$ or $\{x,y\}\notin E$. 
     }
     \vspace{0.05cm}
    \includegraphics[width=\linewidth]{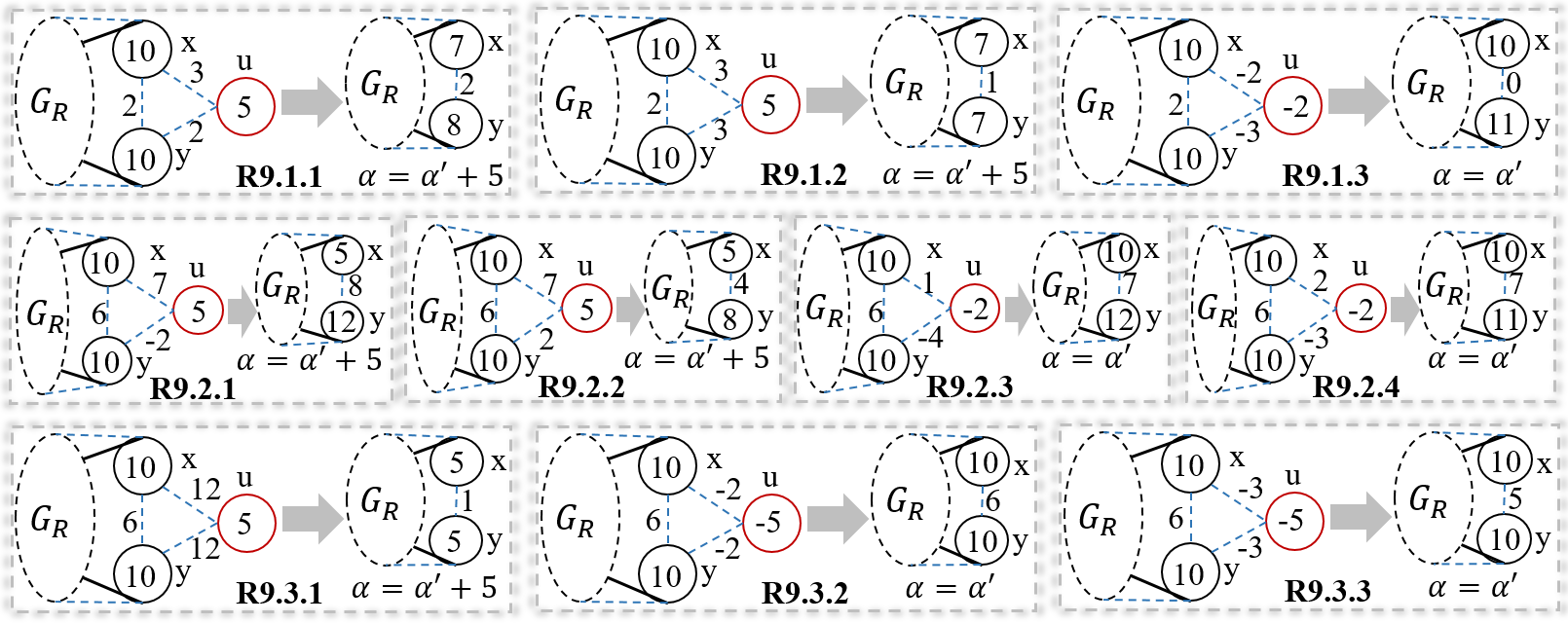}
    \label{fig: reduction 9}
\end{figure}
\vspace{-0.3cm}

Examples of reduction rule $R9$ are presented in Figure~\ref{fig: reduction 9}. 
As a matter of fact, $R8$ and $R9$ together can reduce all degree-two vertices. To be specific, R\ref{lem: reduction 8} covers the cases where the two neighbors are connected by a permanent edge, while R\ref{lem: reduction 9} covers the remaining cases where the two neighbors are connected by a removable edge or not connected.

\subsubsection{Low Permanent-degree Reductions}
We further introduce two reduction rules to reduce high-degree vertices that are connected by only one or two vertices by permanent edges. 

\begin{lemma}[Permanent Degree-one Vertices Reduction] \label{lem: reduction 10}
    If there exists a vertex $u$ { such} that $u$ has exactly one neighbor $x\in N_p(u)$ and $w(u)\geq \sum_{v\in N_r(u)} max(0,p(\{u,v\}))$, then
    $\alpha(G)=\alpha(G')+w(u)$ where $G'\coloneqq (V\setminus \{u\},E(G[V\setminus \{u\}])\cup \{\{x,v\}\mid v\in N_r(u)\setminus N_p(x)\})$, $w'(x)\coloneqq w(x)-w(u)$, $ w'(v)\coloneqq w(v)-p(\{u,v\})$ for each $v\in N_r(u)$, and $  p'(\{x,v\})\coloneqq p(\{x,v\})-p(\{u,v\})$ for each $v\in N_r(u)\setminus N_p(x)$. 
    {$u$ is in at least one set $I\in MGIS(G)$ iff $x$ is not in at least one set $I'\in MGIS(G')$}.
\end{lemma}


\begin{lemma}[Permanent Degree-two Vertices Reduction] \label{lem: reduction 11}
    Let $u$ be a vertex with exactly two neighbors $x,y\in N_p(u)$,  $\{x,y\}\in E_p$, and $L(u)=w(u)- \sum_{v\in N_r(u)} max(0,p(\{u,v\}))\geq 0$.
    W.l.o.g, assume $\Tilde{w}(x)\geq \Tilde{w}(y)$.  

    Case 1: If $L(u)\geq \Tilde{w}(x)$, then 
    { $u$ must be included in at least one set $I\in MGIS(G)$} and
    $\alpha(G)=\alpha(G')+w(u)$ where $G'\coloneqq G[V\setminus \{u,x,y\}]$, and $w'(v)\coloneqq w(v)-p(\{u,v\})$ for each $v\in N_r(u)$.

    Case 2: If $\Tilde{w}(y)\leq L(u)<\Tilde{w}(x)$, then 
    {$y$ is not in at least one set $I\in MGIS(G)$} and 
    $\alpha(G)=\alpha(G')+w(u)$ where $G'\coloneqq (V\setminus \{u,y\},E(G[V\setminus \{u,y\}])\cup \{\{x,v\}\mid v\in N_r(u)\setminus N_p(x)\})$, $w'(x)\coloneqq w(x)-w(u)$, $w'(v)\coloneqq w(v)-p(\{u,v\})$ for each $v\in N_r(u)$, and $p'(\{x,v\})\coloneqq p(\{x,v\})-p(\{u,v\})$ for each $v\in N_r(u)\setminus N_p(x)$. 
    {$u$ is in at least one set $I\in MGIS(G)$ iff $x$ is not in at least one set $I'\in MGIS(G')$}.

    Case 3: If $L(u)< \Tilde{w}(y)$, then $\alpha(G)=\alpha(G')+w(u)$ where $G'\coloneqq (V\setminus \{u\}, E(G[V\setminus \{u\}])\cup \{\{x,v\}\mid v\in N_r(u)\setminus N_p(x)\}\cup \{\{y,v\}\mid v\in N_r(u)\setminus N_p(y)\})$, $w'(x)\coloneqq w(x)-w(u)$, $w'(y)\coloneqq w(y)-w(u)$, $w'(v)\coloneqq w(v)-p(\{u,v\})$ for each $v\in N_r(u)$, $p'(\{x,v\})\coloneqq p(\{x,v\})-p(\{u,v\})$ for each $v\in N_r(u)\setminus N_p(x)$, and $p'(\{y,v\})\coloneqq p(\{y,v\})-p(\{u,v\})$ for each $v\in N_r(u)\setminus N_p(y)$. 
    {$u$ is in at least one set $I\in MGIS(G)$ iff both $x$ and $y$ are not in at least one set $I'\in MGIS(G')$}.
    
\end{lemma}

\subsection{Vertex-pair Reductions} \label{sec: edge based reductions}

We introduce three additional rules that can reduce some specific vertex pairs.
\begin{lemma}[Permanent Edge Reduction] \label{reduction:edge} \label{lem: reduction 12}
If there is an edge $\{u,v\}$ such that $\{u,v\}\in E_p$ and $w(u)\geq \Tilde{w}(v)+w^+(N_p(u)\setminus N_p[v])+\min(w^+(N_r(u)\setminus N_p(v)),\sum_{x\in N_r(u)\setminus N_p(v)} \max(0,p\{u,x\}))$, then {$v$ is not in at least one $I\in MGIS(G)$} and  $\alpha(G)=\alpha(G')$ where $G'\coloneqq G[V\setminus \{v\}]$.
\end{lemma}


\begin{lemma}[Common Neighbors Reduction] \label{reduction: extend edge} \label{lem: reduction 13}
If there is an edge $\{u,v\}$ such that $\{u,v\}\in E_p$ and $w(v)\geq w^+(N(v))-\max(0,w(u))$, then { $N_p(u)\cap N_p(v)$ is not a subset of at least one $I\in MGIS(G)$} and $\alpha(G)=\alpha(G')$, where $G'\coloneqq G[V\setminus (N_p(u)\cap N_p(v))]$.
\end{lemma}

\begin{lemma}[Twin Vertices Reduction] \label{lem: reduction 14}
If there exist two vertices $u,v$ such that $N_p(u)=N_p(v)$, $w(u)\geq \sum_{x\in N_r(u)} \max(0,p(\{u,x\}))$ and $w(v)\geq \sum_{x\in N_r(v)} \max(0,p(\{v,x\}))$, then $\alpha(G)=\alpha(G')$, where $G'$ is obtained by folding $u$ and $v$ into a vertex $v'$, $w(v')\coloneqq w(u)+w(v)-p(\{u,v\})$, add edge $\{v',x\}$ as a permanent edge in $G'$ for each $x\in N_p(u)$, and add edge $\{v',x\}$ as a removable edge in $G'$ with $p'(\{v',x\})\coloneqq p(\{u,x\})+p(\{v,x\})$ for each $x\in (N_r(v)\cup N_r(u))\setminus \{u,v\}$. 
{ Both $u$ and $v$ belong to at least one $I\in MGIS(G)$ iff $v'\in MGIS(G')$}.
\end{lemma}

\noindent

\section{Reduction-driven Local Search} \label{sec: local search}


Now, we elaborate on the local search algorithm driven by these 14 reduction rules for solving the GIS problem.
This reduction-driven local search (RLS) algorithm is presented in Algorithm~\ref{alg:local_search}.


\begin{algorithm}[h!]
\caption{The reduction-driven local search for GIS} \label{alg:local_search}
\KwIn{a graph $G=(V,E)$ with vertex profits and edge penalties, cutoff time}
\KwOut{ the best solution $S^*$ found}
$G,S_{init}\leftarrow \textit{Complete\_Reduction}(G)$;\label{line:reduction} \label{line:begin}
\tcc{Pre-process to obtain a kernel graph and partial solution via $R1$-$R14$} 
\While{cutoff time is not satisfied}{ \label{line:while begin}
       \tcc{construct an initial solution using $R1$-$R14$}
      $S\leftarrow S_{init}\cup \textit{Random\_Peeling}(G)$;\label{line:construct} \\
      \tcc{improve $S$ via local search and partial reductions}
      $S\leftarrow \textit{Neighborhood\_Search\_with\_Partial\_Reduction}(G,S)$;  \label{line:improve} \\
      \tcc{update the best visited solution}
      \lIf{$nb(S)>nb(S^*)$}{    \label{line:update begin} 
             $S^*\leftarrow S$   \label{line:update end} \label{line:while end}
        } 
}  
\Return $S^*$;

 \end{algorithm}

The RLS begins by pre-processing the original graph $G$. 
It uses the \textit{Complete\_Reduction} procedure which applies all the reduction rules to obtain a \textit{kernel graph}, that is, a graph { where no further reduction rules can be applied}. Due to the correctness of the reduction rules, solving the GIS problem in the original graph is equivalent to solving it in this kernel graph.
The $S_{init}$ is a partial vertex set that { must be a subset of at least one optimal solution}.

The remaining steps of the RLS follow the traditional iterated local search algorithm (\cite{glover1998tabu,lourencco2019iterated}): 
At each iteration, it supplements $S_{init}$ with a random generalized independent set obtained by \textit{Random\_Peeling}, then uses \textit{Neighborhood\_Search\_with\_Partial\_Reduction}  to improve the current solution (line~\ref{line:construct}-\ref{line:improve}) . 
The current solution is denoted by $S$ and the best-known solution visited by the iterative search is kept in $S^*$ (line~\ref{line:update begin}). 
The iteration continues until the cutoff time limit is met. 
Compared with traditional local search, the RLS seamlessly integrates reduction techniques into the pre-processing, the initial solution building, and the neighborhood searching procedures.




\begin{algorithm}[ht!]
\caption{The Complete\_Reduction algorithm} \label{alg:apply_reduction}
\KwIn{a graph $G=(V,E)$ with vertex profits and edge penalties}
\KwOut{updated graph $G$ and partial solution $S$}
$S\leftarrow \varnothing$; 
\tcc{ initial an empty solution}  
\tcc{initial candidate sets for 14 reduction rules} \label{line: reduce cand initial}
$\text{Cand}_i\leftarrow E_r$ for $i=1,2$; $\text{Cand}_i\leftarrow V$ for $i=3,4,6$; $\text{Cand}_5\leftarrow \{v\in V\mid w(v)<0\}$; \\
$\text{Cand}_7\leftarrow \{v\in V\mid |N(v)|=1\}$; $\text{Cand}_8\leftarrow\{v\in V\mid |N(v)|=2\}$; $\text{Cand}_9\leftarrow \text{Cand}_8$; \\
$\text{Cand}_{10}\leftarrow \{v\in V\mid |N_p(v)|=1\}$; $\text{Cand}_{11}\leftarrow\{v\in V\mid |N_p(v)|=2\}$; \\
$\text{Cand}_i\leftarrow E_p$ for $i=12,13$; $\text{Cand}_{14}\leftarrow V\times V\setminus E_p$; \label{line: reduce cand end} \\ 
\While{$\exists i: \text{Cand}_i\neq \varnothing$}{ \label{line: reduce while begin} 
\tcc{exhaustively apply reduction rules}
    \For{each $i=1,2,\ldots,14$}{
        \For{each $elem \in \text{Cand}_i$}{
        %
        \If{the condition of $Ri$ is satisfied by checking $elem$}{
            change the structure, profits and penalties of $G$ and set $S$ according to $Ri$; \\
            update the Cand$_j$ if $Rj$ reachable from $Ri$ in Figure \ref{fig:reduction dependency};\\
        }         
        Remove $elem$ from Cand$_i$;
        }
    }
}  \label{line: reduce end}
\Return $G,S$;
 \end{algorithm}

\subsection{Complete Reduction as Pre-processing} \label{sec: apply reduction}
The Complete\_Reduction algorithm exhaustively applies the 14 reduction rules described in Section~\ref{sec: reduction}. 
A naive implementation is to check every rule independently, process the graph when applicable, and stop when no rules apply. 
The times for checking all these rules are bounded by { $O(|V|+|E|)$} as it requires scanning the entire graph for each rule.
Given the large size of real-world graphs, this is inefficient in practice. We propose a more efficient implementation that incrementally detects applicable rules.
Specifically, we observe that rules are interdependent, with one application triggering others. For example, if only the condition of $R1$ is satisfied in graph $G$ (i.e., there exists an edge $\{u, v\}\in E_r$ with $p(\{u,v\}) = 0$), then the edge $\{u, v\}$ is removed, potentially enabling the application of $R3$-$14$, while the remaining rules remain inapplicable. Hence, we make an in-depth analysis of this dependency relationship for all these rules, and we show the result in Figure~\ref{fig:reduction dependency}.
In Algorithm~\ref{alg:apply_reduction}, we implement this reduction algorithm. It takes a graph $G$ as its input and outputs a kernel graph and a set $S\subseteq V$. 
In detail, it initializes a special { \textit{candidate set} Cand$_i$}, which contains either vertices, edges, or vertex pairs, for a rule $Ri$ where $i=1,...,14$.  Lines 1-4 show how Cand$_i$ is initialized.  
When an element $elem\in \text{Cand}_i$ satisfies the condition of  $Ri$, the current graph $G$ is changed and the current set $S$ is updated according to $Ri$.
Along with the change of graph $G$, if another reduction rule $Rj$ can be applied (i.e., $Rj$ is reachable from $Ri$ in the dependency graph in Figure~\ref{fig:reduction dependency}), then the vertices, edges, or vertex pairs which are ``neighbor elements'' of $elem$ are moved into Cand$_j$. { The process continues recursively until all candidate sets are empty.}
Therefore, the algorithm checks whether each reduction rule is satisfied by tracking the candidate set in an incremental manner.

{ In our implementations, $R1$ and $R2$ are prioritized because the removal of edges may enable more reduction rules. $R3$-$R6$ are assigned secondary priority, as they directly identify vertices that are in at least one set in MGIS(G), which may further shrink the candidate set.
$R7$-$R11$ are of the third priority because they can identify vertices that are in at least one set in MGIS(G) in some cases. 
$R12$-$R14$ are put at the last because they are not able to identify such vertices.}


\begin{figure}[ht!]
    \centering
     \caption{The dependency graph of the reduction rules. 
     A directed edge on the left side indicates how the reduction rule changes the graph $G$ to $G'$. A directed edge on the right side indicates what reduction rules can be satisfied due to this change. }
    \includegraphics[width=0.44\linewidth]{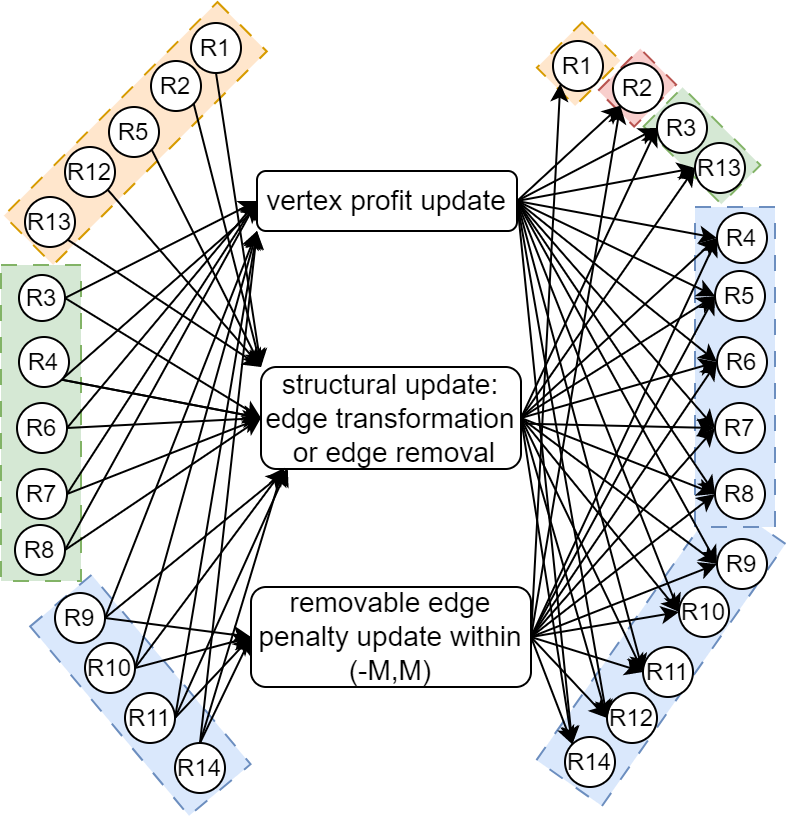}
   
    \label{fig:reduction dependency}
\end{figure}

\begin{algorithm}[ht!]
\caption{The Random\_Peeling algorithm} \label{alg:random_peeling}
\KwIn{a graph $G=(V,E)$ with vertex profits and edge penalties}
\KwOut{ a solution $S$ to the GIS problem}
initialize an empty set $S$; \\ \label{line:empty}
\While{$|V\setminus N_p[S]|>0$}{ \label{line: peel_while_begin}
    $u\leftarrow$ randomly select a vertex from $V\setminus N_p[S]$; \label{line:select} \\ 
    
    $G,S\leftarrow \textit{Complete\_Reduction}(G[V\setminus \{u\}],S)$;\label{line:reduce} \label{line:peel}
}
extend $S$ to be a maximal generalized independent set; \label{line: extend} \\
\Return $S$;
 \end{algorithm}
\subsection{Random Peeling for Initial Solutions}

Random-peeling is a simple random procedure for generating an initial solution.
It starts with an empty vertex set $S$.
When $V\setminus N_p[S]$ is not empty (meaning that $S$ can possibly be a larger generalized independent set), a vertex $u$ is picked randomly from $V\setminus N_p[S]$ and $u$ is moved into $S$. After the removal of $u$ from $G$, the Complete\_Reduction is used again to further shrink the remaining graph $G$ (line~\ref{line:reduce}). When $V\setminus N_p[S]$ becomes empty, we further extend $S$ by randomly picking a vertex $v$ { from the randomly removed vertices} such that $v\notin N_p[S]$ and $nb(S\cup \{v\})>nb(S)$. The $S$ is returned when { there is no other generalized independent set $S'$ such that $S\subset S'$ and $nb(S')\geq nb(S)$. Clearly $S$ is a maximal generalized independent set.}


\subsection{Neighborhood Search} \label{sec:operator}
To improve the current solution, we use an iterated tabu search with three \textit{move operators}.

\noindent
{\bf Move operators.} In the RLS, there are three move operators that update the current solution: ADD, SWAP, and the novel REDUCE operator based on two reduction rules.
Given the current solution $S\subseteq V$, the ADD$(v)$ operator displaces a vertex $v$ (if any) from $V\setminus N_p[S]$ to the $S$ and the SWAP$(u,v)$ operator adds a vertex $v\in N_p(S)$ if $|N_p(v)\cap S|=1$ and removes the only vertex $u\in N_p(v)\cap S$ from the current solution $S$.
Indeed, ADD and SWAP are widely used in existing local search methods such as \cite{zhou2017frequency,jin2015general}.  
The REDUCE$(V\setminus N_p[S])$ operator applies $R3$ and $R4$ to include candidate vertices that are guaranteed to be in a set in $MGIS(G)$. 

\begin{algorithm}[ht!]
\caption{The Neighborhood\_Search\_with\_Partial\_Reduction algorithm} \label{alg:neighborhood_search}
\KwIn{a graph $G=(V,E)$ with vertex profits and edge penalties, initial solution $S$,  {search tolerance $L\in \mathbb{N}$}, perturbation parameter $\epsilon\in (0,1)$}
\KwOut{a solution $S$ to the GIS problem}
$S^*\leftarrow S$, $tabu\_list\leftarrow \varnothing$; \label{line: initial_S} \\
\While{True}{ \label{line: nei_loop_begin}
    $M_1\leftarrow V\setminus N_p[S]$, \label{line: update_Candidate_begin} 
    $M_2\leftarrow \{v\in N_p(S)\mid u\in S, \{u,v\} \in E_p, |N_p(v)\cap S|=1\} \setminus tabu\_list$ \\ \label{line: update_Candidate_end}
    $v\leftarrow$ a vertex with the largest move gain from $M_1 \cup M_2$;\\
    \lIf*{$v\in M_1$}{  \label{line: add_operator}
        $S\leftarrow S\cup\{v\}$, $M_1\leftarrow M_1\setminus \{v\}$ \label{line: add_operator_end} \tcc{add $v$ into $S$}  
    }\uElseIf{$v\in M_2$} {  \label{line: swap_operator}
        $u\leftarrow N_p(v)\cap S$,
        $S\leftarrow S\cup\{v\}\setminus \{u\}$;   \tcc{Swap $v$ and $u$ }
        add $u$ into $tabu\_list$ with tabu tenure $T_u= 10+random(0,|M_2|)$ \label{line: swap_operator_end}
    }
    $REDUCE(M_1)$ and reassign $S$; \label{line: search-reduction} \tcc{apply $R3$ and $R4$}
    \lIf*{$nb(S)>nb(S^*)$}{ \label{line: nei_update_best}
        $S^*\leftarrow S$, $l\leftarrow 0$ }\lElse*{ \label{line: search_depth}
        $l\leftarrow l+1$
    }
    
    \uIf{$l>|S|$ or $M_1=\varnothing \& M_2=\varnothing$}{ \label{line: perturbation begin}
    sequentially drop $\epsilon |S|$ vertices with the least move gain from $S$, ties breaking randomly; \tcc{Perturbation}}   \label{line: perturbation}
    \lIf{$l>L$}{break} \label{line: break}
}
\Return $S^*$;
 \end{algorithm}

\noindent {\bf Fast Calculation of Move Gain.}
For each vertex $u\in V$, we maintain $B_u = w(u)- \sum_{v\in S,\{u,v\}\in E_r}p(\{u,v\})$. 
Then, the move gain of $ADD(v)$---the gain of the net benefit after applying a $ADD(v)$, is $B_v$, and the move gain of $SWAP(u,v)$ is $B_v-B_u$.
After an $ADD(v)$ move, for each $u\in N_r(v)$, we update $B_u$ as $B_u\coloneqq B_u-p(\{u,v\})$.
After a $SWAP(u,v)$ move, for each $x\in N_r(u)$, $B_x\coloneqq B_x+p(\{x,u\})$ and for each $y\in N_r(v)$, $B_y\coloneqq B_y-p(\{y,v\})$.
The time complexity of performing $ADD$ or $SWAP$ is $O(\Delta)$, and of REDUCE is $O(|E|)$. 
The overall complexity of applying all three operators is $O(|E|)$.
Notably, the REDUCE operator is not limited to the GIS problem and can enhance the efficiency of local search in other optimization problems.




\noindent
{\bf Neighborhood Search with Tabu and Perturbation.}
The neighborhood search in the RLS iteratively performs the three basic operators above to improve the current solution $S$, as shown in Algorithm~\ref{alg:neighborhood_search}.

At each iteration in lines~\ref{line: initial_S}-\ref{line: break}, it first builds $M_1$, the set of vertices that can be added into $S$ by $ADD$, and $M_2$, the set of vertices that can be swapped into $S$ by $SWAP$.
By evaluating $M_1$ and $M_2$, $ADD$ or $SWAP$ with the largest move gain is applied, as presented in lines 4-8.
After performing $ADD$ or $SWAP$, we additionally use the REDUCE operator to add more vertices to $S$.
This narrows the search space $M_1=V\setminus N_p[S]$ of $ADD$.

To avoid short-term search cycles, a tabu list is used to exclude recently visited vertices for successive iterations (line~\ref{line: swap_operator_end}). The tabu tenure for a vertex is set to $10+random(0,|M_2|)$, as commonly used in local search algorithms \citep{wu2023solving,zhou2017frequency,jin2015general}.
To further avoid local optima, a perturbation step is applied during neighborhood search (lines~\ref{line: perturbation begin}-\ref{line: perturbation}). Specifically, when we detect that no operators can be applied or the best solution remains unchanged after a certain number of iterations, we remove a portion of low-contribution vertices.
The whole neighborhood search terminates when the consecutive number of iterations without improvement reaches the { search tolerance} parameter $L$ (line~\ref{line: break}).




\section{Computational Experiments} \label{sec: experiments}
This section is dedicated to a comprehensive evaluation of the proposed reduction rules and heuristic approaches. 
To evaluate the algorithm performance, we compare both solution quality and runtime between the RLS and other benchmark algorithms. To evaluate the effect of the reduction rules, we compare the algorithms with and without reductions; we also conduct an ablation study to evaluate the effect of each reduction rule.

\subsection{Datasets}
The proposed reduction rules and the RLS algorithm are evaluated on three sets of $216+47+15=278$ instances, covering various application scenarios.

{\em SET-1.}
The first set of instances consists of $216$ standard instances with up to $400$ vertices and up to $71,820$ edges.
These instances were first introduced by \cite{colombi2017generalized} and tested in \cite{nogueira2021iterated,hosseinian2019algorithms,zheng2024exact}. 
This set of instances consists of 12 graphs with 125 to 400 vertices and 6,963 to 71,820 edges (available at https://or-dii.unibs.it/instances/Instances\_GISP.zip).

{\em SET-2.}
The second set of instances consists of $47$ medium-sized instances with up to $17,903$ vertices and up to $1,997,732$ edges, introduced by \cite{zheng2024exact} as benchmark instances for the GIS problem (available at https://github.com/m2-Zheng/GISP/tree/main).

{\em SET-3.}
To further evaluate the algorithms on large instances and diverse real-world scenarios, we introduce 15 instances with up to 18 million vertices and 261 million edges. The first 4 smaller instances represent ecosystem and road/facility networks, while the remaining 11 larger instances are drawn from hyperlinks, citations, and social networks.
The instances were generated using the same method for generating SET-2 (available at https://github.com/PlutoAiyi/RLS/tree/main).

\subsection{Experimental Settings}

All experiments are conducted on a machine with an Intel(R) Xeon(R) Platinum 8360Y CPU @ 2.40GHz and 2{TB} RAM. 

\noindent
{\bf Benchmark Algorithms.}
To evaluate the effectiveness of our proposed reduction rules and the RLS algorithm, we compare them with the best-performing algorithms for the GIS problem. 
\begin{itemize}[leftmargin=*]
     \item CB\&B: The exact branch-and-bound algorithm using the UBQP formulation proposed by \cite{hosseinian2019algorithms}. 

    \item LA-B\&B: The exact branch-and-bound algorithm using an Lagrange relaxation techniques by \cite{zheng2024exact}. 
    \item CPLEX: The mixed integer programming solver with version 22.12 (released on December 2024).
    
    \item ILS-VND: The iterative local search algorithm proposed by  \cite{nogueira2021iterated}. 
    
    \item ALS: The adaptive local search algorithm proposed by \cite{zheng2024exact} 

    \item RLS: Our proposed reduction-driven local search algorithm. The algorithm is implemented in C++. 


    
\end{itemize}
The source codes of CB\&B, LA-B\&B, ILS-VND and ALS were made available by the authors. All source codes are available at https://github.com/PlutoAiyi/RLS.


\noindent{\bf Parameter Settings.} The RLS algorithm requires two parameters:
the { search tolerance} $L$ and the perturbation parameter $\epsilon\in (0,1)$.
We set $L=10|V|$ where $|V|$ is the number of vertices of the input graph and $\epsilon= 0.2$. This setting was obtained using the general IRACE automatic parameter configuration tool \citep{lopez2016irace}. For large instances like {\em soc-pokec, LiveJ} and  {\em uk2002}, we decrease the { search tolerance} as $L=0.5|V|$ to encourage explorations within the cutoff time.

\subsection{Performance Comparison with Existing Algorithms} 

We first compared five algorithms on the SET-1 instances. 
The two branch \& bound algorithms CB\&B and LA-B\&B ran with a cutoff time of 3 hours (10,800 seconds) for each instance. 
If an exact algorithm fails to finish within the cutoff time, then we report the best solution found so far. 
The remaining three heuristic algorithms ran 10 trials with a cutoff time of 30 seconds for each instance.
We reported the largest value among the best-known objective values obtained in these trials. 
The setting is consistent with that in \citep{hosseinian2019algorithms,nogueira2021iterated}. 

In Table~\ref{tab: Set1 summary}, we summarize the computational results on SET-1.  In the last two columns, we report the number of instances where each algorithm achieves the best (\textit{\#Best}) and worst (\textit{\#Worst}) performance, with respect to the best-known objective value and runtime, respectively.
Note that the runtime of the RLS refers to the wall-clock time, including both pre-processing and local search. 
We observe that, in all instances, the best-known objective values obtained by local search algorithms are comparable to those obtained by exact algorithms. 
In particular, if an exact algorithm obtains the optimal solution within 3 hours, local search algorithms are also able to find an optimal solution.
However, when neither of the exact algorithms finds the optimal solution, local search algorithms always produce a better solution. 
Besides, the exact LA-B\&B dominates CB\&B whereas the heuristics ALS and the RLS dominate ILS-VND in terms of both solution quality and runtime. Therefore, we only consider LA-B\&B, ALS, and the RLS for even larger instances in SET-2 and SET-3.
Detailed results on SET-1 are shown in Table~\ref{tab: Set1 detail} in Appendix~\ref{app: set1}.

\begin{table}[ht!]
\caption{Summary of comparative results on SET-1 instances}
\label{tab: Set1 summary}
\centering
\footnotesize
\resizebox{0.5\textwidth}{!}{%
\begin{tabular}{|cc|c|cc|cc|}
\hline
\multicolumn{2}{|c|}{\multirow{2}{*}{Algorithm}} & \multicolumn{1}{l|}{} & \multicolumn{2}{c|}{Indicator: net benefit} & \multicolumn{2}{c|}{Indicator: time(s)} \\
\cline{4-7}
 & & \multicolumn{1}{l|}{\#Total} & \#Best & \#Worst  & \#Best & \#Worst \\ \hline 
\multirow{2}{*}{Branch \& Bound } &  CB\&B & 216  & 126 & 88  & 78 & 138  \\
   & LA-B\&B & 216  & 155 & 3  & 216 & 0  \\ 
\cline{1-2}
\multirow{3}{*}{Local Search} &ILS-VND & 216  & 215 & 0  & 164 & 52  \\
&ALS & 216  & 216 & 0  & 216 & 0  \\
&RLS & 216  & 216 & 0  & 216 & 0  \\ \hline
\end{tabular}%
}
\end{table}

\begin{table*}[ht]
\caption{Comparative results of RLS on SET-2 instances}
\label{tab: Set2}
\centering
\resizebox{0.9\textwidth}{!}{
\scriptsize
\begin{tabular}{|cccc|ccc|ccc|ccc|ccccc|}
\hline
\multirow{2}{*}{Instance} & \multirow{2}{*}{$|V|$} & \multirow{2}{*}{$|E_p|$} & {\multirow{2}{*}{$|E_r|$}}  &  \multicolumn{3}{c|}{ CPLEX} & \multicolumn{3}{c|}{LA-B\&B} & \multicolumn{3}{c|}{ALS} & \multicolumn{5}{c|}{RLS} \\ \cline{5-18} 
 &  &  & {} & BKV & {$|V_I|$} & htime(s) & BKV & { $|V_I|$} & { htime(s)} & BKV & { $|V_I|$} & \multicolumn{1}{c|}{htime(s)} & BKV & {$|V_I|$ }& htime(s) & \multicolumn{1}{c}{$|V_{\text{ker}}|$} &  time$_{\text{ker}}$(s)\\ \hline
bio-yeast & 1,458 & 980 & 968  &\textbf{68,574}$^*$ &1,111& $<$1 & 66,545 & 1,062 & 2 & {68,572} &1,113  & 9& \textbf{68,574}$^*$ &1,111  & {$<$1} & 0   & $<$1 \\
soc-wiki-Vote & 889 & 727 & 2,187  &\textbf{37,358}$^*$&648& $<$1 & 35,973 &615 & $<$1 & \textbf{37,358}$^*$ &648  & 8& \textbf{37,358}$^*$ &648  & {$<$1} & 0   &$<$1\\
vc-exact\_131 & 2,980 & 1,191 & 4,169  &\textbf{109,641}$^*$&2,098& $<$1 & 107,812 &2,045 & 14 & {109,392}&2,121  & 14& \textbf{109,641}$^*$ &2,098  & {$<$1} & 0   &$<$1\\
web-edu & 3,031 & 4,813 & 1,661  &\textbf{114,140}$^*$&1,752& $<$1 & 109,856 &1,644 & 10  & {113,781}&1,746 & 16& \textbf{114,140}$^*$ &1,753  & $<$1 & 98  &$<$1\\
MANN\_a81 & 3,321 & 3,118 & 3,362  &\textbf{117,327}$^*$&2,084& $<$1 & 114,616 &1,999 & 16  & {116,907} &2,125 & 14& \textbf{117,327}$^*$ &2,084  & {$<$1} & 0  & $<$1\\
tech-routers-rf & 2,113 & 1,539 & 5,093  &\textbf{97,799}$^*$&1,668&$<$1 & 93,690&1,588 & 54 & \textbf{97,799}$^*$ &1,668 & 11& \textbf{97,799}$^*$ & 1,667  & {$<$1} & 0  &$<$1\\
vc-exact\_039 & 6,795 & 2,574 & 8,046  &\textbf{259,542}$^*$&5,018& $<$1 & 254,941&4,908 & 253 & {258,137}&5,164 & 9& \textbf{259,542}$^*$ & 5,018  & {$<$1} & 0  &$<$1\\
ca-GrQc & 4,158 & 3,179 & 10,243 &\textbf{174,782}$^*$&2,991&9 & 167,846 &2,871 & 45 & {174,654} &2,997 & 15& \textbf{174,782}$^*$ &2,994  & 11& 282  &$<$1\\
vc-exact\_038 & 786  & 10,321 &3,703 &\textbf{12,015}$^*$&208& 589 & 10,406 &168 & 1,994 & {12,014}&208 & 12& \textbf{12,015}$^*$ &207  & 4& 697  &$<$1\\
vc-exact\_078 & 11,349 & 4,289 & 13,450 &\textbf{438,671}$^*$&8,330&$<$1 & 431,371 &8,128 & 1,074 & {435,412} &8,671 & $<$1 & \textbf{438,671}$^*$ & 8,330  & {$<$1} & 0  &$<$1\\
vc-exact\_087 & 13,590 & 5,126 & 16,114 &\textbf{518,500}$^*$&10,025& 2 & 509,734 &9,809 & 2,319 & {517,500} &10,083 & 2& \textbf{518,500}$^*$ &10,026  & {$<$1} & 0  &$<$1\\
vc-exact\_107 & 6,402 & 15,646 & 5,594 &\textbf{407,878}$^*$&6,763& $<$1 & 396,456&6,488 & 10,646 & {400,223} &6,846 & 1& \textbf{407,878}$^*$ &6,763  & {$<$1} & 0  &$<$1\\
vc-exact\_151 & 15,783 & 12,144 & 12,519 & \textbf{530,138}$^*$&9,526& 2 & 516,445 &9,135 & 2,402 & {522,462}& 9,764 & 1& \textbf{530,138}$^*$ & 9,526  & {$<$1} & 0  &$<$1\\
vc-exact\_167 & 15,783 & 18,235 & 6,428 &\textbf{465,178}$^*$&7,797& 1 & 449,579 &7,438 & 10,469 & {457,272}&7,923 & 2& \textbf{465,178}$^*$ &7,797  & {$<$1} & 0  &$<$1\\
bio-dmela & 7,393 & 18,897 & 6,672 &\textbf{302,992}$^*$&4,952& $<$1 & 280,897 &4,474 & 170 & {302,449} &4,959  & 10& \textbf{302,992}$^*$ &4,951 & $<$1 & 44  &$<$1\\
vc-exact\_011 & 9,877 & 19,102 & 6,871 &\textbf{303,793}$^*$&5,369&1 & 286,137&4,950 & 4,818 & 301,789 &5,487 & 7& \textbf{303,793}$^*$&5,387   & $<$1 & 172   &$<$1\\
web-spam & 4,767 & 8,089 & 28,386 &\textbf{190,507}$^*$&3,285&4.61  & 180,446&3,130 & 53 & {190,344}&3,299 & 14& 190,459 &3,289 & 13& 487  &$<$1\\
vc-exact\_026 & 6,140 & 27,159 & 9,608 &\textbf{201,509}$^*$&3,887& 1 & 177,902 &3,372&668 & {201,316} &3,922 & 17& 201,503&3,900  & $<$1 & 300  &$<$1\\
vc-exact\_001 & 6,160 & 19,655 & 20,552 &\textbf{215,941}$^*$&4,129& 2 & 194,984 &3,709 & 9,575  & {215,668} &4,175  & 13& \textbf{215,941}$^*$ &4,128  & $<$1 & 134 &$<$1\\
vc-exact\_024 & 7,620 & 11,367 & 35,926 &\textbf{236,109}&4,632& 10,783 & 226,540&4,438 & 244 & {235,440} &4,735 & 14& 236,018& 4,659  & 12& 1,255  &$<$1\\
C1000.9 & 1,000 & 11,989 & 37,432 &5,263&99&10,800 & {6,104} &116 & 430 & {7,229} &126 & 12& \textbf{7,239} &125  & 10& 1,000  &$<$1\\
tech-WHOIS & 7,476 & 42,181 & 14,762 &\textbf{325,584}$^*$&5,368&5 & 309,903&5,050 & 218  & {325,327} &5,367 & 14& 325,581& 5,371 & 13& 338 &$<$1\\
p-hat700-3 & 700 & 30,398 & 31,242 &4,025 &68& 10,800 & {3,030}&52 & 7,436  & \textbf{4,473} &78  & $<$1 & \textbf{4,473}&78  & $<$1 & 700 &$<$1\\
vc-exact\_008 & 7,537 & 35,678 & 37,155 &\textbf{247,091}$^*$&4,800&3 & 210,126&4,089 & 9,093 & {246,771} &4,839 & 14& \textbf{247,091}$^*$ &4,810  & $<$1 & 574  &$<$1\\
keller5 & 776 & 17,800 & 56,910 &2,373&58& 10,800 & {2,860} &54 & 6,327   & \textbf{3,877} &69  & 1& \textbf{3,877} &69  & 2& 776  &$<$1\\
vc-exact\_194 & 1,150 & 39,520 & 41,331 &3,734&77&10,708 & {4,386} &73 & 1,324 & \textbf{5,419}&86  & $<$1 & \textbf{5,419} & 86 & $<$1 & 1,150  &$<$1\\
hamming10-4 & 1,024 & 66,369 & 23,231 &3,256&52& 8,108 & {3,073} &45 & 2,069 & \textbf{3,660} &54  & $<$1 & \textbf{3,660} &54  & $<$1 & 1,024  &$<$1\\
brock800\_2 & 800 & 26,629 & {84,805} &1,177&49&10,800 & {1,996} &36 & 10  & \textbf{2,473} &44  & $<$1 & \textbf{2,473} &43  & $<$1 & {800}  &$<$1\\
brock800\_4 & 800 & 26,831 & {85,126} &1,583&48&10,800  & {1,979} &39& 7,215 & \textbf{2,501} &43  & $<$1 & \textbf{2,501} &42  & $<$1 & {800}  &$<$1\\
ca-HepPh & 11,204 & 57,721 & 59,898 &\textbf{372,096}&6,094&10,800 & 348,738 &5571 & 422 & {370,779} &6,083 & 6& 372,079 &6,098 & 17& 2,165  &$<$1\\
p-hat700-2 & 700 & 90,788 & 32,134 &2,870&41& 10,140  & {2,345} &37 & 2,078 & \textbf{2,971} &42  & $<$1 & \textbf{2,971} &42  & $<$1 & 700  &$<$1\\
vc-exact\_196 & 1,534 & 93,273 & 32,809 &3,916&62& 10,800 & {4,273}&57 & 66 & \textbf{5,321}&76  & $<$1 & \textbf{5,321} &76 & $<$1 & 1,534  &$<$1\\
p-hat700-1 & 700 & 135,804 & 47,847 &730&11& 10,800  & \textbf{844}$^*$ &11 & 5& \textbf{844}$^*$ &11  & $<$1 & \textbf{844}$^*$ &11  & $<$1 & 700  &$<$1\\
ca-AstroPh & 17,903 & 47,135 & 149,837 &\textbf{604,288}&10,468& 10,800 & 558,736 & 9,752 &2,544 & {602,320} &10,521 & 2& 604,019 &10,491 & 14& 4,985  & 1\\
C2000.9 & 2,000 & 147,548 & 51,920 &3,958&65&9,489 & {5,477} &73 & 11 & \textbf{6,396} &86 & 3& \textbf{6,396} &86  & 3& 2,000  &$<$1\\
DSJC1000.5 & 1,000 & 59,694 & 189,980 &925&39& 7,645 & {1,515} &30 & 694 & \textbf{1,861} &32  & {1}& \textbf{1,861}&32  & $<$1 & {1,000}  &$<$1\\
socfb-MIT & 6,402 & 123,064 & 128,166 &7,098&124&10,800 & 110,677 &1,779 & 20  & {129,807} &2,142 & 17& \textbf{129,917} &2,132 & 16& 5,494 &$<$1\\
p-hat1500-3 & 1,500 & 66,206 & 210,800 &4,097&132& 10,800  & {3,877}& 83 & 704 & \textbf{7,501} &148  & 14& \textbf{7,501} &148 & 16& 1,500  &$<$1\\
socfb-UCSB37 & 14,917 & 356,918 & 125,297 &16,708&287&10,800 & 219,965 &3,286 & 256 & {254,965} &4,013 & 16& \textbf{255,924}& 3,994 & 22& 13,655  &3\\
socfb-Duke14 & 9,885 & 374,731 & 131,706 &11,565&198&10,800 & 125,940 &1,887 & 39 & {151,169} &2,415 & 15& \textbf{151,365} &2,390  & 18& 9,003  &4\\
p-hat1500-2 & 1,500 & 411,372 & 143,918 &3,542&53& 10,800  & {2,611}&43 & 2,150 & \textbf{4,253} &68  & $<$1 & \textbf{4,253}&68  & $<$1 & 1,500  &$<$1\\
socfb-Stanford3 & 11,586 & 136,823 & 431,486 &7,250&122&10,800  & 215,381& 3,932 & 265 & {251,350} &4,441 & 12& \textbf{251,932} &4,436  & 14& 9,380  &2\\
socfb-UConn & 17,206 & 448,038 & 156,829 &16,996&279&10,800 & 231,718&3,471 & 322 & {274,637}&4,329  & 12& \textbf{275,320} &4,299  & 18& 16,025  &5\\
p-hat1500-1 & 1,500 & 411,744 & 427,583 &0&0&10,800 & {1,184} &18 & 5,499 & \textbf{1,287} &19  & $<$1 & \textbf{1,287} &19 & 1& 1,500  &$<$1\\
C2000.5 & 2,000 & 489,894 & {509,270} &0&0&10,800 & {1,581} &24 & 865 & \textbf{1,849} &26  & {2}& \textbf{1,849}&26  & 6& {2,000}  &1\\
keller6 & 3,361 & 759,816 & {266,766} &0 &0& 10,800  & {3,545}&55 &1,561 & \textbf{5,236} &73 & {10}& \textbf{5,236} &73 & 16& {3,361}  &$<$1\\
C4000.5 & 4,000 & 1,957,884 & {2,039,848} &0&0&10,800 & {1,611} &25 & 1,659 & {1,948} &28  & {20}& \textbf{1,952}&28  & 22& {4,000}  &9\\ \hline
\multicolumn{4}{|c|}{\# of best / \# of instances} & 24/47 & -- & --
& 1/47 & -- & -- &18/47  &-- & --  & \textbf{41}/47 &--    & -- & --  &--\\ \hline
\multicolumn{4}{|c|}{average running time(s)} &--&-- &5,402.1 &-- &-- &2,087.4  & -- &-- & 7.9 & -- &-- & \textbf{5.9}  &  -- & 1.4\\ 
\hline
\end{tabular}
}
\end{table*}

Next, we compared the RLS with LA-B\&B, ALS { and CPLEX} on SET-2 and SET-3. 
Each of the two heuristic ALS and the RLS were run 10 trials and reported its best known objective value. Each trial was given a cutoff time of 30 seconds. In total, each heuristic algorithm ran 300 seconds per instance in SET-2. The settings are the same as in \cite{zheng2024exact}. 
The cutoff time of exact algorithms  was directly set to {3 hours (10,800 seconds)} per instance in SET-2.
However, for SET-3, the heuristic trial time was increased to 200 seconds.
For large instances like {\em soc-pokec} and {\em LiveJ}, the cutoff time of ALS and the RLS were given 1,000 seconds per trial. For the larger {\em uk2002}, the cutoff time of LA-B\&B and CPLEX was 100,000 seconds, and the cutoff time of ALS and the RLS is 10,000 seconds per trial.
In particular, for the three instances {\em soc-pokec}, {\em LiveJ} and {\em uk2002}, the RLS excluded $R6$ and $R12$–$14$, { as these reductions exhibited considerably slower performance compared to the other rules in the experiments.}
For LA-B\&B and CPLEX, if the algorithm did not finish within the total running time, the best solution found so far was reported.

The comparative results on SET-2 and SET-3 are presented in Tables~\ref{tab: Set2} and \ref{tab:set3}. The first four columns list the instance information. The remaining columns show the best-known objective value (\textit{BKV}) of the algorithm, { the number of vertices in the BKV solution ($|V_I|$)},
the mean running time that the heuristic algorithm first hits the BKV (\textit{htime}), the kernel graph size ($|V_{\text{ker}}|$) and the time of generating a kernel (\textit{time$_{\text{ker}}$}). 
In the last two rows, we show the statistical results. 
The values marked with $*$ denote the optimal objective value.
Note that the htime of the RLS has included the pre-processing time for building graph kernels.

\begin{table}[ht!]
\caption{Comparative results of RLS on SET-3 instances}
\label{tab:set3}
\centering
\scriptsize
\resizebox{\textwidth}{!}{%
\begin{tabular}{|cccc|ccc|ccc|ccc|ccccc|}
\hline
\multirow{2}{*}{Instance} & \multirow{2}{*}{$|V|$} & \multirow{2}{*}{$|E_p|$} & \multirow{2}{*}{$|E_r|$} &  \multicolumn{3}{|c|}{ CPLEX}& \multicolumn{3}{|c|}{LA-B\&B} & \multicolumn{3}{c|}{ALS} & \multicolumn{5}{c|}{RLS} \\ \cline{5-18} 
 &  & & & BKV & { $|V_I|$} & htime(s) & BKV & { $|V_I|$} & { htime(s)} & BKV & {  $|V_I|$} & htime(s) & BKV & {  $|V_I|$}  & htime(s) & $|V_{\text{ker}}|$  &time$_{\text{ker}}$(s)\\ \hline
foodweb-wet & 128 & 1,509 & 566 &\textbf{2,464}$^*$ &44& $<$1 & \textbf{2,464}$^*$ &44 & 8 & \textbf{2,464}$^*$ &44  & $<$1 &  \textbf{2,464}$^*$ &44  & $<$1 
& 125 
 &$<$1 
\\
foodweb-dry & 128 & 1,571 & 535 &\textbf{2,226}$^*$&42& $<$1 & \textbf{2,226}$^*$ &42 & 8 & \textbf{2,226}$^*$  &42& $<$1 & \textbf{2,226}$^*$ &42  & $<$1 
& 125 
 &$<$1 
\\
USAir97 & 332 & 1,589 & 537 &\textbf{10,308}$^*$&194&$<$1 & 9,507&178 &407  & \textbf{10,308}$^*$  &194 &$<$1 & \textbf{10,308}$^*$ &194 & $<$1& 0 
 &$<$1 
\\
powergrid & 4,941 & 1,727 & 4,867 &\textbf{193,339}$^*$&3,710&$<$1& 190,537 &3,684 & 108& 193,095 &3,754 & 59 & \textbf{193,339}$^*$ &3,537  & $<$1& 0 
 &$<$1\\
CondMat & 23,133 & 46,769 & 46,670 &\textbf{711,517}$^*$&12,821&12& 685,442&12,061 & 5,749& 705,901&1,306  & 14 & 
711,334 &12,870  & 45
& 1,648
 &$<$1\\
Email & 265,009 & 91,093 & 273,388 &\textbf{13,054,427}$^*$&257,435&118& OOM & --&-- & OOM &-- & -- & 
\textbf{13,054,427}$^*$ &257,394  & 1 
& 0
 &1\\
Epinion & 75,879 & 202,653 & 203,087 &\textbf{3,052,059}$^*$&57,976&60& OOM& -- &--&  3,038,533&58,267 &  18& \textbf{3,052,059}$^*$ &58,008 & 1 
& 253 
 &1\\
 
Dblp & 317,080 & 787,268 & 262,598 &\textbf{9,737,911}$^*$&171,421&790& OOM & --&-- & OOM &-- & -- & 
9,734,653 &171,403 & 64
& 11,851
 & 11\\
cnr-2000 & 325,557 & 685,686 & 2,053,283 &13,367,460&260,331&10,800& OOM & --&-- & OOM &-- & -- & 
\textbf{13,368,670}&260,716 & 482
& 59,052
 &399 \\

WikiTalk & 2,394,385 & 3,494,873 & 1,164,692 &\textbf{118,708,326}$^*$&2,342,835&3,566& OOM & --&-- & OOM &-- & -- & \textbf{118,708,326}$^*$ &2,342,861  & 333 
& 8
 & 333\\
BerkStan & 685,230 & 1,662,833 & 4,986,637 &\textbf{26,036,210}&501,197&10,800& OOM & --&-- & OOM&--  & -- & 
26,031,704 &501,410 & 1,112
& 114,191
 & 1,031\\
As-Skitter & 1,696,415 & 5,550,661 & 5,544,637 &\textbf{66,603,360}&1,264,844&10,800& OOM & --&-- & OOM &-- & -- & 66,582,254 &1,265,266 & 723
& 31,173
 &636\\
soc-pokec & 1,632,803 & 11,152,945 & 11,149,019 &4,146,172&84,217&10,800& OOM & --&-- & OOM &-- & -- & \textbf{43,999,402} &833,989 & 934
& 881,319
 & 134 \\
LiveJ & 4,846,609 & 10,710,579 & 32,140,658 &0&0&10,800& OOM & --&-- & OOM  & -- &-- &
\textbf{168,431,056} &3,258,831 & 904
& 401,307 
 &469\\
uk2002 & 18,483,186 & 65,451,024 & 196,336,234 &0&0&100,000& OOM & --&-- & OOM  & --&-- & 
\textbf{687,089,901} &13,716,043 & 7,678 & 4,806,660  &6,243\\ 
\hline
\multicolumn{4}{|c|}{\# of best / \# of instances}
& \textbf{11}/15 & -- &-- & 3/15 & -- &-- &3/15  &--  & -- & \textbf{11}/15 &--    & -- &  -- &--\\ \hline
\multicolumn{4}{|c|}{average running time(s)} &--&--&10,570.1&-- &-- &1,256  &  -- &-- & 15.7 & -- &--  & 818.7  & --  & 617.5\\ 
\hline
\end{tabular}%
}
\end{table}

Table~\ref{tab: Set2} shows that the RLS outperforms other methods on all SET-2 instances in terms of both best-known objective value and runtime. 
It is known that the RLS proves optimality if all vertices are reduced (i.e., $|V_{kernel}|=0$), hence 11 out of 47 instances have been solved to optimality by the RLS. In contrast, LA-B\&B can only solve 1 out of these 47 instances. { Similarly, CPLEX fails to solve instances with 40,000 variables (e.g., vc-exact\_024 and C1000.9). Moreover, CPLEX even fails to report a feasible solution for 4 large instances.}
In contrast, the RLS and ALS hit the BKV with an average hitting time of no more than 8 seconds. 

Table~\ref{tab:set3} compares algorithms on SET-3 instances. 
OOM indicates that the algorithm runs out of memory.
Both ALS and LA-B\&B quickly become infeasible as instance size increases. 
CPLEX again fails to report a feasible solution for large instances.
In contrast, the RLS computes feasible solutions even for large instances such as {\em uk2002} . For instances where the benchmarks are feasible, the RLS consistently achieves the best net benefit. For the infeasible instance {\em Email}, the RLS even finds the optimal solution.

\subsection{Evaluation of Reduction Rules} \label{exp: reduction rules}
To evaluate the effectiveness of the reduction rules, we conducted the following two experiments.

\begin{table}[ht]
\caption{Comparative results of algorithms with or without reduction rules}
\label{tab: ablation reduction rules}
\centering
\resizebox{0.8\textwidth}{!}{ 
\scriptsize
\begin{tabular}{|ccc|cc|cc|cccc|cccc|}
\hline
\multirow{3}{*}{Instance} & \multirow{3}{*}{$|V|$} & \multirow{3}{*}{$|V_{ker}|$} &\multicolumn{4}{c|}{ CPLEX} & \multicolumn{4}{c|}{LA-B\&B} & \multicolumn{4}{c|}{RLS} \\ \cline{4-15} 
 &  & & \multicolumn{2}{c|}{without reduction} & \multicolumn{2}{c|}{with reduction}  & \multicolumn{2}{c|}{without reduction} & \multicolumn{2}{c|}{with reduction} & \multicolumn{2}{c|}{without reduction} & \multicolumn{2}{c|}{with reduction} \\ \cline{4-15} 
 &  & & BKV & \multicolumn{1}{c|}{htime(s)} & BKV & htime(s) & BKV & \multicolumn{1}{c|}{ htime(s)} & BKV & { htime(s)} & BKV  & \multicolumn{1}{c|}{htime(s)} & BKV  & htime(s) \\ \hline
bio-yeast & 1,458 & 0 &\textbf{68,574}$^*$& $<$1& \textbf{68,574}$^*$ & $<$1 & 66,545 & \multicolumn{1}{c|}{2} & \textbf{68,574}$^*$ & $<$1 & \textbf{68,574}$^*$ & \multicolumn{1}{c|}{1} & \textbf{68,574}$^*$  & {$<$1} 
\\
soc-wiki-Vote & 889 & 0&\textbf{37,358}$^*$&$<1$&\textbf{37,358}$^*$ & $<$1 & 35,973 & \multicolumn{1}{c|}{$<$1} & \textbf{37,358}$^*$ & $<$1 & \textbf{37,358}$^*$  & \multicolumn{1}{c|}{1} & \textbf{37,358}$^*$  & {$<$1} 
\\
vc-exact\_131 & 2,980 & 0 &\textbf{109,641}$^*$&$<1$&\textbf{109,641}$^*$ & $<$1 &107,812 & \multicolumn{1}{c|}{14} & \textbf{109,641}$^*$ & $<$1 & 109,567  & \multicolumn{1}{c|}{9} & \textbf{109,641}$^*$  & {$<$1} 
\\
web-edu & 3,031 & 98 &\textbf{114,140}$^*$&$<1$&\textbf{114,140}$^*$&$<1$& 109,856 & \multicolumn{1}{c|}{10} & \textbf{114,140}$^*$ & 21 & 114,080  & \multicolumn{1}{c|}{12} & \textbf{114,140}$^*$  & $<$1 
\\
MANN\_a81 & 3,321 & 0 &\textbf{117,327}$^*$&$<1$&\textbf{117,327}$^*$ & $<$1& 114,616 & \multicolumn{1}{c|}{16} & \textbf{117,327}$^*$ & $<$1 & 117,275  & \multicolumn{1}{c|}{14} & \textbf{117,327}$^*$  & {$<$1} 
\\
tech-routers-rf & 2,113 & 0 &\textbf{97,799}$^*$&$<1$&\textbf{97,799}$^*$ & $<$1& 93,690 & \multicolumn{1}{c|}{54} & \textbf{97,799}$^*$ & $<$1 & \textbf{97,799}$^*$  & \multicolumn{1}{c|}{1} & \textbf{97,799}$^*$  & {$<$1} 
\\
vc-exact\_039 & 6,795 & 0 &\textbf{259,542}$^*$&$<1$&\textbf{259,542}$^*$ & $<$1& 254,941 & \multicolumn{1}{c|}{253} & \textbf{259,542}$^*$ & $<$1 & 259,475  & \multicolumn{1}{c|}{13} & \textbf{259,542}$^*$  & {$<$1} 
\\
ca-GrQc & 4,158 & 282 & \textbf{174,782}$^*$&9& \textbf{174,782}$^*$ &6& 167,846 & \multicolumn{1}{c|}{45} & 174,525 & 2,998 & 174,768  & \multicolumn{1}{c|}{7} & \textbf{174,782}$^*$  & 11
\\
vc-exact\_038 & 786 & 697 &\textbf{12,015}$^*$&592&\textbf{12,015}$^*$&585& 10,406 & \multicolumn{1}{c|}{1,994} & 10,782 & 28 & \textbf{12,015}$^*$  & \multicolumn{1}{c|}{13} & \textbf{12,015}$^*$  & 4
\\
vc-exact\_078 & 11,349 & 0 &\textbf{438,671}$^*$&$<1$& \textbf{438,671}$^*$ & 1& 431,371 & \multicolumn{1}{c|}{1,074} & \textbf{438,671}$^*$ & 1 & 438,306  & \multicolumn{1}{c|}{16} & \textbf{438,671}$^*$  & {$<$1} 
\\
vc-exact\_087 & 13,590 & 0 &\textbf{518,500}$^*$&2&\textbf{518,500}$^*$ & $<$1& 509,734 & \multicolumn{1}{c|}{2,319} & \textbf{518,500}$^*$ & $<$1 & 518,069  & \multicolumn{1}{c|}{15} & \textbf{518,500}$^*$  & {$<$1} 
\\
vc-exact\_107 & 6,402 & 0 &\textbf{407,878}$^*$ &$<1$&\textbf{407,878}$^*$ & $<$1&396,456 & \multicolumn{1}{c|}{10,646} & \textbf{407,878}$^*$ & $<$1 & 406,379  & \multicolumn{1}{c|}{14} & \textbf{407,878}$^*$  & {$<$1} 
\\
vc-exact\_151 & 15,783 & 0 &\textbf{530,138}$^*$ &2&\textbf{530,138}$^*$ & $<$1& 516,445 & \multicolumn{1}{c|}{2,402} & \textbf{530,138}$^*$ & $<$1 & 529,003  & \multicolumn{1}{c|}{18} & \textbf{530,138}$^*$  & {$<$1} 
\\
vc-exact\_167 & 15,783 & 0 &\textbf{465,178}$^*$&1&\textbf{465,178}$^*$ & $<$1& 449,579 & \multicolumn{1}{c|}{10,469} & \textbf{465,178}$^*$ & $<$1 & 463,446  & \multicolumn{1}{c|}{12} & \textbf{465,178}$^*$  & {$<$1} 
\\
bio-dmela & 7,393 & 44 & \textbf{302,992}$^*$&$<1$&\textbf{302,992}$^*$&$<1$&280,897 & \multicolumn{1}{c|}{170} & \textbf{302,992}$^*$ & $<1$ & 302,793  & \multicolumn{1}{c|}{18} & \textbf{302,992}$^*$  & $<$1 
\\
vc-exact\_011 & 9,877 & 172&\textbf{303,793}$^*$&1&\textbf{303,793}$^*$&$<1$ & 286,137 & \multicolumn{1}{c|}{4,818} & 303,589 & 1271 & 303,352 & \multicolumn{1}{c|}{9} & \textbf{303,793}$^*$  & $<$1 
\\
web-spam & 4,767 & 487 &\textbf{190,507}$^*$&5&\textbf{190,507}$^*$ &$<1$& 180,446 & \multicolumn{1}{c|}{53} & 189,897 & $<1$ & 190,460  & \multicolumn{1}{c|}{18} & 190,460  & 13
\\
vc-exact\_026 & 6,140 & 300 &\textbf{201,509}$^*$&1&\textbf{201,509}$^*$&$<1$& 177,902 & \multicolumn{1}{c|}{668} & 200,868 & $<1$ & 201,458  & \multicolumn{1}{c|}{13} & 201,503  & $<$1 
\\
vc-exact\_001 & 6,160 & 134 &\textbf{215,941}$^*$&2&\textbf{215,941}$^*$&$<1$& 194,984 & \multicolumn{1}{c|}{9,575} & \textbf{215,941}$^*$ & 2 & 215,934  & \multicolumn{1}{c|}{13} & \textbf{215,941}$^*$  & $<$1 
\\
vc-exact\_024 & 7,620 & 1,255 &\textbf{236,109}&10,800&\textbf{236,109}&10,800& 226,540 & \multicolumn{1}{c|}{244} & 234,933 & $1$ & 235,911  & \multicolumn{1}{c|}{14} & 236,018  & 12
\\
C1000.9 & 1,000 & 1,000 &5,263&10,800&5,263&10,800& 6,104 & \multicolumn{1}{c|}{430} & 6,104 & 430 & 7,235  & \multicolumn{1}{c|}{16} & \textbf{7,237} & 10
\\
tech-WHOIS & 7,476 & 338 &\textbf{325,584}$^*$&5&\textbf{325,584}$^*$&$<1$& 309,903 & \multicolumn{1}{c|}{218} & 325,119 & 255 & 325,480  & \multicolumn{1}{c|}{16} & 325,581  & 13
\\
p-hat700-3 & 700 & 700 &4,025&10,800&4,025&10,800& 3,030 & \multicolumn{1}{c|}{7,436} & 3,030 & 7,436 & \textbf{4,473}  & \multicolumn{1}{c|}{1} & \textbf{4,473}  & $<$1 
\\
vc-exact\_008 & 7,537 & 574 &\textbf{247,091}$^*$ &3&\textbf{247,091}$^*$&$<1$& 210,126 & \multicolumn{1}{c|}{9,093} & 245,289 & 6,627 & 247,036  & \multicolumn{1}{c|}{16} & \textbf{247,091}$^*$  & $<$1 
\\
keller5 & 776 & 776 &2,373&10,800&2,373&10,800& 2,860 & \multicolumn{1}{c|}{6,327} & 2,860 & 6,327 & \textbf{3,877}  & \multicolumn{1}{c|}{2} & \textbf{3,877}  & 2
\\
vc-exact\_194 & 1,150 & 1,150 &3,734&10,800&3,734&10,800& 4,386 & \multicolumn{1}{c|}{1,324} & 4,386 & 1,324 & \textbf{5,419}  & \multicolumn{1}{c|}{2} & \textbf{5,419}  & $<$1 
\\
hamming10-4 & 1,024 & 1,024 &3,256&10,800&3,256&10,800& 3,073& \multicolumn{1}{c|}{2,069} & 3,073 & 2,069 & \textbf{3,660}  & \multicolumn{1}{c|}{1} & \textbf{3,660}  & $<$1 
\\
brock800\_2 & 800 & 800 &1,177&10,800&1,177&10,800& 1,996 & \multicolumn{1}{c|}{10} & 1,996 & 10 & \textbf{2,473}  & \multicolumn{1}{c|}{1} & \textbf{2,473}  & $<$1 
\\
brock800\_4 & 800 & 800 &1,583&10,800&1,583&10,800& 1,979 & \multicolumn{1}{c|}{7,215} & 1,979 & 7,215 & \textbf{2,501}  & \multicolumn{1}{c|}{1} & \textbf{2,501}  & $<$1 
\\
ca-HepPh & 11,204 & 2,165 &372,096&10,800&\textbf{372,229}&10,800& 348,738 & \multicolumn{1}{c|}{422} & 368,637& 2 & 372,079  & \multicolumn{1}{c|}{14} & 372,079  & 17
\\
p-hat700-2 & 700 & 700 &2,870&10,800&2,870&10,800& 2,345 & \multicolumn{1}{c|}{2,078} & 2,345 & 2,078 & \textbf{2,971}  & \multicolumn{1}{c|}{1} & \textbf{2,971}  & $<$1 
\\
vc-exact\_196 & 1,534 & 1,534 &3,916&10,800&3,916&10,800& 4,273 & \multicolumn{1}{c|}{66} & 4,273 & 66 & \textbf{5,321}  & \multicolumn{1}{c|}{2} & \textbf{5,321}  & $<$1 
\\
p-hat700-1 & 700 & 700 &730&10,800&730&10,800& \textbf{844}$^*$ & \multicolumn{1}{c|}{5} & \textbf{844}$^*$ & 5 & \textbf{844}$^*$ & \multicolumn{1}{c|}{1} & \textbf{844}$^*$  & $<$1 
\\
ca-AstroPh & 17,903 & 4,985&604,288&10,800&\textbf{604,342}&10,800 & 558,736 & \multicolumn{1}{c|}{2,544} & 593,899 & 37 & 603,732  & \multicolumn{1}{c|}{15} & 603,945 & 14
\\
C2000.9 & 2,000 & 2,000 &3,958&10,800&3,958&10,800& 5,477 & \multicolumn{1}{c|}{11} & 5,477 & 11 & \textbf{6,396}  & \multicolumn{1}{c|}{3} & \textbf{6,396}  & 3
\\
DSJC1000.5 & 1,000 & 1,000 &925&10,800&925&10,800& 1,515 & \multicolumn{1}{c|}{694} & 1,515 & 694 & \textbf{1,861}  & \multicolumn{1}{c|}{1} & \textbf{1,861}  & $<$1 
\\
socfb-MIT & 6,402 & 5,494 &7,098&10,800&118,114&10,800& 110,677 & \multicolumn{1}{c|}{20} & 113,623 & 10 & 129,831  & \multicolumn{1}{c|}{16} & \textbf{129,872}  & 16
\\
p-hat1500-3 & 1,500 & 1,500 &4,097&10,800&4,097&10,800& 3,877 & \multicolumn{1}{c|}{704} & 3,877 & 704 & 7,500 & \multicolumn{1}{c|}{15} & \textbf{7,501}  & 16
\\
socfb-UCSB37 & 14,917 & 13,655&16,708&10,800&57,084& 10,800 & 219,965 & \multicolumn{1}{c|}{256} & 223,780 & 150 & 255,887  & \multicolumn{1}{c|}{21} & \textbf{255,898}  & 22
\\
socfb-Duke14 & 9,885 & 9,003 &11,565&10,800&127,145&10,800& 125,940 & \multicolumn{1}{c|}{39} & 127,269 & 2,231 & 151,321  & \multicolumn{1}{c|}{14} & \textbf{151,326}  & 18
\\
p-hat1500-2 & 1,500 & 1,500 &3,542&10,800&3,542&10,800& 2,611 & \multicolumn{1}{c|}{2,150} & 2,611 & 2,150 & \textbf{4,253}  & \multicolumn{1}{c|}{1} & \textbf{4,253}  & $<$1 
\\
socfb-Stanford3 & 11,586 & 9,380&7,250&10,800&127,708&10,800 & 215,381 & \multicolumn{1}{c|}{265} & 223,398 & 7,756 & 251,788  & \multicolumn{1}{c|}{19} & \textbf{251,833}  & 14
\\
socfb-UConn & 17,206 & 16,025 &16,996&10,800&51,328&10,800& 231,718 & \multicolumn{1}{c|}{322} & 235,104 & 240 & 275,006  & \multicolumn{1}{c|}{16} & \textbf{275,034}  & 20
\\
p-hat1500-1 & 1,500 & 1,500 &0&10,800&0&10,800& 1,184 & \multicolumn{1}{c|}{5,499} & 1,184 & 5,499 & \textbf{1,287}  & \multicolumn{1}{c|}{2} & \textbf{1,287}  & 1
\\
C2000.5 & 2,000 & 2,000 &0&10,800&0&10,800& 1,581 & \multicolumn{1}{c|}{865} & 1,581 & 865 & \textbf{1,849}  & \multicolumn{1}{c|}{5} & \textbf{1,849}  & 6
\\
keller6 & 3,361 & 3,361 &0&10,800&0&10,800& 3,545 & \multicolumn{1}{c|}{1,561} & 3,545 & 1,561 & 5,229  & \multicolumn{1}{c|}{19} & \textbf{5,236}  & 16
\\
C4000.5 & 4,000 & 4,000 &0&10,800&0&10,800& 1,611 & \multicolumn{1}{c|}{1,659} & 1,611 & 1,659 & \textbf{1,951}  & \multicolumn{1}{c|}{15} & \textbf{1,951}  & 22\\ \hline
foodweb-wet & 128 & 125 &\textbf{2,464}$^*$ &$<1$&\textbf{2,464}$^*$ &$<1$&\textbf{2,464}$^*$ & \multicolumn{1}{c|}{8} & \textbf{2,464}$^*$ & 6& \textbf{2,464}$^*$ & \multicolumn{1}{c|}{$<$1} & \textbf{2,464}$^*$  & $<$1  \\
foodweb-dry & 128 & 125 &\textbf{2,226}$^*$&$<1$&\textbf{2,226}$^*$&$<1$& \textbf{2,226}$^*$ & \multicolumn{1}{c|}{8} & \textbf{2,226}$^*$ & 8  & \textbf{2,226}$^*$  & \multicolumn{1}{c|}{$<$1} & \textbf{2,226}$^*$  & $<$1\\
USAir97 & 332 & 0 &\textbf{10,308}$^*$&$<1$&\textbf{10,308}$^*$ & $<1$& 9,507 & \multicolumn{1}{c|}{407} & \textbf{10,308}$^*$ & $<1$ & \textbf{10,308}$^*$  & \multicolumn{1}{c|}{$<$1} & \textbf{10,308}$^*$  & $<$1\\
powergrid & 4,941 & 0 &\textbf{193,339}$^*$&$<1$&\textbf{193,339}$^*$ & $<1$& 190,537 & \multicolumn{1}{c|}{108} & \textbf{193,339}$^*$ & $<1$ & 193,194  & \multicolumn{1}{c|}{52} & \textbf{193,339}$^*$  &$<$1 \\
CondMat & 23,133 & 1,648 &\textbf{711,517}$^*$&12&\textbf{711,517}$^*$&$<1$& 685,442 & \multicolumn{1}{c|}{5,749} & 708,907 & 84 & 709,927  & \multicolumn{1}{c|}{54} & 711,334  & 45 \\
Email & 265,009 & 0 &\textbf{13,054,427}$^*$&118&\textbf{13,054,427}$^*$  & 1& OOM & \multicolumn{1}{c|}{--} & \textbf{13,054,427}$^*$ & 1 & 13,054,390  & \multicolumn{1}{c|}{69} & \textbf{13,054,427}$^*$  & 1 \\
Epinion & 75,879 & 253 &\textbf{3,052,059}$^*$&60&\textbf{3,052,059}$^*$&1& OOM 
& \multicolumn{1}{c|}{--} & 3,051,256 & 196 & 3,050,415  & \multicolumn{1}{c|}{67} & \textbf{3,052,059}$^*$  & 1 \\
Dblp & 317,080 & 11,851 &\textbf{9,737,911}$^*$&790&\textbf{9,737,911}$^*$&39& OOM & \multicolumn{1}{c|}{--} & 9,715,860 & 307 & 9,577,609  & \multicolumn{1}{c|}{46} & 9,734,653  & 64 \\
cnr-2000 & 325,557 & 59,052 &13,367,460&10,800&\textbf{13,641,278}&10,800& OOM & \multicolumn{1}{c|}{--} & 11,347,241 & 10,796 & 13,361,167  & \multicolumn{1}{c|}{84} & 13,368,670  & 482 \\
WikiTalk & 2,394,385 
& 8&\textbf{118,708,326}$^*$&3,566&\textbf{118,708,326}$^*$  & 333 & OOM & \multicolumn{1}{c|}{--} & \textbf{118,708,326}$^*$ & 333 & 118,672,055  & \multicolumn{1}{c|}{56} & \textbf{118,708,326}$^*$  & 333 \\
BerkStan & 685,230 
& 114,191 &26,036,210&10,800&\textbf{26,055,767}&10,800& OOM & \multicolumn{1}{c|}{--} & 22,591,394 & 1,031 & 25,813,846  & \multicolumn{1}{c|}{200} & 26,031,704 & 1,112 \\
As-Skitter & 1,696,415 & 31,173 &66,603,360&10,800&\textbf{66,830,857}&10,800& OOM & \multicolumn{1}{c|}{--} & 65,907,243 & 10,512 & 59,886,304  & \multicolumn{1}{c|}{200} & 66,582,254 & 723 \\
soc-pokec & 1,632,803 & 881,319 &4,146,172&10,800& 27,037,851&10,800& OOM & \multicolumn{1}{c|}{--} & OOM & -- & 36,540,838  & \multicolumn{1}{c|}{350} & \textbf{43,999,402}  & 934 \\
LiveJ & 4,846,609 & 401,307 &0&10,800&163,061,840&10,800& OOM & \multicolumn{1}{c|}{--} & OOM & -- & 148,758,359  & \multicolumn{1}{c|}{318} & \textbf{168,431,056}  & 904 \\
uk2002 & 18,483,186 & 4,806,660 &0&100,000&567,666,923&100,000& OOM & \multicolumn{1}{c|}{--} & OOM & -- & 633,365,020  & \multicolumn{1}{c|}{1,806} & \textbf{687,089,901}  & 7,678 \\ \hline
\multicolumn{3}{|c|}{\# of best / \# of instances} &31/62&--&36/62&-- & 3/62 & \multicolumn{1}{c|}{--} & 21/62 & -- & {22}/62  & \multicolumn{1}{c|}{--} & \textbf{51}/62  & -- \\ \hline
\multicolumn{3}{|c|}{average running time(s)} &--&7,096.5&--&7,028.8 &-- & \multicolumn{1}{c|}{2,007.5} & --& 1,441.1  &  --  & \multicolumn{1}{c|}{\textbf{61.3}} & --  & {202.6} \\ \hline
\end{tabular}
}
\end{table}

\noindent
{\bf Comparative results with vs. without reduction.}
This experiment evaluates algorithms with and without data reduction. For LA-B\&B and CPLEX with the reductions, we used kernel graphs as input. For the RLS without reduction, all reductions were skipped, and random initialization was applied. ALS cannot handle kernel graphs due to unsupported negative vertex profits and edge penalties.



The comparative results are shown in Table~\ref{tab: ablation reduction rules}. In general, the reduction rules can help all algorithms to find better solutions. 
Specifically, with the reduction techniques, LA-B\&B is able to solve some instances which it cannot solve without reduction, e.g., {\em WikiTalk} (1.6 million vertices, 11.1 million edges).
For some large instances where original LA-B\&B run out-of-memory,  the new LA-B\&B with the reductions can report a feasible solution or solve them to optimality within the cutoff time.
With the reduction techniques, CPLEX can solve some instances with less time. For two large instances {\em LiveJ} and {\em uk2002} where CPLEX fails, CPLEX with the reductions is able to report feasible solutions.
As for the RLS, with these reduction rules, it consistently finds better results across all instances in a shorter time than without them.

\begin{figure}[ht!]
    \centering
    \caption{Variation of kernel sizes when ablating different reductions. Each cell shows the difference in kernel size between { not using and using} the corresponding reductions; higher values indicate larger effects.}
    \includegraphics[width=0.9\linewidth]{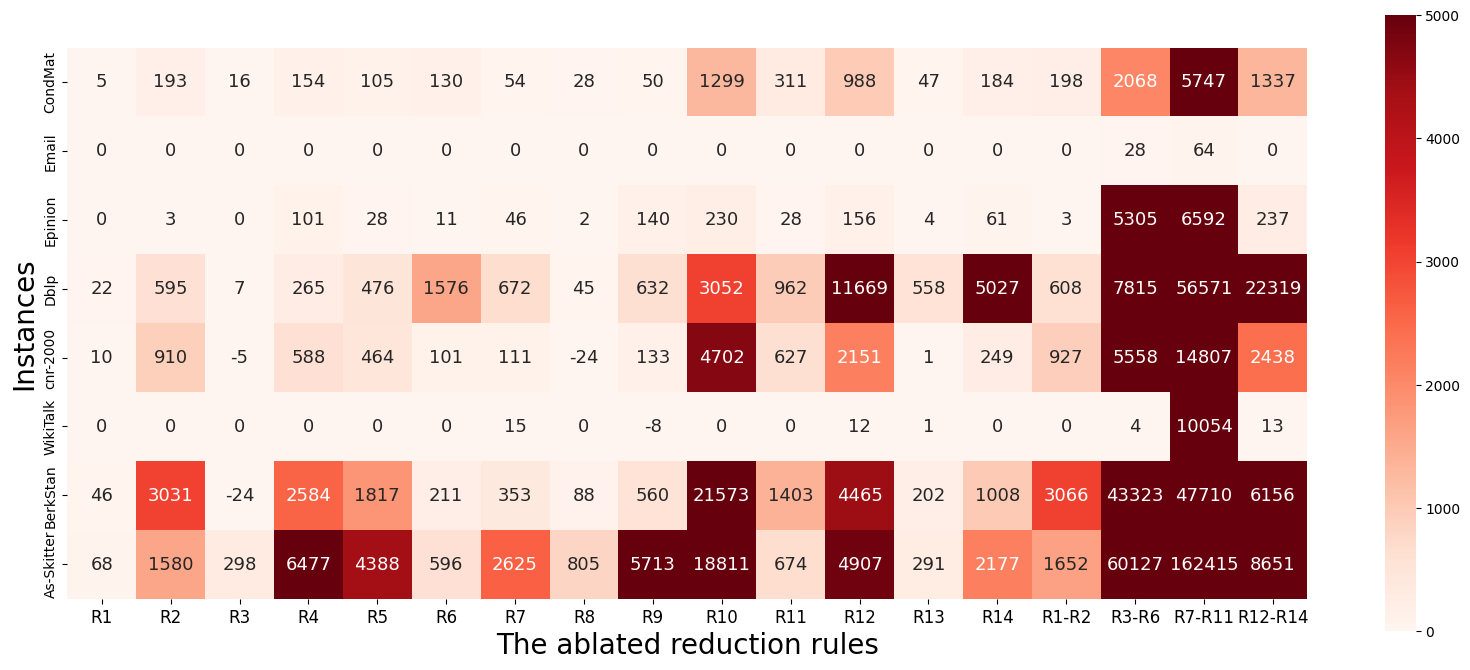}
    \label{fig:reduction heatmap}
\end{figure}

\begin{figure}[ht!]
    \centering
    \caption{Variation of pre-processing time for getting a kernel. Each cell shows the difference in time (seconds) between { not using and using} the corresponding reductions; higher values indicate larger effects.}
    \includegraphics[width=0.9\linewidth]{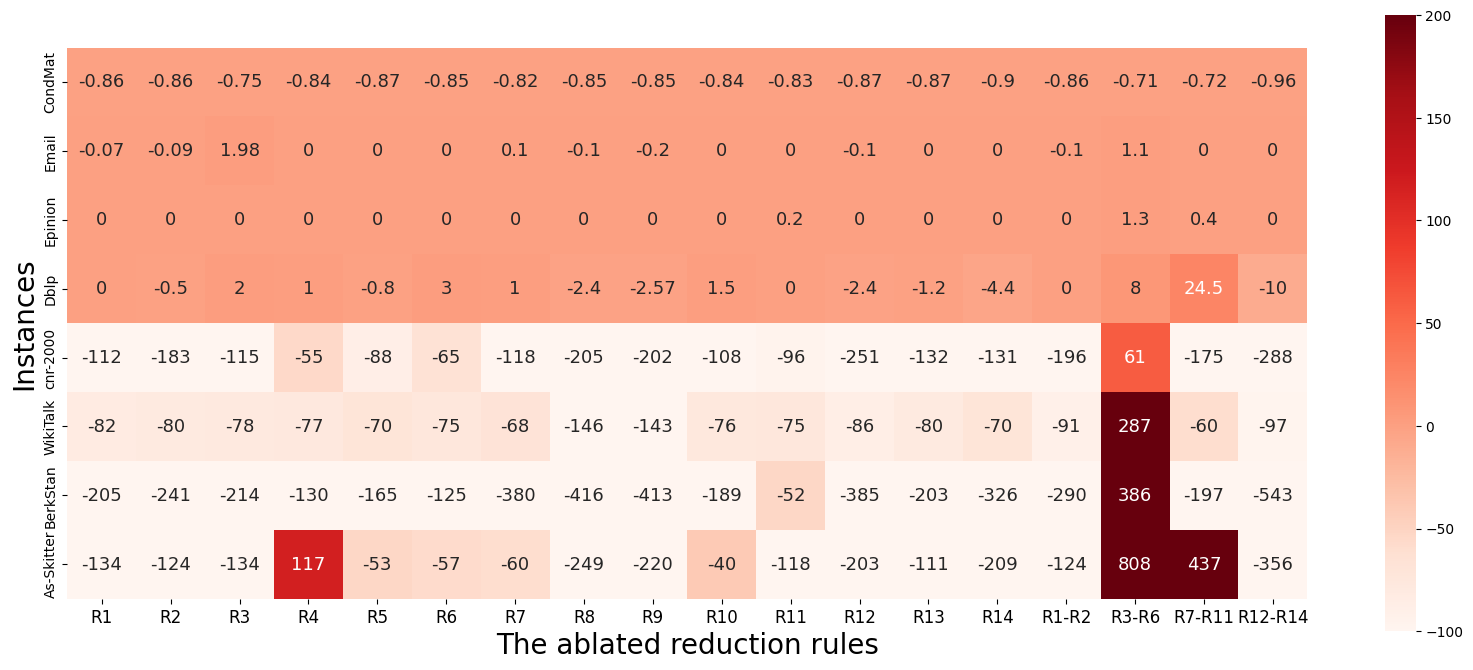}
    
    \label{fig:reduction heatmap time}
\end{figure}


\smallskip
\noindent
{\bf Ablation study of reduction rules.}
In this experiment, we compare kernel sizes (measured by the number of vertices in the graph) and the pre-processing time (i.e., the runtime of Algorithm~\ref{alg:apply_reduction}) after ablating different reduction rules.
Normally, a smaller kernel indicates a smaller search space.
Experimental results are shown in Figures~\ref{fig:reduction heatmap} and \ref{fig:reduction heatmap time}.

Figure~\ref{fig:reduction heatmap} shows that the impacts on the kernel size vary among these reduction rules. 
When comparing individual reductions,  $R10$ and $R12$ are more powerful than the other reduction rules in decreasing the kernel size.
When comparing a group of reductions, the degree-based rules ($R7$–$R11$) are the most powerful in decreasing the kernel size.
The neighborhood-based rules ($R3$–$R6$) are generally more effective than the vertex-pair reduction rules ($R12$–$R14$) except for \textit{Dblp} and \textit{WikiTalk}.
Additionally, although $R1$ and $R2$ do not involve vertex removals, they are also powerful in decreasing the kernel size. In particular, incorporating $R2$ into the RLS can further reduce more than 3,000 vertices for the \textit{BerkStan} instance.

Intuitively, ablating a reduction rule can decrease the overall pre-processing time, but may at the cost of resulting in a larger kernel. 
However, Figure~\ref{fig:reduction heatmap time} shows that ablating reduction rules $R3$–$R6$ or $R7$–$R11$ can even increase the pre-processing time. 
This occurs because the remaining reductions become more time-consuming in the absence of these rules.
We also observe that for the four larger instances {\em As-Skitter}, {\em Berkstan}, {\em WikiTalk} and {\em cnr-2000},  the inclusion of reductions shows a significant improvement over the pre-processing time.

\section{Conclusion and Future Works} \label{sec: conclusion}
In this paper, we proposed a novel approach to the Generalized Independent Set (GIS) problem using data reduction techniques to tackle the complexity of large graphs. 
We developed 14 reduction rules, categorized into edge transformation, neighborhood-based, low-degree, and vertex-pair reductions, and designed a reduction-driven local search (RLS) algorithm based on them. 
Experiments on benchmark instances showed that the RLS algorithm outperformed state-of-the-art methods, especially on large graphs with tens of thousands of vertices.
Future work will focus on developing more advanced reduction techniques to further improve efficiency and extending our algorithms to real-world scenarios such as dynamic graphs with changing weights.
\bibliographystyle{informs2014} 
\bibliography{sample}

\begin{thebibliography}{33}
\providecommand{\natexlab}[1]{#1}
\providecommand{\url}[1]{\texttt{#1}}
\providecommand{\urlprefix}{URL }

\bibitem[{Abu-Khzam et~al.(2022)Abu-Khzam, Lamm, Mnich, Noe, Schulz,
  \protect\BIBand{} Strash}]{abu2022recent}
Abu-Khzam FN, Lamm S, Mnich M, Noe A, Schulz C, Strash D (2022) Recent advances
  in practical data reduction. \emph{Algorithms for Big Data} 97--133.

\bibitem[{Akiba \protect\BIBand{} Iwata(2016)}]{akiba2016branch}
Akiba T, Iwata Y (2016) Branch-and-reduce exponential/fpt algorithms in
  practice: A case study of vertex cover. \emph{Theoretical Computer Science}
  609:211--225.

\bibitem[{Brendel et~al.(2011)Brendel, Amer, \protect\BIBand{}
  Todorovic}]{brendel2011multiobject}
Brendel W, Amer M, Todorovic S (2011) Multiobject tracking as maximum weight
  independent set. \emph{CVPR 2011}, 1273--1280 (IEEE).

\bibitem[{Brendel \protect\BIBand{} Todorovic(2010)}]{brendel2010segmentation}
Brendel W, Todorovic S (2010) Segmentation as maximum-weight independent set.
  \emph{Advances in neural information processing systems} 23.

\bibitem[{Chang et~al.(2017)Chang, Li, \protect\BIBand{}
  Zhang}]{chang2017computing}
Chang L, Li W, Zhang W (2017) Computing a near-maximum independent set in
  linear time by reducing-peeling. \emph{Proceedings of the 2017 ACM
  International Conference on Management of Data}, 1181--1196.

\bibitem[{Colombi et~al.(2017)Colombi, Mansini, \protect\BIBand{}
  Savelsbergh}]{colombi2017generalized}
Colombi M, Mansini R, Savelsbergh M (2017) The generalized independent set
  problem: Polyhedral analysis and solution approaches. \emph{European Journal
  of Operational Research} 260(1):41--55.

\bibitem[{Cygan et~al.(2015)Cygan, Fomin, Kowalik, Lokshtanov, Marx, Pilipczuk,
  Pilipczuk, \protect\BIBand{} Saurabh}]{cygan2015parameterized}
Cygan M, Fomin FV, Kowalik {\L}, Lokshtanov D, Marx D, Pilipczuk M, Pilipczuk
  M, Saurabh S (2015) \emph{Parameterized algorithms}, volume~5 (Springer).

\bibitem[{Dahlum et~al.(2016)Dahlum, Lamm, Sanders, Schulz, Strash,
  \protect\BIBand{} Werneck}]{dahlum2016accelerating}
Dahlum J, Lamm S, Sanders P, Schulz C, Strash D, Werneck RF (2016) Accelerating
  local search for the maximum independent set problem. \emph{Experimental
  Algorithms: SEA, Proceedings 15}, 118--133 (Springer).

\bibitem[{Glover \protect\BIBand{} Laguna(1998)}]{glover1998tabu}
Glover F, Laguna M (1998) \emph{Tabu search} (Springer).

\bibitem[{Gro{\ss}mann et~al.(2023)Gro{\ss}mann, Lamm, Schulz,
  \protect\BIBand{} Strash}]{grossmann2023finding}
Gro{\ss}mann E, Lamm S, Schulz C, Strash D (2023) Finding near-optimal weight
  independent sets at scale. \emph{Proceedings of the Genetic and Evolutionary
  Computation Conference}, 293--302.

\bibitem[{Gschwind et~al.(2021)Gschwind, Irnich, Furini, \protect\BIBand{}
  Calvo}]{gschwind2021branch}
Gschwind T, Irnich S, Furini F, Calvo RW (2021) A branch-and-price framework
  for decomposing graphs into relaxed cliques. \emph{INFORMS Journal on
  Computing} 33(3):1070--1090.

\bibitem[{Gu et~al.(2021)Gu, Zheng, Cai, \protect\BIBand{}
  Peng}]{gu2021towards}
Gu J, Zheng W, Cai Y, Peng P (2021) Towards computing a near-maximum weighted
  independent set on massive graphs. \emph{Proceedings of the 27th ACM SIGKDD
  Conference on Knowledge Discovery \& Data Mining}, 467--477.

\bibitem[{Hochbaum(2004)}]{hochbaum200450th}
Hochbaum DS (2004) 50th anniversary article: Selection, provisioning, shared
  fixed costs, maximum closure, and implications on algorithmic methods today.
  \emph{Management Science} 50(6):709--723.

\bibitem[{Hochbaum \protect\BIBand{} Pathria(1997)}]{hochbaum1997forest}
Hochbaum DS, Pathria A (1997) Forest harvesting and minimum cuts: a new
  approach to handling spatial constraints. \emph{Forest Science}
  43(4):544--554.

\bibitem[{Hosseinian \protect\BIBand{}
  Butenko(2019)}]{hosseinian2019algorithms}
Hosseinian S, Butenko S (2019) Algorithms for the generalized independent set
  problem based on a quadratic optimization approach. \emph{Optimization
  Letters} 13(6):1211--1222.

\bibitem[{Jin \protect\BIBand{} Hao(2015)}]{jin2015general}
Jin Y, Hao JK (2015) General swap-based multiple neighborhood tabu search for
  the maximum independent set problem. \emph{Engineering Applications of
  Artificial Intelligence} 37:20--33.

\bibitem[{Koana et~al.(2021)Koana, Korenwein, Nichterlein, Niedermeier,
  \protect\BIBand{} Zschoche}]{koana2021data}
Koana T, Korenwein V, Nichterlein A, Niedermeier R, Zschoche P (2021) Data
  reduction for maximum matching on real-world graphs: Theory and experiments.
  \emph{Journal of Experimental Algorithmics (JEA)} 26:1--30.

\bibitem[{Kochenberger et~al.(2007)Kochenberger, Alidaee, Glover,
  \protect\BIBand{} Wang}]{kochenberger2007effective}
Kochenberger G, Alidaee B, Glover F, Wang H (2007) An effective modeling and
  solution approach for the generalized independent set problem.
  \emph{Optimization Letters} 1:111--117.

\bibitem[{Lamm et~al.(2019)Lamm, Schulz, Strash, Williger, \protect\BIBand{}
  Zhang}]{lamm2019exactly}
Lamm S, Schulz C, Strash D, Williger R, Zhang H (2019) Exactly solving the
  maximum weight independent set problem on large real-world graphs.
  \emph{ALENEX}, 144--158 (SIAM).

\bibitem[{L{\'o}pez-Ib{\'a}{\~n}ez et~al.(2016)L{\'o}pez-Ib{\'a}{\~n}ez,
  Dubois-Lacoste, C{\'a}ceres, Birattari, \protect\BIBand{}
  St{\"u}tzle}]{lopez2016irace}
L{\'o}pez-Ib{\'a}{\~n}ez M, Dubois-Lacoste J, C{\'a}ceres LP, Birattari M,
  St{\"u}tzle T (2016) The irace package: Iterated racing for automatic
  algorithm configuration. \emph{Operations Research Perspectives} 3:43--58.

\bibitem[{Louren{\c{c}}o et~al.(2019)Louren{\c{c}}o, Martin, \protect\BIBand{}
  St{\"u}tzle}]{lourencco2019iterated}
Louren{\c{c}}o HR, Martin OC, St{\"u}tzle T (2019) Iterated local search:
  Framework and applications. \emph{Handbook of metaheuristics} 129--168.

\bibitem[{Mauri et~al.(2010)Mauri, Ribeiro, \protect\BIBand{}
  Lorena}]{mauri2010new}
Mauri GR, Ribeiro GM, Lorena LA (2010) A new mathematical model and a
  lagrangean decomposition for the point-feature cartographic label placement
  problem. \emph{Computers \& Operations Research} 37(12):2164--2172.

\bibitem[{Nogueira et~al.(2021)Nogueira, Pinheiro, \protect\BIBand{}
  Tavares}]{nogueira2021iterated}
Nogueira B, Pinheiro RG, Tavares E (2021) Iterated local search for the
  generalized independent set problem. \emph{Optimization Letters}
  15:1345--1369.

\bibitem[{Nouranizadeh et~al.(2021)Nouranizadeh, Matinkia, Rahmati,
  \protect\BIBand{} Safabakhsh}]{nouranizadeh2021maximum}
Nouranizadeh A, Matinkia M, Rahmati M, Safabakhsh R (2021) Maximum entropy
  weighted independent set pooling for graph neural networks. \emph{arXiv
  preprint arXiv:2107.01410} .

\bibitem[{Pavlik et~al.(2021)Pavlik, Ludden, Jacobson, \protect\BIBand{}
  Sewell}]{pavlik2021airplane}
Pavlik JA, Ludden IG, Jacobson SH, Sewell EC (2021) Airplane seating assignment
  problem. \emph{Service Science} 13:1--18.

\bibitem[{Puthal et~al.(2015)Puthal, Nepal, Paris, Ranjan, \protect\BIBand{}
  Chen}]{puthal2015efficient}
Puthal D, Nepal S, Paris C, Ranjan R, Chen J (2015) Efficient algorithms for
  social network coverage and reach. \emph{2015 IEEE International Congress on
  Big Data}, 467--474 (IEEE).

\bibitem[{Sanchis \protect\BIBand{} Jagota(1996)}]{sanchis1996some}
Sanchis LA, Jagota A (1996) Some experimental and theoretical results on test
  case generators for the maximum clique problem. \emph{INFORMS Journal on
  Computing} 8(2):87--102.

\bibitem[{Verma et~al.(2015)Verma, Buchanan, \protect\BIBand{}
  Butenko}]{verma2015solving}
Verma A, Buchanan A, Butenko S (2015) Solving the maximum clique and vertex
  coloring problems on very large sparse networks. \emph{INFORMS Journal on
  computing} 27(1):164--177.

\bibitem[{Walteros \protect\BIBand{} Buchanan(2020)}]{walteros2020maximum}
Walteros JL, Buchanan A (2020) Why is maximum clique often easy in practice?
  \emph{Operations Research} 68(6):1866--1895.

\bibitem[{Wu et~al.(2023)Wu, Li, Wang, Hu, Zhao, \protect\BIBand{}
  Yin}]{wu2023solving}
Wu J, Li CM, Wang L, Hu S, Zhao P, Yin M (2023) On solving simplified
  diversified top-k s-plex problem. \emph{Computers \& Operations Research}
  153:106187.

\bibitem[{Xiao et~al.(2021)Xiao, Huang, Zhou, \protect\BIBand{}
  Ding}]{xiao2021efficient}
Xiao M, Huang S, Zhou Y, Ding B (2021) Efficient reductions and a fast
  algorithm of maximum weighted independent set. \emph{Proceedings of the Web
  Conference 2021}, 3930--3940.

\bibitem[{Zheng et~al.(2024)Zheng, Hao, \protect\BIBand{} Wu}]{zheng2024exact}
Zheng M, Hao JK, Wu Q (2024) Exact and heuristic solution approaches for the
  generalized independent set problem. \emph{Computers \& Operations Research}
  164:106561.

\bibitem[{Zhou \protect\BIBand{} Hao(2017)}]{zhou2017frequency}
Zhou Y, Hao JK (2017) Frequency-driven tabu search for the maximum s-plex
  problem. \emph{Computers \& Operations Research} 86:65--78.

\end{thebibliography}

\section*{Appendix}
\addcontentsline{toc}{section}{Appendices}
\renewcommand{\thesubsection}{\Alph{subsection}}

\subsection{Technical Proofs} \label{app: proofs}
\begin{proof}[Proof of R\ref{lem: reduction 1}]
Consider whether $u$ and $v$ are included in a generalized independent set $S$. There are { four cases: $u\in S$ and $v\in S$, $u\in S$ and $v\notin S$, $u\notin S$ and $v\in S$, and $u\notin S$ and $v\notin S$.}
When either $u\notin S$ or $v\notin S$, the edge $e=\{u,v\}$ does not exist in $G[S]$, and thus removing $e$ has no effect on the net benefit $nb(S)$.
If both $u\in S$ and $v\in S$, the edge $\{u,v\}$ has a zero penalty, which has no effect on the net benefit $nb(S)$ either.
\end{proof}

\begin{proof}[Proof of R\ref{lem: reduction 2}]
Assume, for the sake of contradiction, that there exists $I\in MGIS(G)$ such that $u\in I$ and $v\in I$. 
{ Assume, without loss of generality, that 
$\Tilde{w}(u)\leq \Tilde{w}(v)$.}
We have $nb(I)-nb(I\setminus \{u\})=w(u)-\sum_{x\in N_r(u)\cap {I}}p(\{u,x\})\leq w(u)-p(\{u,v\})+\sum_{x\in N_r(u)}max(0,-p(\{u,x\}))=\Tilde{w}(u)-p(\{u,v\})< 0$.
In conclusion, $I\setminus \{u\}$ has a larger net benefit without containing permanent edges, which contradicts the optimality of $I$.
Therefore, $u$ and $v$ cannot be simultaneously included in any $S\in MGIS(G)$, which is equivalent to the case where $\{u,v\}\in E_p$.
\end{proof}


\begin{proof}[Proof of R\ref{lem: reduction 3}]
We prove that $u$ is in at least one set in $MGIS(G)$. 
Assume, for the sake of contradiction, that $u$ is not included in any set in $MGIS(G)$.
Let $I$ be an arbitrary set in $MGIS(G)$.
It is clear that at least one vertex in $N(u)$ must be included in $I$. Otherwise, $I\cup \{u\}$ implies with a larger net benefit without including any permanent edge, a contradiction,
since $w(u)\geq w^+(N(u))$, $nb(I\cup \{u\}\setminus N(u))\geq nb(I)$.
So $u$ is in at least one set in $MGIS(G)$. It holds that $\alpha(G)=\alpha(G')+w(u)$.
\end{proof}

\begin{proof}[Proof of R\ref{lem: reduction 4}]
We prove that $u$ is in at least one set in $MGIS(G)$. 
Assume, for the sake of contradiction, that $u$ is not included in any set in $MGIS(G)$.
Let $I$ be an arbitrary set in $MGIS(G)$.
Let $S:=N_p(u)\cap I$ and $S':=N_r(u)\cap I$.
Now, we remove $S$ from $I$ and add $u$ to $I$ to obtain $I'$. Consider the difference between two net benefits $nb(I')- nb(I)\geq w(u)-\sum_{v\in S}\Tilde{w}(v)-\sum_{v\in S'}p(\{u,v\})\geq 0$. Moreover, $G[I']$ does not contain permanent edges.
Therefore, $I'\in MGIS(G)$,
which contradicts the assumption. So $u$ is in at least one set in $MGIS(G)$. It holds that $\alpha(G)=\alpha(G')+w(u)$.
\end{proof}

\begin{proof}[Proof of R\ref{lem: reduction 5}]
We prove that $u$ is not in any set in $MGIS(G)$. Given an arbitrary $I\in MGIS(G)$, there are two cases. When none of the neighbors of $u$ is included in $I$, $I$ does not include $u$ because $w(u)<0$.
When any subset of neighbors of $u$ is included in $I$, $I$ does not include $u$ either because $w(u)-\sum_{v\in S}p(\{u,v\})<0$ for any subset $S\subseteq N(u)$.
Since $I$ is an arbitrary set in $MGIS(G)$, $u$ will not appear in any $I\in MGIS(G)$ and we can directly remove $u$ from $G$.
\end{proof}

\begin{proof}[Proof of R\ref{lem: reduction 6}]
Observe that there exists a vertex set $I\in MGIS(G)$ such that $I\cap N_p[u]\neq \varnothing$, because if $I\cap N_p[u]=\varnothing$ then $nb(I\cup \{u\})-nb(I)\geq 0$ implies $I\cup \{u\}\in MGIS(G)$.

Whenever a vertex $x$ in $N_p(u)$ is included in $I$, other vertices in the clique $N_p[u]\setminus \{x\}$ must not be included in $I$. Then $I\cup \{u\}\setminus \{x\}\in MGIS(G)$ because $G[I\cup \{u\}\setminus \{x\}]$ does not contain permanent edges, and $nb(I\cup \{u\}\setminus \{x\})\geq nb(I)+ w(u)-w(x)+\sum_{v\in N_r(x)}p(\{v,x\})-\sum_{v\in N_r(u)} max(0,p(\{u,v\})) \geq nb(I)+ w(u)-\sum_{v\in N_r(u)} max(0,p(\{u,v\})) -\max(0,\Tilde{w}(x))\geq nb(I)$. In conclusion, $u$ is in at least one set in $MGIS(G)$. 
\end{proof}

\begin{proof}[Proof of R\ref{lem: reduction 7}]
Given an arbitrary $I\in MGIS(G)$.

When $w(u)\geq p(\{u,v\})$ and $w(u)\geq 0$. 
If $u\notin I$ and $v\in I$, then $I\cup \{u\}\in MGIS(G)$ because $nb(I\cup \{u\})=nb(I)+w(u)-p(\{u,v\})\geq nb(I)$. 
If $u\notin I$ and $v\notin I$, then $I\cup \{u\}\in MGIS(G)$ because $G[I\cup \{u\}]$ does not contain permanent edges, and $nb(I\cup \{u\})=nb(I)+w(u)\geq nb(I)$. 
In summary, $u$ is included in at least one set in $MGIS(G)$.

When $w(u)< 0$ and $w(u)\geq p(\{u,v\})$. 
If $v\in I$, then $I\cup \{u\}\in MGIS(G)$ because $nb(I\cup \{u\})=nb(I)+w(u)-p(\{u,v\})\geq nb(I)$ for $u\notin I$, and $G[I\cup \{u\}]$ does not contain permanent edges. If $v\notin I$, then $nb(I\cup \{u\}) < nb(I\setminus \{u\})$. 
In other words, there exists a set $S\in MGIS(G)$ such that $u\in S$ if and only if there exists a set $S'\in MGIS(G')$ { such}  that $v\in S'$. 
It is easy to verify $\alpha(G')=\alpha(G)$ in all combinations of $u,v$ that could appear in a set in $MGIS(G)$.

When $w(u)\geq 0$, $w(u)<p(\{u,v\})$ and $w(u)\geq \Tilde{w}(v)$. If $u\notin I, v\in I$, then $nb(I\cup \{u\})\setminus \{v\})-nb(I)=w(u)-w(v)-\sum_{x\in N_r(v)\cap I}max(0,-p(\{v,x\}))=w(u)-\Tilde{w}(v)\geq 0$, implying $I\cup \{u\}\setminus\{v\}\in MGIS(G)$. Otherwise, if $u\notin I, v\notin I$, then $nb(I\cup \{u\})-nb(I)= w(u)\geq 0$, thus $I\cup \{u\}\in MGIS(G)$. In summary, $u$ is included in at least one vertex set $I'\in MGIS(G)$ and $v$ is not in $I'$. 

When $w(u)\geq 0$, $w(u)<p(\{u,v\})$ and $w(u)<\Tilde{w}(v)$. $u$ and $v$ cannot be both included in $I$. Indeed, if $u,v\in I$, then $nb(I\setminus\{u\})=nb(I)+p(\{u,v\})-w(u)>nb(I)$. Moreover, $G[I\setminus\{u\}]$ does not contain permanent edges. This contradicts the optimality of $I$. Furthermore, if $v\notin I$, then $I\cup \{u\}\in MGIS(G)$ because $nb(I\cup \{u\})-nb(I\setminus \{u\})=w(u)\geq 0$.
We conclude that there exists a vertex set in $MGIS(G)$ that includes $u$ iff there exists a vertex set in $MGIS(G')$ that does not contain $v$.
It is easy to verify $\alpha(G')=\alpha(G)$ in all combinations of $u,v$ that could appear in a set in $MGIS(G)$.


Finally, when $w(u)< 0$ and $w(u)< p(\{u,v\})$. We have $nb(I\cup\{u\})<nb(I\setminus \{u\})$,  therefore $u$ is not in any set in $MGIS(G)$ and we can directly remove $u$. 
\end{proof}

\begin{proof}[Proof of R\ref{lem: reduction 8}]
Given an arbitrary $I\in MGIS(G)$.
Because that $\{x,y\}\in E_p$, it holds that $x$ and $y$ cannot be both included in $I$.

\noindent    Case 1: $w(u)\geq \max(p(\{u,x\}),p(\{u,y\}))$. 

    When $w(u)\geq 0$. If both $x$ and $y$ are not included in $I$ and $u\notin I$, then $nb(I\cup \{u\})=nb(I)+w(u)\geq nb(I)$.  If $u\notin I$ and either $x\in I$ or $y\in I$, $nb(I\cup \{u\})\geq nb(I)+w(u)-\max(p(\{u,x\}),p(\{u,y\}))\geq nb(I)$.  In summary, there exists a vertex set in $ MGIS(G)$ that includes $u$.
    
    When $w(u)<0$. If both $x$ and $y$ are not included in $I$, then $u\notin I$ because $nb(I\cup \{u\})=nb(I\setminus \{u\})+w(u)< nb(I\setminus \{u\})$.  If either $x\in I$ or $y\in I$  and $u\notin I$,  then $I\cup \{u\}\in MGIS(G)$ because $nb(I\cup \{u\})\geq nb(I)+w(u)-\max(p(\{u,x\}),p(\{u,y\}))\geq nb(I)$.
    We conclude that there exists a vertex set in $MGIS(G)$ that includes $u$ iff there exists a vertex set in $MGIS(G')$ that includes either $x$ or $y$. In other words, we can defer the decision of $u$ after determining $x$ and $y$. 
    It is easy to verify $\alpha(G')=\alpha(G)$ in all combinations of $u,x,y$ that could appear in a set in $MGIS(G)$.

\noindent    Case 2: $p(\{u,y\})\leq w(u)<p(\{u,x\})$. In this case, if $y$ is in $I$, then $x\notin I$ and $I\cup \{u\} \in MGIS(G)$. 
    
    When $w(u)\geq 0$ and $w(u)\geq \Tilde{w}(x)$. 
    If $x$ is in $I$, then $y,u$ are not in $I$ and $I\cup \{u\}\setminus \{x\}\in MGIS(G)$ because $nb(I\cup \{u\}\setminus \{x\})\geq nb(I)+w(u)-\Tilde{w}(x)\geq nb(I)$.
    If neither $x$ nor $y$ is included in $I$, then $I\cup \{u\}\in MGIS(G)$ because $nb(I\cup \{u\})=nb(I)+w(u)\geq nb(I)$ for $u\notin I$. 
    In summary, there exists a vertex set $I'\in MGIS(G)$ that includes $u$ and not includes $x$.

    When $w(u)\geq 0$ and $w(u)<\Tilde{w}(x)$. 
    There are five cases: $x\in I$ and $u,y\notin I$, $y\in I$ and $u,x\notin I$, $u\in I$ and $x,y\notin I$, $u,y\in I$ and $x\notin I$, and $u,x\in I$ and $y\notin I$.
    But $u,x$ cannot appear in the same $I$, otherwise $nb(I\setminus \{u\})=nb(I)+p(\{u,x\})-w(u)>0$ without containing permanent edges, which contradicts the optimality of $I$.
    If $y\in I$ and $u,x\notin I$, then $I\cup \{u\} \in MGIS(G)$.
    In summary, there exists a vertex set $I'\in MGIS(G)$ such that either $x\in I'$, or both $u,y\in I'$, or $u\in I'$.
    Since $u$ only connects to $x$ and $y$ in $G$, we conclude that there exists a vertex set $S\in MGIS(G)$ that includes $u$ iff there exists a vertex set $S'\in MGIS(G')$ that does not include $x$. In other words, we can defer $u$ after the decision of $x$. 
    It is easy to verify $\alpha(G')=\alpha(G)$ in all combinations of $u,x,y$ that could appear in a set in $MGIS(G)$.

    When $w(u)<0$. If $y\notin I$, then $u\notin I$ because otherwise $nb(I\setminus \{u\})>nb(I)$. 
    Since $u$ only connects to $x$ and $y$ in $G$, we conclude that there exists a vertex set $S\in MGIS(G)$ that includes $u$ iff there exists a vertex set $S'\in MGIS(G')$ that includes $y$. In other words, we can defer $u$ after the decision of $y$. 
    It is easy to verify $\alpha(G')=\alpha(G)$ in all combinations of $u,x,y$ that could appear in a set in $MGIS(G)$.
    

\noindent   Case 3: $w(u)<p(\{u,x\})$ and $w(u)<p(\{u,y\})$. When $w(u)\geq 0$, the three vertices form a triangle where every edge has penalty $M$. The result follows from theorem 3.3 of \cite{gu2021towards}.
In particular, (1) if $w(u)\geq 0$ and $w(u)\geq \Tilde{w}(x)$, then $u$ must be in at least one set in $MGIS(G)$. (2) If $w(u)\geq 0$ and $\Tilde{w}(y)\leq w(u)<\Tilde{w}(x)$, then $y$ must not be in at least one set in $MGIS(G)$. There exists a vertex set $S\in MGIS(G)$ that includes $u$ iff there exists a vertex set $S'\in MGIS(G')$ that does not include $x$. (3) If $w(u)\geq 0$ and  $w(u)< \Tilde{w}(y)$, then 
there exists a vertex set $S\in MGIS(G)$ that includes $u$ iff there exists a vertex set $S'\in MGIS(G')$ that does not include both $x$ and $y$. (4) When $w(u)<0$, the result follows from R\ref{lem: negative profit}.
It is easy to verify $\alpha(G')=\alpha(G)$ in all conditions and all combinations of $u,x,y$ that could appear in a set in $MGIS(G)$.
\end{proof}
\begin{proof}[Proof of R\ref{lem: reduction 9}]
Given an arbitrary $I\in MGIS(G)$.
Notice that $x$ and $y$ can be both included in $I$.

\noindent
Case 1: $w(u)\geq p(\{u,x\})$. In this case, if one of $x,y$ is in $I$, say $x$, then $I\cup \{u\}\in MGIS(G)$ because 
$nb(I\cup \{u\})-nb(I)\geq w(u)-p(\{u,x\})\geq 0$ when $u\notin I$.

When $w(u)\geq 0$ and $w(u)\geq p(\{u,x\})+p(\{u,y\})$. If either $x\in I$ or $y\in I$, then $I\cup \{u\}\in MGIS(G)$ because
$nb(I\cup \{u\})-nb(I)\geq w(u)-\max(p(\{u,x\})+p(\{u,y\}),p(\{u,x\}))\geq 0$ for $u\notin I$.
If none of $x,y$ is in $I$, then $I\cup \{u\}\in MGIS(G)$ because $w(u)\geq 0$.
Therefore, $u$ must belong to at least one set in $MGIS(G)$.

When $w(u)\geq 0$ and $w(u)< p(\{u,x\})+p(\{u,y\})$. If both $x,y$ are in $I$, then $nb(I\cup \{u\})<nb(I)$. If either $x\notin I$ or $y\notin I$, then $nb(I\cup\{u\})\geq nb(I)$. 
In other words, including $u$ has negative profit only if both $x$ and $y$ are included in the same set in $MGIS(G)$.
we conclude that there exists a vertex set $S\in MGIS(G)$ that includes $u$ iff there exists a vertex set $S'\in MGIS(G')$ such that either $x\notin S'$ or $y\notin S'$.
It is easy to verify $\alpha(G')=\alpha(G)+w(u)$ in all combinations of $u,x,y$ that could appear in a set in $MGIS(G)$.

When $w(u)<0$, it holds that $w(u)\geq p(\{u,x\})+p(\{u,y\})$ because $w(u)\geq p(\{u,x\})$ and $p(\{u,y\})<0$. 
If either $x\in I$ or $y\in I$, then $nb(I\cup \{u\})\geq nb(I)$. If both $x\notin I$ and $y\notin I$, then $nb(I\cup\{u\})< nb(I)$. In other words, including $u$ has positive profit if and only if either $x$ or $y$ is included. 
we conclude that there exists a vertex set $S\in MGIS(G)$ that includes $u$ iff there exists a vertex set $S'\in MGIS(G')$ such that either $x\in S'$ or $y\in S'$.
It is easy to verify $\alpha(G')=\alpha(G)$ in all combinations of $u,x,y$ that could appear in a set in $MGIS(G)$.

\noindent
Case 2: $p(\{u,y\})\leq w(u)<p(\{u,x\})$. In this case, we consider whether the three vertices $u,x,y$ can be included in $I$. 

When $w(u)\geq 0$ and $w(u)\geq p(\{u,x\})+p(\{u,y\})$, the all three vertices may be included in $I$. 
All combinations of $x,y,u$ may appear in $I$. Except that if $x$ is included and $y$ is not included in $I$, then $u\notin I$ because $nb(I\cup \{u\})=nb(I\setminus \{u\})+w(u)-p(\{u,x\})<nb(I\setminus \{u\})$.
In other words, there exists a vertex set $S\in MGIS(G)$ that includes $u$ iff there exists a vertex set $S'\in MGIS(G')$ such that either $x\notin S'$ or $y\in S'$.
It is easy to verify $\alpha(G')+w(u)=\alpha(G)$ in all combinations of $x,y,u$ that could appear in a set in $MGIS(G)$.

When $w(u)\geq 0$ and $w(u)< p(\{u,x\})+p(\{u,y\})$, the three vertices can not be included in $I$. 
Moreover, $u$ and $x$ cannot be both included by $I$ because otherwise $nb(I\setminus u)>nb(I)$ without containing permanent edges, which contradicts the optimality of $I$.
If $x\notin I$, then $I\cup \{u\}\in MGIS(G)$ because $nb(I\cup \{u\})\geq nb(I\setminus \{u\})$.
In other words, there exists a vertex set $S\in MGIS(G)$ that includes $u$ iff there exists a vertex set $S'\in MGIS(G')$ such that $x\notin S'$.
It is easy to verify $\alpha(G')+w(u)=\alpha(G)$ in all combinations of $x,y,u$ that could appear in a set in $MGIS(G)$.

When $w(u)<0$ and $w(u)\geq p(\{u,x\})+p(\{u,y\})$, the three vertices can be included in $I$.
If $y\notin I$, then $nb(I\cup \{u\})\leq nb(I\setminus \{u\})$ regardless of whether $x\in I$.
If $y\in I$, then $nb(I\cup \{u\})\geq nb(I\setminus \{u\})$ regardless of whether $x\in I$
In other words, there exists a vertex set $S\in MGIS(G)$ that includes $u$ iff there exists a vertex set $S'\in MGIS(G')$ such that $y\in S'$.
It is easy to verify $\alpha(G')=\alpha(G)$ in all combinations of $x,y,u$ that could appear in a set in $MGIS(G)$.

When $w(u)<0$ and $w(u)< p(\{u,x\})+p(\{u,y\})$, the three vertices can not be included in $I$. 
$u$ is in $I$ iff $y$ is included and $x$ is not included in $I$.
In other words, there exists a vertex set $S\in MGIS(G)$ that includes $u$ iff there exists a vertex set $S'\in MGIS(G')$ such that $x\notin S'$ and $y\in S'$.
It is easy to verify $\alpha(G')=\alpha(G)$ in all combinations of $x,y,u$ that could appear in a set in $MGIS(G)$.

\noindent
Case 3: $w(u)<p(\{u,y\})$.
When $w(u)\geq 0$.
$u$ is in at least one set $I\in MGIS(G)$ iff neither $x$ nor $y$ is included in $I$. 
In other words, there exists a vertex set $S\in MGIS(G)$ that includes $u$ iff there exists a vertex set $S'\in MGIS(G')$ such that $x\notin S'$ and $y\notin S'$.
It is easy to verify $\alpha(G')+w(u)=\alpha(G)$ in all combinations of $x,y,u$ that could appear in a set in $MGIS(G)$.

When $w(u)<0$ and $w(u)<p(\{u,x\})+p(\{u,y\})$.
$u$ is not in any $I\in MGIS(G)$ because otherwise $nb(I\setminus \{u\})>nb(I)$ without containing permanent edges, a contradiction.
Therefore, we can remove $u$ directly.

When $w(u)<0$ and $w(u)\geq p(\{u,x\})+p(\{u,y\})$.
$u$ is in at least one set $I\in MGIS(G)$ iff both $x$ and $y$ are included in $I$.
In other words, there exists a vertex set $S\in MGIS(G)$ that includes $u$ iff there exists a vertex set $S'\in MGIS(G')$ such that $x\in S'$ and $y\in S'$.
It is easy to verify $\alpha(G')=\alpha(G)$ in all combinations of $x,y,u$ that could appear in a set in $MGIS(G)$.
\end{proof}

\begin{proof}[Proof of R\ref{lem: reduction 10}]
    Observe that there exists a vertex set $I\in MGIS(G)$ that includes at least one of $x,u$ must be included in $I$. Because if none of them is included in at least one set $I\in MGIS(G)$, then $nb(I\cup \{u\})-nb(I)\geq w(u)- \sum_{v\in N_r(u)} \max(0,p(\{u,v\}))\geq 0$ and therefore $I\cup \{u\}\in MGIS(G)$.

    Given an arbitrary $I\in MGIS(G)$. 
    Furthermore, $x$ and $u$ cannot be both included in $I$ since $\{u,x\}\in E_p$. In summary, there are two possible states in $I$: either $x\in I,u\notin I$ or $u\in I,x\notin I$. 
    Therefore,  there exists a vertex set $S\in MGIS(G)$ that includes $u$ iff there exists a vertex set $S'\in MGIS(G')$ such that $x\notin S'$.
    It is easy to verify $\alpha(G')+w(u)=\alpha(G)$ in all combinations of $x,u$ that could appear in a set in $MGIS(G)$.
\end{proof}

\begin{proof}[Proof of R\ref{lem: reduction 11}]
Observe that there exists a vertex set $I\in MGIS(G)$ that includes at least one of $u,x,y$. If none of them is included in $I$, then $nb(I\cup \{u\})-nb(I)\geq w(u)- \sum_{v\in N_r(u)} \max(0,p(\{u,v\}))\geq 0$
In addition, at most one of $u,x,y$ can be included in $I$ because $\{u,x\}\in E_p, \{u,y\}\in E_p, \{x,y\}\in E_p$.

Given an arbitrary $I\in MGIS(G)$. 
When $L(u)\geq \Tilde{w}(x)$. If one of $x,y$ is included in $I$,say $x$, then $nb(I\cup\{u\}\setminus \{x\})\geq nb(I)$.
In summary, $u$ must be included in at least one set in $MGIS(G)$.

When $\Tilde{w}(y)\leq L(u)<\Tilde{w}(x)$. If $y$ is included in $I$, then $u\notin I$ and
$nb(I\cup\{u\}\setminus \{y\})\geq nb(I)$.
Therefore, there is a vertex set $I'\in MGIS(G)$ such that $y\notin I'$ and exactly one of $u,x$ is included in $I'$. 
We conclude that there exists a vertex set $S\in MGIS(G)$ such that  $u\in S$ iff there exists a vertex set $S'\in MGIS(G')$ such that $x\notin S'$.
It is easy to verify $\alpha(G')+w(u)=\alpha(G)$ in all combinations of $x,y,u$ that could appear in a set in $MGIS(G)$.

When $L(u)< \Tilde{w}(y)$.  
$u$ is in at least one set $I\in MGIS(G)$ iff both $x,y$ are not included in $I$. 
In other words, there exists a vertex set $S\in MGIS(G)$ that includes $u$ iff there exists a vertex set $S'\in MGIS(G')$ such that $x\notin S'$ and $y\notin S'$.
It is easy to verify $\alpha(G')+w(u)=\alpha(G)$ in all combinations of $x,y,u$ that could appear in a set in $MGIS(G)$.
\end{proof}

\begin{proof}[Proof of R\ref{lem: reduction 12}]
Given an arbitrary $I\in MGIS(G)$. 

When $w(u)\geq \Tilde{w}(v)+w^+(N(u)\setminus N_p[v])$. If $v\in I$, then $u\notin I$. 
We can replace $N(u)\cap I$ with $u$ to get a larger or equal net benefit because $nb(I\cup \{u\}\setminus N(u))-nb(I)\geq  w(u)-( \Tilde{w}(v)+w^+(N(u)\setminus N_p[v]))\geq w(u)- (\Tilde{w}(v)+w^+(N(u)\cap I))\geq 0$. Therefore, $I\cup \{u\}\setminus N(u) \in MGIS(G)$. Since $v\in N_p(u)$, there must exist at least one set in $MGIS(G)$ that does not include $v$.

When $w(u)\geq \Tilde{w}(v)+w^+(N_p(u)\setminus N_p[v])+\sum_{x\in N_r(u)\setminus N_p(v)} \max(0,p(\{u,x\}))$.
If $v\in I$, then $u\notin I$. 
We can replace $N_p(u)\cap I$ with $u$ to get a larger net benefit because $nb(I\cup \{u\}\setminus N_p(u))-nb(I)\geq w(u)- (\Tilde{w}(v)+w^+(N_p(u)\setminus N_p(v))+\sum_{x\in N_r(u)\setminus N_p(v)} \max(0,p(\{u,x\})))\geq w(u)- (\Tilde{w}(v)+w^+(N(u)\cap I)+\sum_{x\in N_r(u)\cap I} \max(0,p(\{u,x\})))\geq 0$.
Since $v\in N_p(u)$, there must exist a set in $MGIS(G)$ that does not include $v$.
\end{proof}

\begin{proof}[Proof of R\ref{lem: reduction 13}]
Assume, for the sake of contradiction, that neither $u$ nor $v$ is included in any set in $MGIS(G)$.
If for any $I\in MGIS(G)$, $N(v)\cap I= \varnothing$, then $I\cup \{v\}\in MGIS(G)$ because $nb(I\cup \{v\})\geq nb(I)$ without containing permanent edges, a contradiction.
Note that $u\notin N(v)\cap I$.
If there exists a set $I\in MGIS(G)$ such that $N(v)\cap I\neq \varnothing$, then $w(v)\geq w^+(N(v))-max(0,w(u))\geq w^+(S)$. 
$nb(I\cup \{v\}\setminus N(v))\geq nb(I)$ without containing permanent edges, again a contradiction.
Therefore, there exists a set in $MGIS(G)$ that includes at least one of $u,v$. We conclude that $N_p(u)\cap N_p(v)$ must not be included in at least one set in $MGIS(G)$.
\end{proof}

\begin{proof}[Proof of R\ref{lem: reduction 14}]
Given an arbitrary $I\in MGIS(G)$. 
If $u$ is included in $I$, then $N_p(v)\cap I=\varnothing$. 
$I\cup \{v\}\in MGIS(G)$ because $nb(I\cup \{v\})-nb(I\setminus \{v\})= w(v)- \sum_{x\in N_r(v)}\max(0,p(\{v,x\}))\geq 0$. The same argument holds when $v$ is included in $I$. 
Therefore, we can fold $u$ and $v$ into a vertex $v'$.
There exists $S\in MGIS(G)$ such that $u\in S$ and $v\in S$ iff there exists $S'\in MGIS(G')$ such that $v'\in S'$.
\end{proof}

\subsection{Comparative Details on 216 Instances of SET-1} \label{app: set1}
\setlength{\tabcolsep}{4pt}  
{\tiny  
\begin{longtable}{l|ccc|cc|cc|cc|cc|cc}
\caption{Comparative results on SET-1 instances}
\label{tab: Set1 detail} \\
\hline
\multirow{2}{*}{Instance} & \multirow{2}{*}{$|V|$} & \multirow{2}{*}{$|E_p|$} & \multirow{2}{*}{$|E_r|$} & \multicolumn{2}{c|}{CB\&B} & \multicolumn{2}{c|}{LA-B\&B} &\multicolumn{2}{c|}{ILS-VND}  & \multicolumn{2}{c}{ALS} & \multicolumn{2}{c}{RLS} \\
\cline{5-14}
& & && BKV & time(s) & BKV & time(s)  & \textbf{BKV} &htime(s) & \textbf{BKV} &htime(s) & \textbf{BKV} &htime(s) \\
\hline
\endfirsthead

\multicolumn{14}{c}{Table \thetable{} -- Continued from previous page} \\
\hline
\multirow{2}{*}{Instance} & \multirow{2}{*}{$|V|$} & \multirow{2}{*}{$|E_p|$} & \multirow{2}{*}{$|E_r|$} & \multicolumn{2}{c|}{CB\&B} & \multicolumn{2}{c|}{LA-B\&B} &\multicolumn{2}{c|}{ILS-VND}  & \multicolumn{2}{c}{ALS} & \multicolumn{2}{c}{RLS} \\
\cline{5-14}
& & && BKV & time(s) & BKV & time(s)  & \textbf{BKV} &htime(s) & \textbf{BKV} &htime(s) & \textbf{BKV} &htime(s) \\
\hline
\endhead

\hline \multicolumn{14}{r}{Continued on next page} \\
\endfoot

\hline
\endlastfoot
\multicolumn{14}{l}{$p=0.25$, Setting 1} \\
\hline
brock200\_2\_A & 200 & 7,416 & 2,460 &\textbf{986}& 280 & \textbf{986} & 2&\textbf{986} & $<$1 & \textbf{986} & $<$1 & \textbf{986} & $<$1 \\
brock400\_2\_A& 400 & 44,847 & 14,939 &\textbf{765} & 441 & \textbf{765} & 1 &\textbf{765} & $<$1  & \textbf{765} & $<$1 & \textbf{765} & $<$1 \\
C125.9\_A & 125 & 5,250 & 1,713 &\textbf{454} & $<$1 & \textbf{454} & $<$1 & \textbf{454} & $<$1& \textbf{454} & $<$1 & \textbf{454} & $<$1\\
C250.9\_A & 250 & 20,940 & 7,044&\textbf{581} & 10 & \textbf{581} & $<$1& \textbf{581} & $<$1 & \textbf{581} & $<$1 & \textbf{581} & $<$1 \\
gen200\_p0.9\_55\_A & 200 & 13,434 & 4,176 &\textbf{535} & 4  & \textbf{535} & $<$1 & \textbf{535} & $<$1 & \textbf{535} & $<$1 & \textbf{535} & $<$1\\
gen400\_p0.9\_75\_A & 400 & 53,671 & 18,149 &\textbf{628} & 82& \textbf{628} & $<$1 & \textbf{628} & 3 & \textbf{628} & $<$1& \textbf{628} & $<$1 \\
hamming8-4\_A & 256 & 15,580 & 5,284 &\textbf{1,094} & 240 & \textbf{1,094} & $<$1 & \textbf{1,094} &  $<$1 & \textbf{1,094} & $<$1 & \textbf{1,094} &  $<$1\\
keller4\_A & 171 & 7,057 & 2,378 &\textbf{941} & 15 & \textbf{941} & $<$1 & \textbf{941} &  $<$1 & \textbf{941} & $<$1& \textbf{941} &  $<$1 \\
MANN\_a27\_A & 378 & 53,050 & 17,501 &\textbf{533} & 28 & \textbf{533} & $<$1& \textbf{533} &$<$1 & \textbf{533} &$<$1& \textbf{533} &$<$1  \\
p\_hat300-1\_A & 300 & 8,239 & 2,694 & 2,308 & 10,800 & 2,680 & 10,800 & \textbf{2,744}  & $<$1 & \textbf{2,744} & $<$1& \textbf{2,744}  & $<$1 \\
p\_hat300-2\_A & 300 & 16,403 & 5,525 &\textbf{2,076} & 5,182 & \textbf{2,076} & 608 & \textbf{2,076} & $<$1 & \textbf{2,076} & $<$1& \textbf{2,076} & $<$1 \\
p\_hat300-3\_A & 300 & 25,132 & 8,258 &\textbf{739} & 100 & \textbf{739} & $<$1 &  \textbf{739} &$<$1 & \textbf{739} &$<$1& \textbf{739} &$<$1 \\
brock200\_2\_B & 200 & 7,416 & 2,460 &\textbf{962} & 276  & \textbf{962} & 2 & \textbf{962} & $<$1 & \textbf{962} & $<$1& \textbf{962} & $<$1 \\
brock400\_2\_B & 400 & 44,847 & 14,939 &\textbf{741} & 445 & \textbf{741} & 1 & \textbf{741}  & $<$1 & \textbf{741} & $<$1& \textbf{741}  & $<$1 \\
C125.9\_B & 125 & 5,250 & 1,713 & \textbf{437}& 1 & \textbf{437} & $<$1& \textbf{437}  & $<$1 & \textbf{437} & $<$1& \textbf{437}  & $<$1 \\
C250.9\_B & 250 & 20,940 & 7,044 & \textbf{549}& 11 & \textbf{549} & $<$1 & \textbf{549} & $<$1 &  \textbf{549} & $<$1& \textbf{549} & $<$1 \\
gen200\_p0.9\_55\_B & 200 & 13,434 & 4,176 &\textbf{510} & 5 & \textbf{510} & $<$1 & \textbf{510}  & $<$1 & \textbf{510}  & $<$1& \textbf{510}  & $<$1 \\
gen400\_p0.9\_75\_B & 400 & 53,671 & 18,149 &\textbf{595} & 80  & \textbf{595} & $<$1 & \textbf{595}  & $<$1 &  \textbf{595} & $<$1& \textbf{595}  & $<$1 \\
hamming8-4\_B & 256 & 15,580 & 5,284&\textbf{1,094} & 234 & \textbf{1,094} & $<$1 & \textbf{1,094} & $<$1 &\textbf{1,094} & $<$1& \textbf{1,094} & $<$1 \\
keller4\_B & 171 & 7,057 & 2,378 &\textbf{941} &15  & \textbf{941} & $<$1 &\textbf{941}  & $<$1 & \textbf{941} & $<$1& \textbf{941}  & $<$1 \\
MANN\_a27\_B & 378 & 53,050 & 17,501 &\textbf{503} & 31 & \textbf{503} & $<$1 & \textbf{503} & $<$1 &\textbf{503}& $<$1 & \textbf{503} & $<$1\\
p\_hat300-1\_B & 300 & 8,239 & 2,694 &2,259 & 10,800 & 2,643 & 10,800 & \textbf{2,712}  & $<$1 & \textbf{2,712} & $<$1& \textbf{2,712}  & $<$1 \\
p\_hat300-2\_B & 300 & 16,403 & 5,525 &\textbf{2,062} & 5,194 & \textbf{2,062} & 680 & \textbf{2,062} & $<$1 & \textbf{2,062} & $<$1& \textbf{2,062} & $<$1 \\
p\_hat300-3\_B & 300 & 25,132 & 8,258 &\textbf{713} & 102 & \textbf{713} &  $<$1 & \textbf{713}  & $<$1 &\textbf{713} & $<$1& \textbf{713}  & $<$1\\

brock200\_2\_C & 200 & 7,416 & 2,460 & \textbf{932}& 280 & \textbf{932} & 2 & \textbf{932}  & $<$1 & \textbf{932} & $<$1& \textbf{932}  & $<$1 \\
brock400\_2\_C & 400 & 44,847 & 14,939 &\textbf{698} & 455 & \textbf{698} & 2 & \textbf{698}  & $<$1 & \textbf{698} & $<$1& \textbf{698}  & $<$1 \\
C125.9\_C & 125 & 5,250 & 1,713 &\textbf{403} & 1 & \textbf{403} & $<$1& \textbf{403} & $<$1 & \textbf{403} & $<$1& \textbf{403} & $<$1 \\
C250.9\_C & 250 & 20,940 & 7,044 &\textbf{502} & 12 & \textbf{502} & $<$1 & \textbf{502}  & $<$1 &  \textbf{502} & $<$1& \textbf{502}  & $<$1 \\
gen200\_p0.9\_55\_C & 200 & 13,434 & 4,176 &\textbf{467} & 4 & \textbf{467} & $<$1 & \textbf{467}   & $<$1 & \textbf{467}  & $<$1& \textbf{467}   & $<$1 \\
gen400\_p0.9\_75\_C & 400 & 53,671 & 18,149 &\textbf{533} & 81  & \textbf{533} & $<$1 & \textbf{533}  & $<$1 &  \textbf{533} & $<$1& \textbf{533}  & $<$1 \\
hamming8-4\_C & 256 & 15,580 & 5,284 &\textbf{1,094} & 229 & \textbf{1,094} & 1 & \textbf{1,094}  & $<$1&\textbf{1,094} & $<$1& \textbf{1,094}  & $<$1 \\
keller4\_C & 171 & 7,057 & 2,378 &\textbf{941} & 14 & \textbf{941} & $<$1 & \textbf{941} & $<$1 & \textbf{941} & $<$1& \textbf{941} & $<$1 \\
MANN\_a27\_C & 378 & 53,050 & 17,501 &\textbf{443} & 30 & \textbf{443} & $<$1 & \textbf{443}& $<$1 &\textbf{443}& $<$1& \textbf{443}& $<$1 \\
p\_hat300-1\_C & 300 & 8,239 & 2,694 &2,205 & 10,800 & 2,566 & 10,800 & \textbf{2,649}  & $<$1 & \textbf{2,649} & $<$1& \textbf{2,649}  & $<$1 \\
p\_hat300-2\_C & 300 & 16,403 & 5,525 & \textbf{2,033}& 5,360 & \textbf{2,033} & 778 & \textbf{2,033}  & $<$1 & \textbf{2,033} & $<$1& \textbf{2,033}  & $<$1 \\
p\_hat300-3\_C & 300 & 25,132 & 8,258 &\textbf{688} &111 & \textbf{688} & $<$1 & \textbf{688}  & $<$1 &\textbf{688} & $<$1& \textbf{688}  & $<$1 \\
\hline
\multicolumn{14}{l}{$p=0.25$, Setting 2} \\
\hline
brock200\_2\_A & 200 & 7,416 & 2,460 &\textbf{1,489} & 452 & \textbf{1,489} & 5 & \textbf{1,489} & $<$1 & \textbf{1,489} & $<$1 & \textbf{1,489} & $<$1 \\
brock400\_2\_A & 400 & 44,847 & 14,939 &\textbf{1,084} &593 & \textbf{1,084} & 2 & \textbf{1,084} & 2 & \textbf{1,084} & $<$1 & \textbf{1,084} & $<$1 \\
C125.9\_A & 125 & 5,250 & 1,713 & \textbf{685}& 1 & \textbf{685} & $<$1 & \textbf{685} & $<$1 & \textbf{685} & $<$1 & \textbf{685} & $<$1 \\
C250.9\_A & 250 & 20,940 & 7,044 &\textbf{785} & 12 & \textbf{785} & $<$1& \textbf{785} & $<$1 & \textbf{785} & $<$1 & \textbf{785} & $<$1 \\
gen200\_p0.9\_55\_A & 200 & 13,434 & 4,176 &\textbf{778} & 5 & \textbf{778} & $<$1 & \textbf{778} & $<$1 & \textbf{778} & $<$1 & \textbf{778} & $<$1 \\
gen400\_p0.9\_75\_A & 400 & 53,671 & 18,149 &\textbf{882} & 94 & \textbf{882} & $<$1 & \textbf{882} & $<$1 & \textbf{882} & $<$1& \textbf{882} & $<$1 \\
hamming8-4\_A & 256 & 15,580 & 5,284 & \textbf{1,790} & 331 & \textbf{1,790} & 2 & \textbf{1,790} & $<$1 & \textbf{1,790} & $<$1 & \textbf{1,790} & $<$1 \\
keller4\_A & 171 & 7,057 & 2,378 & \textbf{1,500}& 23 & \textbf{1,500} & $<$1& \textbf{1,500} & $<$1 & \textbf{1,500} & $<$1 & \textbf{1,500} & $<$1 \\
MANN\_a27\_A & 378 & 53,050 & 17,501 &\textbf{683}  & 32  & \textbf{683} & $<$1 &\textbf{683} &$<$1 & \textbf{683} &$<$1 & \textbf{683} &$<$1  \\
p\_hat300-1\_A & 300 & 8,239 & 2,694 & 3,981& 10,800 & 4,182 & 10,800 & \textbf{4,674} & 1 & \textbf{4,674} & $<$1 & \textbf{4,674} & $<$1 \\
p\_hat300-2\_A & 300 & 16,403 & 5,525 &\textbf{2,994} & 8,743 & \textbf{2,994} & 9,402 & \textbf{2,994} & $<$1 & \textbf{2,994} & $<$1 & \textbf{2,994} & $<$1 \\
p\_hat300-3\_A & 300 & 25,132 & 8,258 &\textbf{1,180} & 137 & \textbf{1,180} & 1 & \textbf{1,180} &$<$1 & \textbf{1,180} &$<$1 & \textbf{1,180} &$<$1 \\

brock200\_2\_B & 200 & 7,416 & 2,460 &\textbf{739} & 456 & \textbf{739} & 5 & \textbf{739} & $<$1 & \textbf{739} & $<$1 & \textbf{739} & $<$1 \\
brock400\_2\_B & 400 & 44,847 & 14,939 &\textbf{534} &611  & \textbf{534} & 2 & \textbf{534} & 2 & \textbf{534} & $<$1 & \textbf{534} & $<$1 \\
C125.9\_B & 125 & 5,250 & 1,713 &\textbf{335} &1  & \textbf{335} & $<$1& \textbf{335} & $<$1 & \textbf{335} & $<$1 & \textbf{335} & $<$1 \\
C250.9\_B & 250 & 20,940 & 7,044 & \textbf{385}& 13 & \textbf{385} & $<$1 & \textbf{385} & $<$1 &  \textbf{385} & $<$1  &  \textbf{385} & $<$1 \\
gen200\_p0.9\_55\_B & 200 & 13,434 & 4,176 &\textbf{378} & 5 & \textbf{378} & $<$1 & \textbf{378}  & $<$1 & \textbf{378}  & $<$1 & \textbf{378}  & $<$1 \\
gen400\_p0.9\_75\_B & 400 & 53,671 & 18,149 & \textbf{432}& 96 & \textbf{432} & $<$1 &   \textbf{432} & $<$1 &  \textbf{432} & $<$1 &  \textbf{432} & $<$1 \\
hamming8-4\_B & 256 & 15,580 & 5,284 &\textbf{890} &333  & \textbf{890} & 2 & \textbf{890} & $<$1 &\textbf{890} & $<$1 &\textbf{890} & $<$1 \\
keller4\_B & 171 & 7,057 & 2,378 & \textbf{750} & 22 & \textbf{750} & $<$1 & \textbf{750} & $<$1 & \textbf{750} & $<$1 & \textbf{750} & $<$1 \\
MANN\_a27\_B & 378 & 53,050 & 17,501 & \textbf{333} &34  & \textbf{333} & $<$1 & \textbf{333}& $<$1 &\textbf{333}& $<$1 &\textbf{333}& $<$1 \\
p\_hat300-1\_B & 300 & 8,239 & 2,694 &1,981 & 10,800 & 2,082 & 10,800 & \textbf{2,324} & 2 & \textbf{2,324} & $<$1 & \textbf{2,324} & $<$1 \\
p\_hat300-2\_B & 300 & 16,403 & 5,525 &\textbf{1,494} & 8,947 & \textbf{1,494} & 10,800 & \textbf{1,494} & $<$1 & \textbf{1,494} & $<$1 & \textbf{1,494} & $<$1 \\
p\_hat300-3\_B & 300 & 25,132 & 8,258 &\textbf{580} & 134 & \textbf{580} & 1 & \textbf{580} & 1&\textbf{580} & $<$1 &\textbf{580} & $<$1 \\
brock200\_2\_C & 200 & 7,416 & 2,460 &\textbf{364} & 460 & \textbf{364} & 5 & \textbf{364} & $<$1 & \textbf{364} & $<$1 & \textbf{364} & $<$1 \\
brock400\_2\_C & 400 & 44,847 & 14,939 &\textbf{259}& 646 & \textbf{259} & 2 & \textbf{259} & 1 & \textbf{259} & $<$1 & \textbf{259} & $<$1 \\
C125.9\_C & 125 & 5,250 & 1,713 & \textbf{160} & 1 & \textbf{160} & $<$1& \textbf{160} & $<$1 & \textbf{160} & $<$1 & \textbf{160} & $<$1 \\
C250.9\_C & 250 & 20,940 & 7,044 &\textbf{185} & 13 & \textbf{185} & $<$1 & \textbf{185} & $<$1 &  \textbf{185} & $<$1 &  \textbf{185} & $<$1 \\
gen200\_p0.9\_55\_C & 200 & 13,434 & 4,176 &\textbf{178} & 5 & \textbf{178} & $<$1 & \textbf{178}  & $<$1 & \textbf{178}  & $<$1 & \textbf{178}  & $<$1 \\
gen400\_p0.9\_75\_C & 400 & 53,671 & 18,149 &\textbf{207} & 98 & \textbf{207} & $<$1 & \textbf{207} & $<$1 &  \textbf{207} & $<$1  &  \textbf{207} & $<$1 \\
hamming8-4\_C & 256 & 15,580 & 5,284 &\textbf{440} & 355 & \textbf{440} & 2 & \textbf{440} & $<$1 &\textbf{440} & $<$1 &\textbf{440} & $<$1 \\
keller4\_C & 171 & 7,057 & 2,378 & \textbf{375} & 23 & \textbf{375} & $<$1 & \textbf{375} & $<$1 & \textbf{375} & $<$1 & \textbf{375} & $<$1 \\
MANN\_a27\_C & 378 & 53,050 & 17,501 & \textbf{158} & 33 & \textbf{158} & $<$1 & \textbf{158}& $<$1 &\textbf{158}& $<$1 &\textbf{158}& $<$1 \\
p\_hat300-1\_C & 300 & 8,239 & 2,694 &981 & 10,800 & 1,032 & 10,800 & \textbf{1,149} & 1 & \textbf{1,149} & $<$1 & \textbf{1,149} & $<$1 \\
p\_hat300-2\_C & 300 & 16,403 & 5,525 & \textbf{744} & 10,800 & \textbf{744} & 10,658 & \textbf{744} & $<$1 & \textbf{744} & $<$1 & \textbf{744} & $<$1 \\
p\_hat300-3\_C & 300 & 25,132 & 8,258 &\textbf{280} & 142 & \textbf{280} & $<$1 & \textbf{280} & $<$1 &\textbf{280} & $<$1 &\textbf{280} & $<$1 \\




\hline 
\multicolumn{14}{l}{$p=0.50$, Setting 1} \\
\hline
brock200\_2\_A & 200 & 4,916 & 4,960 &\textbf{1,298} & 8,390 & \textbf{1,298} & 93 & \textbf{1,298} & $<$1 & \textbf{1,298} & $<$1   & \textbf{1,298} & $<$1\\
brock400\_2\_A & 400 & 29,711 & 30,075 & 1,103& 10,800& \textbf{1,123} & 355 & \textbf{1,123} & $<$1 & \textbf{1,123} & $<$1 & \textbf{1,123} & $<$1 \\
C125.9\_A & 125 & 3,500 & 3,463 &\textbf{627} &5 & \textbf{627} & $<$1 & \textbf{627} & $<$1 & \textbf{627} & $<$1 & \textbf{627} & $<$1 \\
C250.9\_A & 250 & 14,017 & 13,967 &\textbf{817}  & 266 & \textbf{817} & 1& \textbf{817} & $<$1 & \textbf{817} & $<$1 & \textbf{817} & $<$1 \\
gen200\_p0.9\_55\_A & 200 & 8,908 & 9,002 & \textbf{785} &68 & \textbf{785} & $<$1 & \textbf{785} & $<$1 & \textbf{785} & $<$1 & \textbf{785} & $<$1 \\
gen400\_p0.9\_75\_A & 400 & 35,823 & 35,997 & \textbf{895} & 5,699 & \textbf{895} & 29 & \textbf{895} & $<$1 & \textbf{895} & $<$1 & \textbf{895} & $<$1 \\
hamming8-4\_A & 256 & 10,329 & 10,535 & \textbf{1,301}  & 7,608& \textbf{1,301} & 49 & \textbf{1,301} & $<$1& \textbf{1,301} & $<$1 & \textbf{1,301} & $<$1 \\
keller4\_A & 171 & 4,738 & 4,697 & \textbf{1,118} & 153  & \textbf{1,118} & 2 & \textbf{1,118} & $<$1 & \textbf{1,118} & $<$1 & \textbf{1,118} & $<$1 \\
MANN\_a27\_A & 378 & 35,345 & 35,206 & \textbf{812}  & 1,217 & \textbf{812} & 6 & \textbf{812} &$<$1 & \textbf{812} &$<$1 & \textbf{812} &$<$1  \\
p\_hat300-1\_A & 300 & 5,505 & 5,428 & 2,568  & 10,800 & 3,026 & 10,800 & \textbf{3,129} & $<$1 & \textbf{3,129} & $<$1 & \textbf{3,129} & $<$1 \\
p\_hat300-2\_A & 300 & 11,051 & 10,877 &  2,063 & 10,800 & \textbf{2,477} & 10,800 &\textbf{2,477} & $<$1 & \textbf{2,477} & $<$1 & \textbf{2,477} & $<$1 \\
p\_hat300-3\_A & 300 & 16,820 & 16,570 & \textbf{1,029} & 3,748 & \textbf{1,029} & 38& \textbf{1,029} &$<$1 & \textbf{1,029} &$<$1 & \textbf{1,029} &$<$1 \\
brock200\_2\_B & 200 & 4,916 & 4,960 & \textbf{1,224} & 9,589  & \textbf{1,224} & 125 & \textbf{1,224} & $<$1 & \textbf{1,224} & $<$1 & \textbf{1,224} & $<$1 \\
brock400\_2\_B & 400 & 29,711 & 30,075 & 1,010 & 10,800 & \textbf{1,035} & 445 & \textbf{1,035} & $<$1 & \textbf{1,035} & $<$1 & \textbf{1,035} & $<$1 \\
C125.9\_B & 125 & 3,500 & 3,463 & \textbf{582} & 5 & \textbf{582} & $<$1& \textbf{582} & $<$1 & \textbf{582} & $<$1  & \textbf{582} & $<$1\\
C250.9\_B & 250 & 14,017 & 13,967 & \textbf{744} & 292  & \textbf{744} & 2 & \textbf{744} & $<$1 &  \textbf{744} & $<$1 &  \textbf{744} & $<$1 \\
gen200\_p0.9\_55\_B & 200 & 8,908 & 9,002 &  \textbf{716} & 72 & \textbf{716} & $<$1 & \textbf{716}  & $<$1 & \textbf{716}  & $<$1 & \textbf{716}  & $<$1 \\
gen400\_p0.9\_75\_B & 400 & 35,823 & 35,997 &  \textbf{805} & 6,602 & \textbf{805} & 34 & \textbf{805} & $<$1 &  \textbf{805} & $<$1 &  \textbf{805} & $<$1 \\
hamming8-4\_B & 256 & 10,329 & 10,535 & \textbf{1,255} & 7,780  & \textbf{1,255} & 62 & \textbf{1,255} & $<$1 &\textbf{1,255} & $<$1 &\textbf{1,255} & $<$1 \\
keller4\_B & 171 & 4,738 & 4,697 & \textbf{1,094} & 156  & \textbf{1,094} & 2 & \textbf{1,094} & $<$1 & \textbf{1,094} & $<$1 & \textbf{1,094} & $<$1 \\
MANN\_a27\_B & 378 & 35,345 & 35,206 &  \textbf{707} & 1,417 & \textbf{707} & 7 & \textbf{707}& 1 &\textbf{707}& $<$1 &\textbf{707}& $<$1 \\
p\_hat300-1\_B & 300 & 5,505 & 5,428 & 2,492 & 10,800 & 2,950 & 10,800 & \textbf{3,023} & $<$1 & \textbf{3,023} & $<$1 & \textbf{3,023} & $<$1 \\
p\_hat300-2\_B & 300 & 11,051 & 10,877 & 1,979 & 10,800 & \textbf{2,405} & 10,800 & \textbf{2,405} & $<$1 & \textbf{2,405} & $<$1 & \textbf{2,405} & $<$1 \\
p\_hat300-3\_B & 300 & 16,820 & 16,570 & \textbf{967} &4,067 & \textbf{967} & 46 & \textbf{967} & $<$1 &\textbf{967} & $<$1 &\textbf{967} & $<$1 \\

brock200\_2\_C & 200 & 4,916 & 4,960 & \textbf{1,101} & 10,800 & \textbf{1,101} & 204 & \textbf{1,101} & $<$1 & \textbf{1,101} & $<$1 & \textbf{1,101} & $<$1 \\
brock400\_2\_C & 400 & 29,711 & 30,075 &861 & 10,800  & \textbf{892} & 691 & \textbf{892} & $<$1 & \textbf{892} & $<$1 & \textbf{892} & $<$1 \\
C125.9\_C & 125 & 3,500 & 3,463 & \textbf{506} & 6 & \textbf{506} & $<$1 & \textbf{506} & $<$1 & \textbf{506} & $<$1 & \textbf{506} & $<$1 \\
C250.9\_C & 250 & 14,017 & 13,967 & \textbf{623} & 334 & \textbf{623} & 2 & \textbf{623} & $<$1 &  \textbf{623} & $<$1 &  \textbf{623} & $<$1 \\
gen200\_p0.9\_55\_C & 200 & 8,908 & 9,002 & \textbf{597} & 91 & \textbf{597} & $<$1 & \textbf{597}  & $<$1 & \textbf{597}  & $<$1 & \textbf{597}  & $<$1 \\
gen400\_p0.9\_75\_C & 400 & 35,823 & 35,997 & \textbf{651} & 7,858 & \textbf{651} & 50 & \textbf{651} & $<$1 &  \textbf{651} & $<$1 &  \textbf{651} & $<$1 \\
hamming8-4\_C & 256 & 10,329 & 10,535 &  \textbf{1,184} & 7,447  & \textbf{1,184} & 85 & \textbf{1,184} & $<$1 &\textbf{1,184} & $<$1 &\textbf{1,184} & $<$1 \\
keller4\_C & 171 & 4,738 & 4,697 & \textbf{1,049} & 145 & \textbf{1,049} & 3 & \textbf{1,049} & $<$1 & \textbf{1,049} & $<$1 & \textbf{1,049} & $<$1 \\
MANN\_a27\_C & 378 & 35,345 & 35,206 & \textbf{552} & 1,849 & \textbf{552} &10 & \textbf{552}& 1 &\textbf{552}& $<$1 &\textbf{552}& $<$1 \\
p\_hat300-1\_C & 300 & 5,505 & 5,428 & 2,303 & 10,800 & 2,818 & 10,800 & \textbf{2,897} & $<$1 & \textbf{2,897} & $<$1 & \textbf{2,897} & $<$1 \\
p\_hat300-2\_C & 300 & 11,051 & 10,877 &  1,852 & 10,800 & \textbf{2,263} & 10,800 & \textbf{2,263} & $<$1 & \textbf{2,263} & $<$1 & \textbf{2,263} & $<$1 \\
p\_hat300-3\_C & 300 & 16,820 & 16,570 & \textbf{851} &4,709 & \textbf{851} & 69 & \textbf{851} & $<$1 &\textbf{851} & $<$1 &\textbf{851} & $<$1 \\
\hline
\multicolumn{14}{l}{$p=0.50$, Setting 2} \\
\hline
brock200\_2\_A & 200 & 4,916 & 4,960 & \textbf{2,034} & 10,800 & \textbf{2,034} & 677 & \textbf{2,034} & $<$1 & \textbf{2,034} & $<$1 & \textbf{2,034} & $<$1 \\
brock400\_2\_A & 400 & 29,711 & 30,075 & 1,628  &10,800  & \textbf{1,630} & 1,323 & \textbf{1,630} & 11 & \textbf{1,630} & $<$1 & \textbf{1,630} & $<$1 \\
C125.9\_A & 125 & 3,500 & 3,463 & \textbf{1,152} & 7 & \textbf{1,152} & $<$1 & \textbf{1,152} & $<$1 & \textbf{1,152} & $<$1 & \textbf{1,152} & $<$1 \\
C250.9\_A & 250 & 14,017 & 13,967 &\textbf{1,236}   & 343 & \textbf{1,236} & 3&\textbf{1,236} & 3 & \textbf{1,236} & $<$1 & \textbf{1,236} & $<$1 \\
gen200\_p0.9\_55\_A & 200 & 8,908 & 9,002  & \textbf{1,151} & 113 & \textbf{1,151} & $<$1 &\textbf{1,151} & 3 & \textbf{1,151} & $<$1 & \textbf{1,151} & $<$1 \\
gen400\_p0.9\_75\_A & 400 & 35,823 & 35,997  & \textbf{1,335} &8,467  & \textbf{1,335} & 69 & \textbf{1,335} &  2 & \textbf{1,335} &  $<$1 & \textbf{1,335} &  $<$1 \\
hamming8-4\_A & 256 & 10,329 & 10,535 & \textbf{2,155}  & 10,800 & \textbf{2,155} &166 & \textbf{2,155} & 1 & \textbf{2,155} & $<$1 & \textbf{2,155} & $<$1 \\
keller4\_A & 171 & 4,738 & 4,697 & \textbf{1,759} & 353  & \textbf{1,759} & 5 & \textbf{1,759} & $<$1 & \textbf{1,759} & $<$1 & \textbf{1,759} & $<$1 \\
MANN\_a27\_A & 378 & 35,345 & 35,206 &  \textbf{1,226} & 1,950 & \textbf{1,226} & 10 & \textbf{1,226} &12 & \textbf{1,226} &$<$1 & \textbf{1,226} &$<$1  \\
p\_hat300-1\_A & 300 & 5,505 & 5,428 & 4,740 & 10,800 & \textbf{5,637} & 10,800 & \textbf{5,637} & $<$1 & \textbf{5,637} & $<$1 & \textbf{5,637} & $<$1 \\
p\_hat300-2\_A & 300 & 11,051 & 10,877 & 3,342 & 10,800 & 3,750 & 10,800 & \textbf{3,943} & $<$1 & \textbf{3,943} & $<$1 & \textbf{3,943} & $<$1 \\
p\_hat300-3\_A & 300 & 16,820 & 16,570 &  \textbf{1,658} & 7,243 & \textbf{1,658} &157 & \textbf{1,658} &$<$1 & \textbf{1,658} &$<$1 & \textbf{1,658} &$<$1 \\

brock200\_2\_B & 200 & 4,916 & 4,960 & \textbf{984}  & 10,800 & \textbf{984} & 757 & \textbf{984} & $<$1 & \textbf{984} & $<$1 & \textbf{984} & $<$1 \\

brock400\_2\_B & 400 & 29,711 & 30,075 &  765 & 10,800 & \textbf{780} &1,385 & \textbf{780} & 10 & \textbf{780} & $<$1 & \textbf{780} & $<$1 \\
C125.9\_B & 125 & 3,500 & 3,463 & \textbf{552} & 7 & \textbf{552} & $<$1& \textbf{552} & $<$1 & \textbf{552} & $<$1 & \textbf{552} & $<$1 \\
C250.9\_B & 250 & 14,017 & 13,967 & \textbf{586} & 391 & \textbf{586} & 3 & \textbf{586} & 4 &  \textbf{586} & $<$1 &  \textbf{586} & $<$1 \\
gen200\_p0.9\_55\_B & 200 & 8,908 & 9,002 & \textbf{551} & 121 & \textbf{551} & $<$1 & \textbf{551}  & 1 & \textbf{551}  & $<$1 & \textbf{551}  & $<$1 \\
gen400\_p0.9\_75\_B & 400 & 35,823 & 35,997 & \textbf{635} & 9,395 & \textbf{635} & 72 & \textbf{635} & $<$1 &  \textbf{635} & $<$1 &  \textbf{635} & $<$1 \\
hamming8-4\_B & 256 & 10,329 & 10,535 &\textbf{1,055}  &10,800  & \textbf{1,055} & 163 & \textbf{1,055} & $<$1 &\textbf{1,055} & $<$1 &\textbf{1,055} & $<$1 \\
keller4\_B & 171 & 4,738 & 4,697 & \textbf{859} & 367 & \textbf{859} & 5 & \textbf{859} & $<$1 & \textbf{859} & $<$1 & \textbf{859} & $<$1 \\
MANN\_a27\_B & 378 & 35,345 & 35,206 & \textbf{576} &2,161  & \textbf{576} & 11 & \textbf{576}& 13 &\textbf{576}& $<$1 &\textbf{576}& $<$1 \\
p\_hat300-1\_B & 300 & 5,505 & 5,428 & 2,340 & 10,800 & 2,325 & 10,800 & \textbf{2,787} & $<$1 & \textbf{2,787} & $<$1 & \textbf{2,787} & $<$1 \\
p\_hat300-2\_B & 300 & 11,051 & 10,877 & 1,642 &10,800 & 1,850 & 10,800  &\textbf{1,943} & $<$1 &\textbf{1,943} & $<$1 & \textbf{1,943} & $<$1 \\
p\_hat300-3\_B & 300 & 16,820 & 16,570 & \textbf{808} & 7,790 & \textbf{808} & 161 & \textbf{808} & $<$1  &\textbf{808} & $<$1 &\textbf{808} & $<$1 \\
brock200\_2\_C & 200 & 4,916 & 4,960 & \textbf{459} & 10,800  & \textbf{459} & 934 & \textbf{459} & $<$1 & \textbf{459} & $<$1 & \textbf{459} & $<$1 \\
brock400\_2\_C & 400 & 29,711 & 30,075 & 341  &10,800 & \textbf{355} & 1,711 & \textbf{355} & 12 & \textbf{355} & $<$1 & \textbf{355} & $<$1 \\
C125.9\_C & 125 & 3,500 & 3,463 &  \textbf{252} & 8 & \textbf{252} & $<$1& \textbf{252} & $<$1 & \textbf{252} & $<$1 & \textbf{252} & $<$1 \\
C250.9\_C & 250 & 14,017 & 13,967 & \textbf{261} & 501 & \textbf{261} &3 & \textbf{261} & 3 &  \textbf{261} & $<$1 &  \textbf{261} & $<$1 \\
gen200\_p0.9\_55\_C & 200 & 8,908 & 9,002 & \textbf{251} &146  & \textbf{251} & 1 & \textbf{251}  & 2 & \textbf{251}  & $<$1 & \textbf{251}  & $<$1 \\
gen400\_p0.9\_75\_C & 400 & 35,823 & 35,997 & \textbf{285} &10,800  & \textbf{285} & 83 &\textbf{285} & 2 &  \textbf{285} & $<$1 &  \textbf{285} & $<$1 \\
hamming8-4\_C & 256 & 10,329 & 10,535 & \textbf{505} & 10,800 & \textbf{505} & 170 & \textbf{505} & $<$1 &\textbf{505} & $<$1 &\textbf{505} & $<$1 \\
keller4\_C & 171 & 4,738 & 4,697 &  \textbf{409}& 389 & \textbf{409} & 5 & \textbf{409} & $<$1 & \textbf{409} & $<$1 & \textbf{409} & $<$1 \\
MANN\_a27\_C & 378 & 35,345 & 35,206 & \textbf{251} & 2,666 & \textbf{251} & 13 & \textbf{251}& 10 &\textbf{251}& $<$1 &\textbf{251}& $<$1 \\
p\_hat300-1\_C & 300 & 5,505 & 5,428 & 1,138 & 10,800 & 1,125 & 10,800 & \textbf{1,362} & $<$1 & \textbf{1,362} & $<$1 & \textbf{1,362} & $<$1  \\
p\_hat300-2\_C & 300 & 11,051 & 10,877 & 788 & 10,800 & 900 & 10,800 & \textbf{943} & $<$1 & \textbf{943} & $<$1 & \textbf{943} & $<$1 \\
p\_hat300-3\_C & 300 & 16,820 & 16,570 & \textbf{383} &9,014  & \textbf{383} & 178 & \textbf{383} & $<$1 &\textbf{383} & $<$1 &\textbf{383} & $<$1 \\




\hline 
\multicolumn{14}{l}{$p=0.75$, Setting 1} \\
\hline
brock200\_2\_A & 200 & 2,438 & 7,438 &1,615 & 10,800 & 1,851 & 10,800 & \textbf{1,885} & $<$1 & \textbf{1,885} & $<$1 & \textbf{1,885} & $<$1 \\
brock400\_2\_A & 400 & 14,751 & 45,035 & 1,403& 10,800 & 1,667 & 10,800 & \textbf{1,728} & 5 & \textbf{1,728} & $<$1 & \textbf{1,728} & $<$1 \\
C125.9\_A & 125 & 1,733 & 5,230 & \textbf{1,023}& 247  & \textbf{1,023} & 4 & \textbf{1,023} & $<$1 & \textbf{1,023} & $<$1 & \textbf{1,023} & $<$1 \\
C250.9\_A & 250 & 7,073 & 20,911 & 1,200&10,800  & \textbf{1,236} & 5,143& \textbf{1,236} & 2 & \textbf{1,236} & $<$1 & \textbf{1,236} & $<$1 \\
gen200\_p0.9\_55\_A & 200 & 4,425 & 13,485 & 1,186 &10,800  & \textbf{1,206} & 601 & \textbf{1,206} & $<$1 & \textbf{1,206} & $<$1 & \textbf{1,206} & $<$1 \\
gen400\_p0.9\_75\_A & 400 & 18,063 & 53,757 &1,155 & 10,800 & \textbf{1,490} & 10,800 & \textbf{1,490} & $<$1 & \textbf{1,490} & $<$1  & \textbf{1490} & $<$1 \\
hamming8-4\_A & 256 & 5,173 & 15,691 &1,549 &10,800  & 1,723 & 10,800 & \textbf{1,759} & $<$1 & \textbf{1,759} & $<$1 & \textbf{1,759} & $<$1 \\
keller4\_A & 171 & 2,400 & 7,035 &1,400 &10,800  & \textbf{1,434} & 1573 & \textbf{1,434} & $<$1 & \textbf{1,434} & $<$1 & \textbf{1,434} & $<$1 \\
MANN\_a27\_A & 378 & 17,580 & 52,971 &1,313 & 10,800 & \textbf{1,323} & 10,800 & \textbf{1,323} & $<$1 & \textbf{1,323} & $<$1 & \textbf{1,323} & $<$1 \\
p\_hat300-1\_A & 300 & 2,734 & 8,199 &2,504 & 10,800 & 3,932 & 10,800 & \textbf{4,164} & $<$1 & \textbf{4,164} & $<$1 & \textbf{4,164} & $<$1 \\
p\_hat300-2\_A & 300 & 5,603 & 16,325 &2,120 &10,800  & 2,902 & 10,800 & \textbf{2,990} & 1 & \textbf{2,990} & $<$1 & \textbf{2,990} & $<$1 \\
p\_hat300-3\_A & 300 & 8,388 & 25,002 &1,350 & 10,800 & 1561 & 10,800 & \textbf{1,564} & $<$1 & \textbf{1,564} & $<$1 & \textbf{1,564} & $<$1 \\
brock200\_2\_B & 200 & 2,438 & 7,438 & 1,416& 10,800 & 1,602 & 10,800 & \textbf{1,641} & $<$1 & \textbf{1,641} & $<$1 & \textbf{1,641} & $<$1 \\
brock400\_2\_B & 400 & 14,751 & 45,035 & 1,129& 10,800 & 1,341 & 10,800 & \textbf{1,386} & 2 & \textbf{1,386} & $<$1 & \textbf{1,386} & $<$1 \\
C125.9\_B & 125 & 1,733 & 5,230 & \textbf{856}& 395 & \textbf{856} & 8 & \textbf{856} & $<$1 & \textbf{856} & $<$1 & \textbf{856} & $<$1 \\
C250.9\_B & 250 & 7,073 & 20,911 & 960& 10,800  & \textbf{1,001} & 10,800 & \textbf{1,001} & 1 & \textbf{1,001} & $<$1 & \textbf{1,001} & $<$1 \\
gen200\_p0.9\_55\_B & 200 & 4,425 & 13,485 &961 & 10,800 & \textbf{983} & 1585 & \textbf{983} & $<$1 & \textbf{983} & $<$1 & \textbf{983} & $<$1 \\
gen400\_p0.9\_75\_B & 400 & 18,063 & 53,757 & 920& 10,800 & 1,105 & 10,800 & \textbf{1,120} & $<$1 & \textbf{1,120} & $<$1 & \textbf{1,120} & $<$1\\
hamming8-4\_B & 256 & 5,173 & 15,691 &1,356 & 10,800 & 1,486 & 10,800 & \textbf{1,579} & $<$1 & \textbf{1,579} & $<$1 & \textbf{1,579} & $<$1 \\
keller4\_B & 171 & 2,400 & 7,035 & 1,230& 10,800 & \textbf{1,268} & 4468 & \textbf{1,268} & $<$1 & \textbf{1,268} & $<$1 & \textbf{1,268} & $<$1 \\
MANN\_a27\_B & 378 & 17,580 & 52,971 & 922&10,800  & 1,011 & 10,800 & \textbf{1,021} & $<$1 & \textbf{1,021} & $<$1 & \textbf{1,021} & $<$1 \\
p\_hat300-1\_B & 300 & 2,734 & 8,199 & 2,306& 10,800 & 3,547 & 10,800 & \textbf{3,886} & $<$1 & \textbf{3,886} & $<$1 & \textbf{3,886} & $<$1 \\
p\_hat300-2\_B & 300 & 5,603 & 16,325 & 1,762& 10,800 & 2,588 & 10,800 & \textbf{2,782} & $<$1 & \textbf{2,782} & $<$1 & \textbf{2,782} & $<$1 \\
p\_hat300-3\_B & 300 & 8,388 & 25,002 & 1,144& 10,800 & 1,209 & 10,800 & \textbf{1,299} & $<$1 & \textbf{1,299} & $<$1 & \textbf{1,299} & $<$1 \\

brock200\_2\_C & 200 & 2,438 & 7,438 & 1,158& 10,800 & 1,303 & 10,800 & \textbf{1,321} & $<$1 & \textbf{1,321} & $<$1 & \textbf{1,321} & $<$1 \\
brock400\_2\_C & 400 & 14,751 & 45,035 & 793 & 10,800 & 961 & 10,800 & \textbf{1,033} & $<$1 & \textbf{1,033} & $<$1 & \textbf{1,033} & $<$1 \\
C125.9\_C & 125 & 1,733 & 5,230 & \textbf{644}& 10,800 & \textbf{644} & 45 & \textbf{644} & $<$1& \textbf{644} & $<$1 & \textbf{644} & $<$1 \\
C250.9\_C & 250 & 7,073 & 20,911 & 688& 10,800  & \textbf{734} & 10,800 & \textbf{734} & 1 & \textbf{734} & $<$1 & \textbf{734} & $<$1 \\
gen200\_p0.9\_55\_C & 200 & 4,425 & 13,485 & 707& 10,800 & \textbf{727} &  9,879 & \textbf{727} & $<$1 & \textbf{727} & $<$1 & \textbf{727} & $<$1 \\
gen400\_p0.9\_75\_C & 400 & 18,063 & 53,757 &613  & 10,800 & 756 & 10,800 & \textbf{772} & 1 & \textbf{772} & $<$1 & \textbf{772} & $<$1 \\
hamming8-4\_C & 256 & 5,173 & 15,691 & 1,120& 10,800 & 1,248 & 10,800 & \textbf{1,378} & $<$1 & \textbf{1,378} & $<$1 & \textbf{1,378} & $<$1 \\
keller4\_C & 171 & 2,400 & 7,035 & 1,060&10,800  & \textbf{1,109} & 10,800 & \textbf{1,109} & $<$1 & \textbf{1,109} & $<$1 & \textbf{1,109} & $<$1 \\
MANN\_a27\_C & 378 & 17,580 & 52,971 &583& 10,800 & 648 & 10,800 & \textbf{651} & 2 & \textbf{651} & $<$1 & \textbf{651} & $<$1 \\
p\_hat300-1\_C & 300 & 2,734 & 8,199 &1,866 &10,800  & 3,191 & 10,800 & \textbf{3,480} & $<$1 & \textbf{3,480} & $<$1 & \textbf{3,480} & $<$1 \\
p\_hat300-2\_C & 300 & 5,603 & 16,325 &1,511& 10,800 & 2,249 & 10,800 &\textbf{2,473} & $<$1 & \textbf{2,473} & $<$1 & \textbf{2,473} & $<$1 \\
p\_hat300-3\_C & 300 & 8,388 & 25,002 & 886& 10,800 & 919 & 10,800 & \textbf{1,004} & $<$1 & \textbf{1,004} & $<$1 & \textbf{1,004} & $<$1 \\
\hline
\multicolumn{14}{l}{$p=0.75$, Setting 2} \\
\hline
brock200\_2\_A & 200 & 2,438 & 7,438 &2,999 &10,800  & 3,078 & 10,800 & \textbf{3,326} & 8 & \textbf{3,326} & $<$1 & \textbf{3,326} & $<$1 \\
brock400\_2\_A & 400 & 14,751 & 45,035 &2,603 &10,800 & 2,787 & 10,800 & \textbf{2,941} & 12 & \textbf{2,941} & 1& \textbf{2,941} &3  \\
C125.9\_A & 125 & 1,733 & 5,230 & \textbf{1,837}&365  & \textbf{1,837} & 62.71 &\textbf{1,837} & $<$1 & \textbf{1,837} & $<$1 & \textbf{1,837} & $<$1 \\
C250.9\_A & 250 & 7,073 & 20,911 &2,011 & 10,800 & \textbf{2,171} & 10,800 & \textbf{2,171} & 5 & \textbf{2,171} & $<$1 & \textbf{2,171} & $<$1 \\
gen200\_p0.9\_55\_A & 200 & 4,425 & 13,485 &2,094 &10,800  & \textbf{2,096} &7,832 & \textbf{2,096} & $<$1 & \textbf{2,096} & $<$1  & \textbf{2,096} & $<$1 \\
gen400\_p0.9\_75\_A & 400 & 18,063 & 53,757 & 2,164&  10,800& 2,310 & 10,800 &\textbf{2,404} & 19 & \textbf{2,404} & $<$1 & \textbf{2,404} & 1 \\
hamming8-4\_A & 256 & 5,173 & 15,691 & 2,813& 10,800 & \textbf{3,124} & 10,800 & \textbf{3,124} & 1 & \textbf{3,124} & $<$1 & \textbf{3,124} & $<$1 \\
keller4\_A & 171 & 2,400 & 7,035 &2,532 & 10,800 & \textbf{2,690} & 10,800 & \textbf{2,690} & $<$1 & \textbf{2,690} & $<$1 & \textbf{2,690} & $<$1 \\
MANN\_a27\_A & 378 & 17,580 & 52,971 &1,976 &10,800  & 2,206 & 10,800 & \textbf{2,208} & $<$1 & \textbf{2,208} & $<$1 & \textbf{2,208} & $<$1 \\
p\_hat300-1\_A & 300 & 2,734 & 8,199 &5,566 &10,800  & 6,666 & 10,800 & \textbf{7,899} & 25 & \textbf{7,899} & $<$1 & \textbf{7,899} & $<$1 \\
p\_hat300-2\_A & 300 & 5,603 & 16,325 &4,102 &10,800  & 4,714 & 10,800 & \textbf{5,343} & $<$1 & \textbf{5,343} & $<$1 & \textbf{5,343} & $<$1 \\
p\_hat300-3\_A & 300 & 8,388 & 25,002 &2,477 &10,800  & 2,658 & 10,800 & \textbf{2,838} & 2 & \textbf{2,838} & $<$1 & \textbf{2,838} & $<$1 \\
brock200\_2\_B & 200 & 2,438 & 7,438 &1,398 &10,800  & 1,427 & 10,800 & \textbf{1,533} & $<$1 & \textbf{1,533} & $<$1 & \textbf{1,533} & $<$1 \\

brock400\_2\_B & 400 & 14,751 & 45,035 &1,161 &10,800  & 1,210 & 10,800 & \textbf{1,291} & 16 & \textbf{1,291} & 1  & \textbf{1,291} & 4 \\
C125.9\_B & 125 & 1,733 & 5,230&\textbf{837} &523  & \textbf{837} & 84 & \textbf{837} &$<$1 & \textbf{837} &$<$1 & \textbf{837} &$<$1 \\
C250.9\_B & 250 & 7,073 & 20,911 & 911& 10,800 & 969 & 10,800 &\textbf{971} & 3 & \textbf{971} & $<$1 & \textbf{971} & $<$1 \\
gen200\_p0.9\_55\_B & 200 & 4,425 & 13,485 &944 &10,800  & \textbf{946} & 10,800 & \textbf{946} & $<$1 & \textbf{946} & $<$1 & \textbf{946} & $<$1 \\
gen400\_p0.9\_75\_B & 400 & 18,063 & 53,757 &936 &10,800  & 1,010 & 10,800 & 1,052&11 & \textbf{1,054} & 1 & \textbf{1,054} & 1 \\
hamming8-4\_B & 256 & 5,173 & 15,691 &1,309 &10,800  & \textbf{1,474} & 10,800 & \textbf{1,474} & $<$1 & \textbf{1,474} & $<$1 & \textbf{1,474} & $<$1 \\
keller4\_B & 171 & 2,400 & 7,035 &1,190 & 10,800 & 1,240 & 10,800 & \textbf{1,254} & $<$1 & \textbf{1,254} & $<$1 & \textbf{1,254} & $<$1 \\
MANN\_a27\_B & 378 & 17,580 & 52,971 & 875& 10,800 & 956 & 10,800 & \textbf{958} & 11 & \textbf{958} & $<$1 & \textbf{958} & $<$1 \\
p\_hat300-1\_B & 300 & 2,734 & 8,199 & 2,666&10,800  & 3,216 & 10,800 & \textbf{3,818} & 2 & \textbf{3,818} & $<$1 & \textbf{3,818} & $<$1 \\
p\_hat300-2\_B & 300 & 5,603 & 16,325 &1,952 &10,800 & 2,177 & 10,800 & \textbf{2,543} & $<$1 & \textbf{2,543} & $<$1 & \textbf{2,543} & $<$1 \\
p\_hat300-3\_B & 300 & 8,388 & 25,002 & 1,101&10,800  & 1,208 & 10,800 & \textbf{1,288} & 2 & \textbf{1,288} & $<$1 & \textbf{1,288} & $<$1 \\
brock200\_2\_C & 200 & 2,438 & 7,438 &607 &10,800  & 607 & 10,800 & \textbf{658} & $<$1 & \textbf{658} & $<$1 & \textbf{658} & $<$1 \\
brock400\_2\_C & 400 & 14,751 & 45,035 &461 &10,800  & 465 & 10,800 & \textbf{500} & 3 & \textbf{500} & $<$1 & \textbf{500} & $<$1 \\
C125.9\_C & 125 & 1,733 & 5,230 & \textbf{337}& 1,519 & \textbf{337} & 204 & \textbf{337} & $<$1 & \textbf{337} & $<$1 & \textbf{337} & $<$1 \\
C250.9\_C & 250 & 7,073 & 20,911 & 359& 10,800 & 370 & 10,800 & \textbf{372} & 5& \textbf{372} & $<$1  & \textbf{372} & $<$1  \\
gen200\_p0.9\_55\_C & 200 & 4,425 & 13,485 &352 & 10,800 & \textbf{371} & 10,800 & \textbf{371} & $<$1 & \textbf{371} & $<$1 & \textbf{371} & $<$1 \\
gen400\_p0.9\_75\_C & 400 & 18,063 & 53,757 & 356& 10,800 & 379 & 10,800 & \textbf{395} & 16 & \textbf{395} & $<$1 & \textbf{395} & $<$1 \\
hamming8-4\_C & 256 & 5,173 & 15,691 & 556&10,800  & 649 & 10,800 & \textbf{652} & $<$1 & \textbf{652} & $<$1 & \textbf{652} & $<$1 \\
keller4\_C & 171 & 2,400 & 7,035 & 510&10,800  & 554 & 10,800 & \textbf{558} & $<$1 & \textbf{558} & $<$1 & \textbf{558} & $<$1 \\
MANN\_a27\_C & 378 & 17,580 & 52,971 & 325& 10,800 & \textbf{335} & 10,800 &\textbf{335} & 11 & \textbf{335} & $<$1 & \textbf{335} & $<$1 \\
p\_hat300-1\_C & 300 & 2,734 & 8,199 & 1,205& 10,800 & 1,491 & 10,800 & \textbf{1,793} & 1 & \textbf{1,793} & $<$1 & \textbf{1,793} & $<$1 \\
p\_hat300-2\_C & 300 & 5,603 & 16,325 & 854& 10,800 & 981 & 10,800 &\textbf{1,159} & $<$1 & \textbf{1,159} & $<$1 & \textbf{1,159} & $<$1 \\
p\_hat300-3\_C & 300 & 8,388 & 25,002 & 453& 10,800 & 494 & 10,800 & \textbf{529} & 1 & \textbf{529} & $<$1 & \textbf{529} & $<$1 \\

\end{longtable}}





\end{document}